\documentclass[11pt,draftcls,onecolumn]{IEEEtran}
\usepackage{mathptmx} 
\usepackage{latexsym,bm}
\usepackage{amsmath}
\usepackage{amsfonts}
 \usepackage{amssymb,amsmath}
 \usepackage{graphicx}
\usepackage{epstopdf}
\usepackage{epsfig}       
 \usepackage{amsfonts}
\usepackage{indentfirst}
\usepackage{graphicx}      
\usepackage{epsfig} 
\usepackage{epstopdf}
\usepackage{times}
\usepackage{bbm}
\usepackage[numbers,sort&compress]{natbib}
\usepackage{cleveref}
\usepackage{color}
\usepackage[justification=centering]{caption}
\usepackage{subfigure}

\newcommand*{\QEDA}{\hfill\ensuremath{\blacksquare}}  

\allowdisplaybreaks
 \newtheorem{theorem}{Theorem}[section]
\newtheorem{proof}{Proof}
\newtheorem{corollary}{Corollary}[section]
\newtheorem{lemma}{Lemma}[section]
\newtheorem{remark}{Remark}[section]

\newtheorem{assumption}{Assumption}[section]

\ifCLASSINFOpdf
\else
\fi
\begin{document}

\def\ba{\begin{array}}
\def\ea{\end{array}}
\def\ban{\begin{eqnarray*}}
\def\ean{\end{eqnarray*}}
\def\bd{\begin{description}}
\def\ed{\end{description}}
\def\be{\begin{equation}}
\def\ee{\end{equation}}
\def\bna{\begin{eqnarray}}
\def\ena{\end{eqnarray}}
\allowdisplaybreaks
\title{Distributed Stochastic Optimization With Unbounded Subgradients Over Randomly Time-Varying Networks}
%
%
%

\author{Yan~Chen,
       Alexander L. Fradkov,~\IEEEmembership{Life Fellow,~IEEE},
       Keli~Fu,
       Xiaozheng~Fu and Tao~Li,~\IEEEmembership{Senior Member,~IEEE}
\thanks{Tao Li's work was funded by the National Natural Science Foundation of
China under Grant No. 62261136550. Alexander L. Fradkov's work was funded by Russian Science Foundation under Grant No. 23-41-00060. The authors are listed in alphabetical order. Corresponding author: Tao Li. }
\thanks{Y. Chen, K. Fu and X. Fu and T. Li are with the Shanghai Key Laboratory of Pure Mathematics and Mathematical Practice, School of Mathematical Sciences, East China Normal University, Shanghai 200241, China (e-mail: tli@math.ecnu.edu.cn).}
\thanks{Alexander L. Fradkov is with the Institute for
Problems in Mechanical Engineering of the Russian Academy of Sciences
(IPME RAS), 199178, Saint Petersburg, Russia.}
}

%
%

\markboth{Journal of \LaTeX\ Class Files,~Vol.~14, No.~8, August~2015}%
{Shell \MakeLowercase{\textit{et al.}}: Bare Demo of IEEEtran.cls for IEEE Journals}
%



\maketitle

\begin{abstract}
Motivated by distributed statistical learning over uncertain communication networks, we study distributed stochastic optimization by networked nodes to cooperatively minimize a sum of convex cost functions. The network is modeled by a sequence of time-varying random digraphs with each node representing a local optimizer and each edge representing a communication link. We consider the distributed subgradient optimization algorithm with noisy measurements of local cost functions' subgradients, additive and multiplicative noises among information exchanging between each pair of nodes. By stochastic Lyapunov method, convex analysis, algebraic graph theory and martingale convergence theory, we prove that if the local subgradient functions grow linearly and the sequence of digraphs is conditionally balanced and uniformly conditionally jointly connected,  then proper algorithm step sizes can be designed so that all nodes' states converge to the global optimal solution almost surely.
\end{abstract}


%
\IEEEpeerreviewmaketitle

\section{Introduction}
\subsection{Related work}
\hspace*{12pt} In recent years, distributed cooperative optimization over networks has attracted extensive attentions,
such as the economic dispatch in power grids (\cite{YiP}) and  the traffic flow control in intelligent transportation networks (\cite{MohebifardR}),  et al. Considering the various uncertainties in practical network environments, distributed stochastic optimization algorithms have been widely studied. The (sub)gradients of local cost functions are used in many distributed optimization algorithms.
However, it is difficult to get accurate (sub)gradients in many practical applications. For example, in distributed statistical machine learning (\cite{LiuX}), the local loss functions are the mathematical expectations of random functions so that the local optimizers can only obtain the measurement of the (sub)gradients with random noises. The influence of (sub)gradient measurement noises has been considered for distributed optimization algorithms in \cite{ThinhT}-\cite{AlghunaimSA}.
In real networked systems, the information exchange among nodes is often affected by communication noises, and the structure of the network often changes randomly due to packet dropouts, link/node failures and recreations, which are studied in \cite{BastianelloN2}-\cite{WangD}.

Most of the above works consider the  randomly switching   networks, (sub)gradient measurement and communication link noises separately.
However, a variety of random factors may co-exist in practical environment.
For distributed statistical machine learning algorithms, the (sub)gradients of local loss functions cannot be obtained accurately, the graphs may change randomly and the communication links may be noisy. There are many excellent results on the distributed optimization with multiple uncertain factors (\cite{HongM}-\cite{Srivastava}).
Both (sub)gradient noises and random graphs are considered in \cite{HongM}-\cite{SahuAK}. In \cite{HongM}, the local gradient noises are independent with bounded second-order moments and the graph sequence is i.i.d.
In \cite{YiP2}-\cite{LeiJ}, the (sub)gradient measurement noises are martingale difference sequences and their second-order conditional moments depend on the states of the local optimizers. The random graph sequences in \cite{YiP2}-\cite{Srivastava}  are i.i.d. with connected and undirected mean graphs. In addition, additive communication noises are considered in \cite{LeiJ}-\cite{Srivastava}.

In addition to uncertainties in information exchange, different assumptions on the cost functions have been discussed.
 In the most of existing works on the distributed convex optimization, it is assumed that the subgradients are bounded if the local cost
functions are not differentiable (\cite{Srivastava}-\cite{LiuS}) and the subgradients are Lipschitz continuous
   only for the case with  differentiable local cost functions (\cite{SahuAK}-\cite{LeiJ}, \cite{ShiW}-\cite{MokhtariA}).

\subsection{Main contribution}

Though the above works have made the deep research on distributed stochastic optimization, the practical cases may be more complex.
For example, for the LASSO (Least Absolute Shrinkage and Selection Operator) regression problem,
the local cost functions are not differentiable and  subgradients are not bounded as in \cite{NedicA}-\cite{LiuS} .
Besides, the network graphs may change randomly with spatial and temporal dependency (i.e. Both the weights of different edges in the network graphs at the same time instant and the network graphs at  different time instants may be mutually dependent.) rather than i.i.d. graph sequences as in \cite{YiP2}-\cite{Srivastava},
and additive and multiplicative communication noises may co-exist in communication links (\cite{JingWang}).
In summary, in quite a number of important problems, the assumptions required in the existing works are  not satisfied.

Motivated by distributed statistical learning over uncertain communication networks, we study the distributed stochastic convex optimization by networked local optimizers to cooperatively minimize a sum of local convex cost functions. The network is modeled by a sequence of time-varying random digraphs which may be spatially and temporally dependent. The local cost functions are not required to be differentiable, nor do their subgradients need to be bounded. The local optimizers can only obtain the measurement information of the local subgradients with random noises. The additive and multiplicative communication noises co-exist in communication links. We consider the distributed stochastic subgradient optimization algorithm and prove that if the sequence of random digraphs is conditionally balanced and uniformly conditionally jointly connected, then the states of all local optimizers converge to the same global optimal solution almost surely. The main contributions of our paper are listed as follows.

$\textrm{\uppercase\expandafter{\romannumeral1}.}$  The local cost functions in this paper are not required to be differentiable and the subgradients only satisfy the linear growth condition.
The inner product of the subgradients and the error between local optimizers' states and the global optimal solution inevitably exists in  the  recursive inequality of the conditional mean square error. This leads the nonnegative supermartingale convergence theorem  not to be applied directly
 (Lemma 3.1).
  To this end, we estimate  the upper bound of the mean square increasing rate of the local optimizers' states at first (Lemma 3.2). Then we substitute this upper bound  into the Lyapunov function difference inequality of the consensus error, and obtain the estimated convergence rate of  mean square consensus (Lemma 3.3). Further, the estimations of these rates are substituted into the recursive inequality of the conditional mean square error between the states and the global optimal solution. Finally, by properly choosing the  step sizes,  we prove that the states of all local optimizers converge to the same global optimal solution almost surely by the non-negative supermartingale convergence theorem. The key lies in that the algorithm step sizes should be chosen carefully to eliminate the possible increasing effect caused by the linear growth of the subgradients and to balance the rates between achieving consensus and seeking the optimal solution.

$\textrm{\uppercase\expandafter{\romannumeral2}.}$ The structure of the networks among optimizers is modeled by a more general sequence of random digraphs. The sequence of random digraphs is conditionally balanced, and the weighted adjacency matrices are not required to have special statistical properties such as independency with identical distribution,  Markovian switching, or stationarity, etc. The edge weights are also not required to be nonnegative at every time instant. By introducing the concept of conditional digraphs and developing the stochastic Lyapunov method for distributed optimization over non-stationary randomly time-varying networks,  uniformly conditionally joint connectivity condition is established to ensure the convergence of the distributed stochastic optimization algorithms.
The joint connectivity condition for Markovian and deterministic switching graphs, and the connectivity condition on the mean graph for i.i.d. graphs are all special cases of our condition.

$\textrm{\uppercase\expandafter{\romannumeral3}.}$ The co-existence of random graphs, subgradient measurement noises, additive and multiplicative communication noises are considered.  Compared with the case with only a single random factor, the coupling terms of different random factors inevitably affect the mean square difference between optimizers' states and any given vector. What's more, multiplicative noises relying on the relative states between adjacent local optimizers make states, graphs and noises coupled together. It becomes more complex to estimate the mean square upper bound of the local optimizers' states  (Lemma 3.1). We firstly employ the property of conditional independence to deal with the coupling term of different random factors. Then, we prove that the mean square upper bound of the coupling term of states, network graphs and noises depends on the second-order   moment of the difference between optimizers' states and the given vector. Finally, we get an estimate of the mean square increasing rate of the local optimizers' states in terms of the step sizes of the algorithm (Lemma 3.2).



\subsection{Notations and symbols}
$\mathbf{1}_{N}$: $N$ dimensional vector with all ones; $\mathbf{0}_{N}$: $N$ dimensional vector with all zeros;
  $I_{N}$: $N$ dimensional identity matrix; $O_{m \times n}$: $m \times n$ dimensional zero matrix; $A \geq B$: the matrix $A-B$ is positive semi-definite; $A \succeq B$:   the elements of  $A-B$ are all nonnegative; $A\otimes B$: the Kronecker product of matrices $A$ and $B$;  $\operatorname{Tr} (A)$: the trace of matrix $A$; $\lambda_{2}(A)$: the second smallest eigenvalue of a real symmetric matrix $A$; $\operatorname{diag}\left(B_{1}, \ldots, B_{n}\right)$: {the} block diagonal matrix with entries being $B_{1}, \ldots, B_{n}$; $\|A\|$: {the} 2-norm of matrix $A$; $\|A\|_{F}$: {the}  Frobenius-norm of matrix $A$; $E[\xi]$: the  mathematical expectation of random variable $\xi$; $|S|$: the cardinal number of set $S$;
$C_{r}$ inequality: $\left|\sum_{i=1}^n a_i\right|^r\leq\sum_{i=1}^n\left|a_i\right|^r,\  0<r<1$ and $\left|\sum_{i=1}^n a_i\right|^r\leq n^{r-1}  \sum_{i=1}^n\left|a_i\right|^r$,  $r \geq 1,$ $\ a_{i} \in \mathbb{R},\ i=1,\ldots,n$;
$d_{f}(\bar{x})$: a subgradient of the convex function $f$ at $\bar{x}$;
$\partial f(\bar{x}):$ the sub-differential set of the convex function $f$ at $\bar{x}$; for a weighted digraph  $\mathcal{G}=\{\mathcal{V},\mathcal{E}_{\mathcal{G}},\mathcal{A}_{\mathcal{G}}
=[a_{ij}]_{i,j=1}^N\}$,
the in-degree and out-degree of node $i$  are denoted by $\text{deg}_{i}^{in} =\sum^N_{j=1}a_{ij}$ and $\text{deg}_{i}^{out} =\sum^N_{j=1}a_{ji}$, respectively.  If $\text{deg}_{i}^{in}=\text{deg}_{i}^{out}$, $\forall\ i\in \mathcal{V}$, then $\mathcal{G}$ is balanced.

\section{Problem formulation}
\label{sec:problem}

Consider a network with $N$ nodes. Each node represents a local optimizer. The objective of the network is to solve the optimization problem
\begin{equation}\label{model}
\min\limits_{x\in\mathbb{R}^{n}}\ f(x)\triangleq\sum\limits_{i=1}^Nf_{i}(x),
\end{equation}
where each local cost function $f_{i}(\cdot)$: $\mathbb{R}^{n}\rightarrow \mathbb{R}$ is convex and is only known to optimizer $i$.
For the problem \eqref{model},
denote  the optimal value by $f^{*}= \min_{x\in\mathbb{R}^{n}} f(x)$ and the set of optimal solutions by $\mathcal{X}^{*}=\left\{x \in \mathbb{R}^{n}: f(x)=f^{*}\right\}$.

The information structure of the network is described by a sequence of random digraphs $\{\mathcal{G}(k)=\{\mathcal{V},\mathcal{E}_{\mathcal{G}(k)},\mathcal{A}_{\mathcal{G}(k)}\},k\geq0\}$, where $\mathcal{V}=\{1,\ldots,N\}$ is the set of nodes, $\mathcal{E}_{\mathcal{G}(k)}$ is the set of edges at time instant $k$, and $(j,i)\in\mathcal{E}_{\mathcal{G}(k)}$ if and only if the $j$th optimizer can send information to the $i$th optimizer directly.
The neighbourhood of the $i$th optimizer at time instant $k$ is denoted by $\mathcal{N}_i(k)=\{j\in \mathcal{V}|(j,i)\in\mathcal{E}_{\mathcal{G}(k)}\}$. $\mathcal{A}_{\mathcal{G}(k)}=[a_{ij}(k)]_{i,j=1}^N$ is \emph{the generalized weighted adjacency matrix} at time instant $k$, where $a_{ii}(k)=0$, and $a_{ij}(k)\neq0\Leftrightarrow j\in\mathcal{N}_i(k)$, representing the weight on channel $(j, i)$ at time instant $k$.
The generalized Laplacian matrix of the digraph $\mathcal{G}(k)$ is denoted by $\mathcal{L}_{\mathcal{G}(k)}=\mathcal{D}_{\mathcal{G}(k)}-\mathcal{A}_{\mathcal{G}(k)}$, where $\mathcal{D}_{\mathcal{G}(k)}=$ $\operatorname{diag}(\text{deg}_{1}^{in}(k),\ldots,\text{deg}_{N}^{in}(k))$.
Let $\widetilde{\mathcal{G}}(k):=\{\mathcal{V},\mathcal{E}_{\widetilde{\mathcal{G}}(k)},\mathcal{A}_{\widetilde{\mathcal{G}}(k)}\}$ be the reversed digraph of $\mathcal{G}(k)$, where $(i,j)\in\mathcal{E}_{\widetilde{\mathcal{G}}(k)}$ if and only if $(j,i)\in\mathcal{E}_{\mathcal{G}(k)}$ and $\mathcal{A}_{\widetilde{\mathcal{G}}(k)}=\mathcal{A}^T_{\mathcal{G}(k)}$.
Let $\widehat{\mathcal{G}}(k):=\Big\{\mathcal{V},\mathcal{E}_{\mathcal{G}(k)}\cup\mathcal{E}_{\widetilde{\mathcal{G}}(k)},
\frac{\mathcal{L}_{\mathcal{G}(k)}+\mathcal{L}^T_{\mathcal{G}(k)}}{2}\Big\}$ be the symmetrized graph of $\mathcal{G}(k)$.
Denoted $\widehat{\mathcal{L}}_{\mathcal{G}(k)}=\frac{\mathcal{L}_{\mathcal{G}(k)}+\mathcal{L}^T_{\mathcal{G}(k)}}{2}$.

We consider the  distributed stochastic subgradient algorithm
\begin{equation}\label{algorithm}
x_{i}(k+1)=x_{i}(k)+c(k)\sum\limits_{j\in\mathcal{N}_{i}(k)}a_{ij}(k)(y_{ji}(k)-x_{i}(k))-\alpha(k)\tilde{d}_{f_{i}}(x_{i}(k)),\ k\geq0,\ i\in \mathcal{V},
\end{equation}
where $x_{i}(k)\in \mathbb{R}^{n}$ is the state of the $i$th optimizer at time instant $k$, representing its local estimate of the global optimal solution to the  problem \eqref{model}; $x_{i}(0)\in\mathbb{R}^{n}$, $i=1,2,\ldots, N$ are the initial values; $c(k)$ and $\alpha(k)$ are the time-varying step sizes; $y_{ji}(k)\in \mathbb{R}^{n}$  denotes the measurement of the neighbouring optimizer $j$'s state by  optimizer $i$ at time instant $k$, which is given by
\begin{equation}\label{algorithm1}
y_{ji}(k)=x_{j}(k)+\psi_{ji}(x_{j}(k)-x_{i}(k))\xi_{ji}(k),\ j\in\mathcal{N}_{i}(k),\ i\in\mathcal{V},
\end{equation}
where $\{\xi_{ji}(k),k\geq0\}$ is the sequence of communication noises in channel $(j,i)$ and $\psi_{ji}(\cdot):\mathbb{R}^{n}\rightarrow \mathbb{R}$ is the noise intensity function. Let $\tilde{d}_{f_{i}}(x_{i}(k))$  denote the noisy measurement of the subgradient $d_{f_{i}}(x_{i}(k))$ by optimizer $i$, i.e.
\begin{equation}\label{algorithm2}
\tilde{d}_{f_{i}}(x_{i}(k))=d_{f_{i}}(x_{i}(k))+\zeta_{i}(k),
\end{equation}
where $\{\zeta_{i}(k),k\geq0\}$ is the measurement noise sequence.

Denote $X(k)=[x^T_{1}(k),\ldots,x^T_{N}(k)]^T$, $\xi(k)=[\xi^T_{11}(k)$, $\ldots,$  $\xi^T_{N1}(k);\ldots;$ $\xi^T_{1N}(k),\ldots,\xi^T_{NN}(k)]^T$ and $\zeta(k)=[\zeta^T_{1}(k),\ldots, \zeta^T_{N}(k)]^T$, where $\xi_{ji}(k)\equiv\mathbf{0}_{n}$ if $j\notin \mathcal{N}_{i}(k)$ for all $k\geq0$.
\vskip 0.2cm
We have the following assumptions.

\vskip 0.2cm
\begin{assumption}
\label{subgradient}
 (Linear growth condition)  There exist nonnegative constants $\sigma_{di}$ and $C_{di}$, such that $\|d_{f_{i}}(x)\|\leq\sigma_{di}\|x\|+C_{di}, \forall \ x \in \mathbb{R}^{n},  d_{f_{i}}(x)\in\partial f_{i}(x)$, $i=1,\cdots,N$.
\end{assumption}

\vskip 0.2cm

\begin{assumption}
\label{noise1}
There exists a $\sigma$-algebra flow $\{\mathcal{F}(k), k\geq 0\}$, such that
$\{\mathcal{A}_{\mathcal{G}(k)}$, $\mathcal{F}(k)$, $k\geq0\}$ is an adapted sequence. The communication noise process $\{\xi(k),\mathcal{F}(k)$, $k\geq0\}$ is a vector-valued martingale difference and there exists a positive constant $C_{\xi}$ such that $\sup_{k\geq0}E\big[\|\xi(k)\|^2\big|$ $\mathcal{F}(k-1)\big]\leq C_{\xi}$ a.s. For any given time instant $k$, $\sigma\{\xi(k)\}$ and $\sigma\{\mathcal{A}_{\mathcal{G}(k)}$, $\mathcal{A}_{\mathcal{G}(k+1)},\ldots\}$ are conditionally independent given $\mathcal{F}(k-1)$.
\end{assumption}

\vskip 0.2cm

\begin{assumption}
\label{noise2}
For the $\sigma$-algebra flow given by \Cref{noise1}, the subgradient measurement noise sequence $\{\zeta(k),\mathcal{F}(k),$ $k\geq0\}$ is a vector-valued martingale difference. There exist nonnegative constants $\sigma_{\zeta}$ and $C_{\zeta}$ such that $E\left[\|\zeta(k)\|^2|\mathcal{F}(k-1)\right]\leq \sigma_{\zeta}\|X(k)\|^2+C_{\zeta}$ a.s.  For any given time instant $k$, $\sigma\{\zeta(k)\}$ and $\sigma\{\mathcal{A}_{\mathcal{G}(k)},\mathcal{A}_{\mathcal{G}(k+1)},\ldots\}$ are conditionally independent given $\mathcal{F}(k-1)$.
\end{assumption}

\vskip 0.2cm
Note that Koloskova et al. (\cite{Koloskova})  studied a similar problem as the problem (\ref{model}) in the sense that
 $f(x)=\frac{1}{N} \sum_{i=1}^N f_i(x)$ with $ f_i(x)= E[ F_i\left(x, \eta_i\right)]$.
Though relative weak assumptions are used in \cite{Koloskova}  compared with prior works,  they are not weaker than ours.

\vskip 0.2cm

\begin{assumption}
\label{intensity}
There exist nonnegative constants $\sigma_{ji}$ and $b_{ji},i,j\in\mathcal{V}$, such that $|\psi_{ji}(x)|\leq\sigma_{ji}\|x\|+b_{ji},\ \forall\  x\in \mathbb{R}^{n}$.
\end{assumption}

\vskip 0.2cm

\begin{remark}
\Cref{intensity} means that the measurement model \eqref{algorithm1} covers both cases with additive and multiplicative measurtment/communication noises.
Dithered quantization in distributed parameter estimation
 leads to additive communication noises (\cite{SKar}).
If the logarithmic quantization is used, then the quantized measurement of $x_{j}(k)-x_{i}(k)$ by agent $i$ is given by $x_{j}(k)-x_{i}(k)+(x_{j}(k)-x_{i}(k))\xi_{ji}(k)$, where  $\xi_{ji}(k)$ can be regarded as the white noise (\cite{Carli}).

\end{remark}

\vskip 0.2cm

\begin{assumption}
\label{countable}
The set $\mathcal{X}^{*}$   is non-empty and countable.
\end{assumption}

\vskip 0.2cm
We call $E\left[\mathcal{A}_{\mathcal{G}(k)}|\mathcal{F}(m)\right],m\leq k-1,$ the conditional generalized weighted adjacency matrix of $\mathcal{A}_{\mathcal{G}(k)}$ with respect to $\mathcal{F}(m)$, and call its associated random graph the conditional digraph of $\mathcal{G}(k)$ with respect to $\mathcal{F}(m)$, denoted by $\mathcal{G}(k|m)$, i.e. $\mathcal{G}(k|m)=\{\mathcal{V},E\left[\mathcal{A}_{\mathcal{G}(k)}|\mathcal{F}(m)\right]\}$
(\cite{LiT}).  Denote $\lambda_{k}^h:=\lambda_{2}(\sum^{k+h-1}_{i=k}
E[\widehat{\mathcal{L}}_{\mathcal{G}(i)}\\ |\mathcal{F}(k-1)])$. We
consider the sequence of balanced conditional digraphs as follows.
\vskip 0.2cm
\begin{assumption}
\label{stochasticgraph}
 The  random graph sequence $\{\mathcal{G}(k),k\geq0\}\in\Gamma_{1}$, where
$$\Gamma_{1}=\Big\{\{\mathcal{G}(k),k\geq0\}|E\left[\mathcal{A}_{\mathcal{G}(k)}
|\mathcal{F}(k-1)\right]
\succeq O_{N\times N}\ \mbox{a.s.},\ \mathcal{G}(k|k-1) \text{ is balanced}\ \mbox{a.s.},\ k\geq0\Big\}.$$
\end{assumption}

\vskip 0.2cm



We give an example satisfying the above assumptions. The local cost function $f_{i}$ is the risk function associated with the $i$th optimizer's local data, i.e.
$
f_{i}(x)=E\left[\ell_{i}(x;\mu_{i})\right]+R_{i}(x),
$
where $\ell_{i}(\cdot\ ;\cdot)$ is a loss function {which} is convex with respect to its first argument, $\mu_{i}$ is the data sample of optimizer $i$, and $R_{i}:\mathbb{R}^n\rightarrow\mathbb{R}$ is a convex regularization term  (\cite{THastie}).
An example of $L_{2}$-regularization is given in \cite{JakoveticD1}, for which \Cref{subgradient}  naturally holds.
If the quadratic loss is considered with $L_{1}$-regularization,
then it is called the LASSO regression problem
\begin{equation}\label{lasso}
\min\limits_{x\in \mathbb{R}^{n}}\sum_{i=1}^N\left(E\left[\ell_{i}(x;u_{i}(k),p_{i}(k))\right]+\kappa \|x\|_{1}\right),
\end{equation}
where
$
\ell_{i}(x;u_{i}(k),p_{i}(k))=\frac{1}{2}\|p_{i}(k)-u^T_{i}(k)x\|^2,
$
$
p_{i}(k)=u^T_{i}(k)x_{0}+\nu_{i}(k),
$
$x_{0}\in \mathbb{R}^n$ is the unknown parameter, $u_{i}(k)\in \mathbb{R}^n$ is the regression vector   and $\nu_{i}(k)$ is the local measurement noise. Random sequences $\{u_{i}(k),k\geq0\}$ and $\{\nu_{i}(k),k\geq0\}$ are mutually independent i.i.d. Gaussian sequences with distributions $N(\mathbf{0}_{n},R_{u,i})$ and  $N(0,\sigma^{2}_{\nu,i})$. It can be verified that Assumptions \ref{subgradient}-\ref{noise2} hold. See Appendix D  for  details.

\vskip 0.2cm
 For the LASSO regression problem \eqref{lasso},  one can also employ a distributed proximal gradient (DPG) algorithm
 \begin{equation}
\begin{cases}
    q_i (k)  =c(k)\sum\limits_{j\in\mathcal{N}_{i}(k)}a_{ij}(k)(m_{ji}(k)-z_{i}(k)), \notag\\
    \hat{q}_i (k) =q_i (k) -\alpha(k)\tilde{d}_{i}(q_i (k) ), \notag \\
    z_i (k+1) = \underset{z \in \mathbb{R}}{\operatorname{argmin}}\left\{\kappa |z|+\frac{1}{2 \alpha(k)}\|z-\hat{q}_i (k) \|^2\right\},\ k\geq0,\ i\in \mathcal{V},\notag
\end{cases}
\end{equation}
where $z_{i}(k)\in \mathbb{R} $ is the state  of the $i$th optimizer at time instant $k$, representing its local estimate  of the global optimal solution to  (\ref{lasso}),  $m_{ji}(k)=z_{j}(k)+(0.1(z_{j}(k)-z_{i}(k))+0.1)\xi_{ji}(k)$ and $\tilde{d}_{i}(q_i (k) )=  q_i (k) -x_{0} +(u_{i}(k)u _{i}(k)-1)(q_i (k) -x_{0})-u_{i}(k)\nu_{i}(k)$, $\{\xi_{ji}(k),i,j=1,\ldots,N,k\geq0\}$ are independent standard normally distributed random variables,  $\{u_{i}(k),k\geq0\}$ and $\{\nu_{i}(k),k\geq0\}$ are mutually independent i.i.d.   sequences with standard normal distribution.
The elapsed time of DPG  and algorithm  (2)-(4) to achieve  the same precision may be roughly equivalent. This is mainly because the subproblem $\min_{z\in \mathbb{R}^{n}} (\|z\|_{1}+\frac{1}{2\alpha}\|z-y\|^{2})$, $\alpha >0$, $ \forall \ y \in \mathbb{R}^{n}$ in DPG can be solved in a closed form (\cite{W. Shi}). However,  for many scenarios,  the subproblem may not have an analytical solution (\cite{Schmidt}), and may require more time and computational resources to be solved. Generally, whether to use DSG  or DPG may depend on specific problems.

\vskip 0.2cm

We consider the following conditions of algorithm step sizes, which are required to hold simultaneously.

\begin{itemize}[\itemindent=8pt]
  \global\setlength{\labelsep}{3pt}
  \item[(C1)] $c(k)\downarrow0$,  $\alpha(k)\downarrow0$, $\sum_{k=0}^{\infty}\alpha(k)=\infty, \ \sum_{k=0}^{\infty}\alpha^2(k)<\infty,\ \sum_{k=0}^{\infty}c^2(k)<\infty$, $c(k)=O(c(k+1))$, $k\rightarrow\infty$;
  \item[(C2)] $\lim\limits_{k\rightarrow\infty}\dfrac{c^2(k)}{\alpha(k)}=0;$
  \item[(C3)] for any given positive constant $C$, $\sum_{k=0}^{\infty}\alpha(k)\exp(-C\sum_{t=0}^k\alpha(t))<\infty$;
  \item[(C4)] for any given positive constant $C$,  $\lim_{k\rightarrow\infty}\frac{\alpha(k)\exp(C\sum_{t=0}^k\alpha(t))}{c(k)}=0$;
  \item[(C5)] for any given positive constant $C$, the sequence  $\{\alpha(k)\exp(C\sum_{t=0}^k\alpha(t)),k\geq0\}$  decreases monotonically for   sufficiently large $k$ and
      $\alpha(k)\exp(C\sum_{t=0}^k\alpha(t)) -\alpha(k+1)\exp(C\sum_{t=0}^{k+1}\alpha(t))$
      =$O(\alpha^{2}(k)$ $\times \exp(2C\sum_{t=0}^k\alpha(t)))$.
\end{itemize}\vskip 0.2cm

\vskip 0.2cm

\begin{remark}\label{examplealphack}
 There exist step sizes satisfying Conditions (C1)-(C5). For example, $\alpha(k)=\frac{\alpha_{1}}{(k+3)\ln^{\tau_{1}}(k+3)},\\ \ \tau_{1}\in(0,1],\ c(k)=\frac{\alpha_{2}}{(k+3)^{\tau_{2}}\ln^{\tau_{3}}(k+3)},\ \tau_{2}\in(0.5,1),\ \tau_{3}\in(-\infty,1],$
where $\alpha_{1},\alpha_{2}$ are given positive constants. See Appendix E for further details.
\end{remark}

\section{Main results}
\label{sec:main}

Let $D(k)=\operatorname{diag}(a_{1}^T(k),\ldots,a_{N}^T(k))\otimes I_{n}$, {where} $a_{i}^T(k)$ {is} the $i$th row of $\mathcal{A}_{\mathcal{G}(k)}$, {$i=1,...,N$}, $\psi_{i}(k)=\operatorname{diag}(\psi_{1i}(x_{1}(k)-x_{i}(k)),\ldots,\psi_{Ni}(x_{N}(k)-x_{i}(k)))$, {$i=1,...,N$},
$\Psi(k)=\operatorname{diag}(\psi_{1}(k),\ldots,$ $ \psi_{N}(k))\otimes I_{n}$, {and} $d(k)=[d^T_{f_{1}}(x_{1}(k)),\ldots,d^T_{f_{N}}(x_{N}(k))]^T$. Let $\sigma_{d}=\max_{1\leq i\leq N}\{\sigma_{di}\}$, $C_{d}=\max_{1\leq i\leq N}\{C_{di}\}$, $\sigma=\max_{1\leq i,j\leq N}\{\sigma_{ji}\}$, and $b=\max_{1\leq i,j\leq N}\{b_{ji}\}$.

Rewrite the algorithm \eqref{algorithm}-\eqref{algorithm2} in a compact form as
\begin{equation}\label{compact}
X(k+1)=((I_{N}-c(k)\mathcal{L}_{\mathcal{G}(k)})\otimes I_{n})X(k)+c(k)D(k)\Psi(k)\xi(k)-\alpha(k)(d(k)+\zeta(k)).
\end{equation}

Denote $\bar{x}(k):=\frac{1}{N}\sum_{i=1}^Nx_{i}(k)$, the consensus error vector $\delta(k):=(P\otimes I_{n})X(k)$ and the Lyapunov function $V(k):=\|\delta(k)\|^2$, where $P=I_{N}-\frac{1}{N}\mathbf{1}_{N}\mathbf{1}_{N}^T$.  By $(\mathcal{L}_{\mathcal{G}(k)}\otimes I_{n})(\mathbf{1}_{N}\mathbf{1}_{N}^T\otimes I_{n})=O_{nN\times nN}$, we have $(\mathcal{L}_{\mathcal{G}(k)}\otimes I_{n})X(k)=(\mathcal{L}_{\mathcal{G}(k)}\otimes I_{n})\delta(k)$. Therefore,
$
(P\otimes I_{n})((I_{N}-c(k)\mathcal{L}_{\mathcal{G}(k)})\otimes I_{n})X(k)
=((I_{N}-c(k)P\mathcal{L}_{\mathcal{G}(k)})\otimes I_{n})\delta(k),
$
which together with \eqref{compact} gives
\begin{align}
\label{error}
\delta(k+1)
=&((I_{N}-c(k)P\mathcal{L}_{\mathcal{G}(k)})\otimes I_{n})\delta(k)+(P\otimes I_{n})(c(k)D(k)\Psi(k)\xi(k)-\alpha(k)\zeta(k))
-\alpha(k)(P\otimes I_{n})d(k).
\end{align}

\subsection{The convergence of the algorithm \eqref{algorithm}-\eqref{algorithm2}}

In the following theorem, we will prove the convergence of the algorithm \eqref{algorithm}-\eqref{algorithm2}.

\vskip 0.2cm

\begin{theorem}\label{theorem3.1}
For the  convex optimization problem \eqref{model} and the algorithm \eqref{algorithm}-\eqref{algorithm2}, assume that

(a) Assumptions \ref{subgradient}-\ref{stochasticgraph} and Conditions (C1)-(C5) hold;

(b) there {exists} a positive integer $h$, positive constants $\theta$ and $\rho_{0}$, such that

\ \ \ (b.1) $\inf_{m\geq0}\lambda_{mh}^h\geq\theta$\ \ a.s.;

\ \ \ (b.2) $\sup_{k\geq0}\left[E\left[\|\mathcal{L}_{\mathcal{G}(k)}\|^{2^{\max\{h,2\}}}\Big|
\mathcal{F}(k-1)\right]\right]^{\frac{1}{2^{\max\{h,2\}}}}\leq\rho_{0}$\ \ a.s.\\
Then, there exists a random vector $z^{*}$ taking values in $\mathcal{X}^{*}$, such that
$\lim_{k\rightarrow\infty}x_{i}(k)=z^{*}$, $i=1, \cdots, N$ a.s.
\end{theorem}

\vskip 0.2cm

The proof of \Cref{theorem3.1} needs the following three lemmas whose proofs are put in Appendix A.

\vskip 0.2cm

\begin{lemma}\label{lemma2}
For the  convex optimization problem \eqref{model} and the algorithm \eqref{algorithm}-\eqref{algorithm2}, if Assumptions \ref{subgradient}-\ref{intensity} and Assumption \ref{stochasticgraph} hold, and there exists a positive constant $\rho_{0}$, such that $\sup_{k\geq0}\Big[E\big[\|\mathcal{L}_{\mathcal{G}(k)}\|^2\Big|
\mathcal{F}(k-1)\big]\Big]^{\frac{1}{2}}\leq\rho_{0}$ a.s.,  then the following inequalities hold.

(i)\begin{align}
& E[V(k+1)|\mathcal{F}(k-1)]\notag\\
\leq&\big(1+2c^2(k)(\rho^2_{0}+8\sigma^2C_{\xi}\rho_{1})
\big)V(k)+8b^2C_{\xi}\rho_{1}c^2(k)+2\alpha^2(k)(2\sigma_{\zeta}
+3\sigma^2_{d})\|X(k)\|^2\notag\\
&+2\alpha^2(k)
\left(2C_{\zeta}+3NC^2_{d}\right)+2\alpha(k)\left\|d(k)\right\|\left\|\delta(k)\right\|,\ \forall\ k\geq0 \text{ a.s.};\label{itconvk}
\end{align}

(ii)
\begin{align}
&E\left[\|X(k+1)-\mathbf{1}_{N}\otimes x\|^2|\mathcal{F}(k-1)\right]\notag\\
\leq&\left(1+2c^2(k)(\rho^2_{0}+8\sigma^2C_{\xi}\rho_{1})
+4\alpha^2(k)\left(2\sigma_{\zeta}+3\sigma^2_{d}\right)\right)\|X(k)-\mathbf{1}_{N}\otimes x\|^2+8b^2C_{\xi}\rho_{1}c^2(k)\notag\\
&+2\alpha^2(k)\big(2C_{\zeta}+3NC^2_{d}+2(3\sigma^2_{d}
+2\sigma_{\zeta})N\|x\|^2\big)-2\alpha(k)d^T(k)(X(k)-\mathbf{1}_{N}\otimes x),\ \forall\ x\in \mathbb{R}^n,\ k\geq0 \text{ a.s.}, \label{itconvk1}
\end{align}
where $\rho_{1}$ is a positive constant satisfying $\sup_{k\geq0}E\left[|\mathcal{E}_{\mathcal{G}(k)}| \max_{1\leq i,j\leq N}a_{ij}^2(k)\big|\mathcal{F}(k-1)\right]\leq\rho_{1}$ a.s.
\vskip 2mm
\end{lemma}
\vskip 0.2cm

\vskip 0.2cm

\begin{lemma}\label{lemma3}
For the  convex optimization problem \eqref{model} and the algorithm \eqref{algorithm}-\eqref{algorithm2},  if Assumptions \ref{subgradient}-\ref{intensity}, Assumption \ref{stochasticgraph} and Conditions (C1)-(C3) hold, and there exists a positive constant $\rho_{0}$, such that $\sup_{k\geq0}
\Big[E\big[\|\mathcal{L}_{\mathcal{G}(k)}\|^2\big|\mathcal{F}(k-1)\big]\Big]^{\frac{1}{2}}\leq\rho_{0}$ a.s., then  there exists a constant $C_{1}>0$, such that
\begin{align}
E\left[\|X(k)\|^2\right]\leq C_{1}\beta(k),\ \forall\ k\geq0,\label{incresecon}
\end{align}
where $\beta(k)=\exp(C_{0}\sum_{t=0}^k\alpha(t))$, and $C_{0}=1+2\rho^2_{0}+16\sigma^2C_{\xi}\rho_{1}+8\sigma_{\zeta}+14\sigma^2_{d}$.
\end{lemma}
\vskip 0.2cm

\vskip 0.2cm

\begin{lemma}\label{lemmac}
For the  convex optimization problem \eqref{model} and the algorithm \eqref{algorithm}-\eqref{algorithm2}, assume that

(a) Assumptions \ref{subgradient}-\ref{intensity}, Assumption \ref{stochasticgraph} and Conditions (C1)-(C5) hold;

(b) there {exists} a positive integer $h$, positive constants $\theta$ and $\rho_{0}$, such that

\ \ \ (b.1) $\inf_{m\geq0}\lambda_{mh}^h\geq\theta$\ \ a.s.;

\ \ \ (b.2) $\sup_{k\geq0}\left[E\left[\|\mathcal{L}_{\mathcal{G}(k)}\|^{2^{\max\{h,2\}}}\Big|
\mathcal{F}(k-1)\right]\right]^{\frac{1}{2^{\max\{h,2\}}}}\leq\rho_{0}$\ \ a.s.\\
%
%
%
Then,
\begin{itemize}
\item[(i)] there exists a constant $C_{2}>0$, such that
$E[V(k)]\leq C_{2}\beta^{-3}(k),\ k\geq0,
$
where $\beta(k)$ is given in \Cref{lemma3};
\item[(ii)]  $V(k)$ vanishes almost surely, i.e. $
V(k)\rightarrow0,\ k\rightarrow\infty\ \mbox{a.s.}
$
\end{itemize}
\end{lemma}

\noindent
{\bfseries Proof of \Cref{theorem3.1}: }
By Assumption \ref{subgradient} and the convexity of local cost functions, we have, for any $x^{*}\in\mathcal{X}^{*},$
\begin{align}
& -d^T_{f_{i}}(x_{i}(k))(x_{i}(k)-x^{*})\notag\\
\leq& f_{i}(x^{*})-f_{i}(\bar{x}(k))+f_{i}(\bar{x}(k))-f_{i}(x_{i}(k))\notag\\
\leq& f_{i}(x^{*})-f_{i}(\bar{x}(k))+d^T_{f_{i}}(\bar{x}(k))(\bar{x}(k)-x_{i}(k))\notag\\
\leq &f_{i}(x^{*})-f_{i}(\bar{x}(k))+(\sigma_{d}\|\bar{x}(k)
\|+C_{d})\|\bar{x}(k)-x_{i}(k)\|.\notag
\end{align}
Then, by  H\"{o}lder inequality,  we have
\begin{align}
  &-2\alpha(k)d^T(k) (X(k)-\mathbf{1}_{N}\otimes x^{*})\notag\\
 \leq  &  2\alpha(k)(\sqrt{N}(\sigma_{d}\|\bar{x}(k)
\|+C_{d}) \|\delta(k)\| -(f(\bar{x}(k))-f(x^{*}))).\notag
 \end{align}
Then, by \Cref{lemma2} (ii), we have, for any $x^{*}\in\mathcal{X}^{*},$
\begin{align}
&E[\|X(k+1)-\mathbf{1}_{N}\otimes x^{*}\|^2|\mathcal{F}(k-1)]\notag\\
\leq&(1+2c^2(k)(\rho^2_{0}+8\sigma^2C_{\xi}\rho_{1})
+4\alpha^2(k)(2\sigma_{\zeta}+3\sigma^2_{d}))  \|X(k)-\mathbf{1}_{N}\otimes x^{*}\|^2
+8b^2C_{\xi}\rho_{1}c^2(k)\notag\\
&+2\alpha^2(k)\big(2C_{\zeta}+3NC^2_{d}+2(3\sigma^2_{d}
+2\sigma_{\zeta})N\| x^{*}\|^2\big)+2\sqrt{N}\alpha(k)\|\delta(k)\| (\sigma_{d}\|\bar{x}(k)
\|+C_{d})\notag\\
&-2\alpha(k) (f(\bar{x}(k))-f(x^{*})))\  \hbox{a.s.} \label{itetooptimal1}
\end{align}
From $C_{r}$ inequality,
 H\"{o}lder inequality, \Cref{lemma3}, \Cref{lemmac} (i),  $\beta(k)\geq1$,  Condition (C3) and Corollary 4.1.2 in \cite{YSChow}, we have
 \begin{align}
 &E\left[ \sum_{k=0}^{\infty}  \alpha(k)\|\delta(k)\| \left(\sigma_{d}\|\bar{x}(k)\|+C_{d}\right)\right]\notag\\
 \leq &\sum_{k=0}^{\infty}  \alpha(k)\left[E\left[(\sigma_{d}\|\bar{x}(k)\|+C_{d})^2\right]\right]^{\frac{1}{2}}
\left[E\left[\|\delta(k)\|^2\right]\right]^{\frac{1}{2}}\notag\\
\leq &\sum_{k=0}^{\infty}  \alpha(k)\Big(\frac{2}{N}\sigma^2_{d}E\big[ \|X(k) \|^2\big]+2C^2_{d}
\Big)^{\frac{1}{2}}\sqrt{E[V(k)]}\notag\\
\leq &\Big(\Big(\frac{2}{N}\sigma^2_{d}C_{1}+2C^2_{d}\Big)C_{2}\Big)^{\frac{1}{2}}\sum_{k=0}^{\infty}  \alpha(k)\beta^{-1}(k)<\infty.\notag
 \end{align}
This together with the non-negativity of $\sum_{k=0}^{\infty} \alpha(k) \|\delta(k)\|  (\sigma_{d}\|\bar{x}(k)\|  +C_{d}) $ gives
\begin{align}
 \sum_{k=0}^{\infty} \alpha(k)\|\delta(k)\|(\sigma_{d}\|\bar{x}(k)
\|+C_{d})<\infty \text{ a.s.}\notag
\end{align}
By Condition (C1) and the above inequality, we have
\begin{align}
& \sum_{k=0}^{\infty}  \big(8b^2 C_{\xi}\rho_{1}c^2(k)+2\alpha^2(k)\big(2C_{\zeta}+3NC^2_{d}+2(3\sigma^2_{d}
+2\sigma_{\zeta}) N\| x^{*}\|^2\big)\notag\\
& +2\sqrt{N}\alpha(k)\|\delta(k)\|(\sigma_{d}\|\bar{x}(k)
\|+C_{d}) \big) <\infty \ \text{ a.s.} \notag
 \end{align}
Then, noting that
 $f(\bar{x}(k))-f^{*}\geq0$ and by
Theorem 1 in \cite{Robbins}  and Condition (C1), we obtain that, for any given $x^*\in \mathcal{X}^{*}$, there exists a measurable set $\Omega_{x^*}$ with $P\{\Omega_{x^*}\}=1$, such that  for any  $\omega\in\Omega_{x^*}$, $\{\|X(k,\omega)-\mathbf{1}_{N}\otimes x^{*}\|,k\geq0\}$  converges  and $\sum_{k=0}^{\infty}\alpha(k)[f(\bar{x}(k,\omega))-f^*]<\infty$.
This together with $f(\bar{x}(k,\omega))\geq f^*$ and $\sum_{k=0}^{\infty}\alpha(k)=\infty$ gives that, for any  $\omega\in\Omega_{x^*}$, $\sup_{k\geq0}\|X(k,\omega)\|<\infty$ and
\begin{equation}\label{limf}
\liminf\limits_{k\rightarrow\infty}f(\bar{x}(k,\omega))=f^*.
\end{equation}
Denote $\Omega_{1}=\{\omega|\lim_{k\rightarrow\infty}\|x_{i}(k,\omega)
-\bar{x}(k,\omega)\|=0,\ i=1,\ldots,N.\}$. From 
\Cref{lemmac} (ii), we know that $P\{\Omega_{1}\}=1$.
Denote $\Omega_{0}=(\bigcap_{x^{*}\in\mathcal{X}^{*}}\Omega_{x^{*}})\bigcap\Omega_{1}$. From \Cref{countable}, it follows that $P\{\Omega_{0}\}=1$. For any given $\omega\in\Omega_{0}$, by (\ref{limf}), we know that there is a subsequence $\{\bar{x}(k_{l},\omega),l\geq0\}$ of $\{\bar{x}(k,\omega),k\geq0\}$ such that $\lim_{l\rightarrow\infty}f(\bar{x}(k_{l},\omega))=f^*$.
By the continuity of $f$ and the boundedness of $\{\bar{x}(k_l,\omega)$, $l\geq0\}$,   we know that  there is a subsequence $\{\bar{x}(k_{l^{'}}, \omega), l^{'}\geq0\}$ of $\{\bar{x}(k_l,\omega)$, $l\geq0\}$,  converging to a point $z^{*}(\omega)$ in $\mathcal{X}^{*}$, i.e. $\lim_{l^{'}\rightarrow\infty}\bar{x}(k_{l^{'}},\omega)=z^*(\omega)$, which gives $\lim_{l^{'}\rightarrow\infty}\|x_i(k_{l^{'}},\omega)-z^*(\omega)\|=0$, $i=1,2,...,N$. Then we get
$\lim_{l^{'}\rightarrow\infty}\|X(k_{l^{'}},\omega)-\mathbf{1}_{N}\otimes z^{*}(\omega)\|=0$.
This together with the convergence  of $\{\|X(k,\omega)-\mathbf{1}_{N}\otimes z^{*}(\omega)\|,k\geq0\}$  leads to $\lim_{k\rightarrow\infty}\|x_i(k,\omega)-z^*(\omega)\|=0$, $i=1,2,...,N$. Then by the arbitrariness of $\omega$ and $P\{\Omega_{0}\}=1$, we get that $\lim_{k\rightarrow\infty}x_{i}(k)=z^*$  a.s., $i=1,\cdots,N.$       \QEDA

\vskip 0.2cm

\subsection{Special cases}
Next, we consider some special cases of random graph sequences.

At first, we suppose $\{\mathcal{G}(k),k\geq0\}$  is a Markov chain with countable state space. For this  case, Condition (b.1) of \Cref{theorem3.1} becomes more intuitive and Condition (b.2) is weakened.

Denote $S_{1}=\{\mathcal{A}_{j},j=1,2,\ldots\}$, which is a countable set of generalized weighted adjacency matrices and denote the associated generalized Laplacian matrix of $\mathcal{A}_{j}$ by $\mathcal{L}_{j}$. Let $\widehat{\mathcal{L}}_{j}=\frac{\mathcal{L}_{j}+\mathcal{L}^T_{j}}{2}$. We consider the random graph sequences
$$\begin{array}{rcl}
\Gamma_{2}&=&\left\{\{\mathcal{G}(k),k\geq0\}|\{\mathcal{A}_{\mathcal{G}(k)},k\geq0\}\subseteq S_{1}
\text{ is a homogeneous and uniformly ergodic Mar-}\right.\\
&&\text{kov chain with unique stationary distribution}\ \pi; E\left[\mathcal{A}_{\mathcal{G}(k)} \big| \mathcal{A}_{\mathcal{G}(k-1)}\right] \succeq O_{N \times N} \ \mbox{a.s.},\\
&&\left.\text{and the associated digraph of } E\left[\mathcal{A}_{\mathcal{G}(k)} \big| \mathcal{A}_{\mathcal{G}(k-1)}\right] \text{is balanced a.s., } k\geq 0\right\},
\end{array}$$
where $\pi=[\pi_{1},\pi_{2},\ldots]^T$, $\pi_{j}\geq 0$, $\sum_{j=1}^{\infty}\pi_{j}=1$ and $\pi_{j}$ is the stationary probability at $\mathcal{A}_{j}$.
We have the following corollary.

\begin{corollary}\label{corollary3.1}

For the convex optimization problem \eqref{model}, the algorithm \eqref{algorithm}-\eqref{algorithm2} and the associated random graph sequence $\{\mathcal{G}(k),k\geq0\}\in\Gamma_{2}$, assume that

(i) Assumptions \ref{subgradient}-\ref{countable} and Conditions (C1)-(C5) hold;

(ii) the associated graph of the Laplacian matrix $\sum_{j=1}^{\infty}\pi_{j}\mathcal{L}_{j}$ contains a spanning tree;

(iii) $\sup_{j\geq 1}\|\widehat{\mathcal{L}}_{j}\|<\infty$.\\
Then, there exists  a random vector $z^{*}$ taking values in $\mathcal{X}^{*}$, such that
$\lim_{k\rightarrow\infty}x_{i}(k)=z^{*}$\ a.s.,\ $i=1,\cdots,N.$
\end{corollary}
\vskip 0.2cm
\begin{proof}
From the definition of $\Gamma_{2}$, we know that $\Gamma_{2}\subseteq\Gamma_{1}$. Then, similar to the proof of Theorem 2 in \cite{LiT}, we get that Condition (b.1) of \Cref{theorem3.1} holds by Condition (ii). From Condition (iii), we know that Condition (b.2) of \Cref{theorem3.1}  holds. Finally, the conclusion of \Cref{corollary3.1} is obtained by \Cref{theorem3.1}.
\end{proof}
\vskip 0.2cm

Consider the independent graph sequences
$$\begin{array}{rcl}
\Gamma_{3}&=&\left\{\{\mathcal{G}(k),k\geq0\}|\{\mathcal{G}(k),k\geq0\}\text{ is an independent process, } E\left[\mathcal{A}_{\mathcal{G}(k)}\right] \succeq O_{N \times N}, \right.\\
&&\left.\text{and the associated digraph of }E\left[\mathcal{A}_{\mathcal{G}(k)}\right] \text{is balanced, } k\geq 0\right\}.
\end{array}$$


\vskip 0.2cm
\begin{theorem}\label{corollary3.2}

For the convex optimization problem \eqref{model}, the algorithm \eqref{algorithm}-\eqref{algorithm2} and the associated random graph sequence $\{\mathcal{G}(k),k\geq0\}\in\Gamma_{3}$, assume that

(i) Assumptions \ref{subgradient}-\ref{countable} and Conditions (C1)-(C5) hold;

(ii) there exists a positive integer $h$ such that
$$\inf _{m \geq 0}\lambda_{2}\left(\sum_{i=m h}^{(m+1) h-1} E\left[\widehat{\mathcal{L}}_{\mathcal{G}(i)}\right]\right)>0;$$

(iii) $\sup _{k \geq 0} E\left[\left\|\mathcal{L}_{\mathcal{G}(k)}\right\|^{2}\right]<\infty$.\\
Then, there exists  a random vector $z^{*}$ taking values in $\mathcal{X}^{*}$, such that
$\lim_{k\rightarrow\infty}x_{i}(k)=z^{*}$\ a.s.,\ $i=1,\cdots,N.$
\end{theorem}

The proof of \Cref{corollary3.2} is similar to that of Theorem 3.1 and is omitted here. For details, see Appendix A. The only difference is that
by the independence between $\mathcal{L}_{\mathcal{G}(i)}$ and $\mathcal{L}_{\mathcal{G}(j)}$, $i\neq j$,  we prove (\ref{Vmh4}) of Lemma 3.3 by Lyapunov inequality and condition (iii) of Theorem 3.2 instead of using conditional H\"{o}lder inequality  in the proof of Theorem 3.1. Thus, there is no dependence on $h$  in condition (iii) of \Cref{corollary3.2}.

\vskip 0.2cm
Noting that an undirected graph is  balanced, we get the following corollary, which is consistent with the results in \cite{JakoveticD1} and  \cite{Srivastava}, directly from \Cref{corollary3.2}.
\vskip 0.2cm

Consider the i.i.d undirected graph sequences
$$\begin{array}{rcl}
\Gamma_{4}&=&\left\{\{\mathcal{G}(k),k\geq0\}|\{\mathcal{G}(k),k\geq0\}\text{ is an i.i.d process, } E\left[\mathcal{A}_{\mathcal{G}(k)}\right] \succeq O_{N \times N}, \right.\\
&&\left.\text{and the associated graph of }E\left[\mathcal{A}_{\mathcal{G}(k)}\right] \text{is undirected, } k\geq 0\right\}.
\end{array}$$
\vskip 0.2cm

\begin{corollary}\label{corollary3.3}

For the convex optimization problem \eqref{model}, the algorithm \eqref{algorithm}-\eqref{algorithm2} and the associated random graph sequence $\{\mathcal{G}(k),k\geq0\}\in\Gamma_{4}$, assume that

(i) Assumptions \ref{subgradient}-\ref{countable} and Conditions (C1)-(C5) hold;

(ii) $\lambda_2\left(E\left[\mathcal{L}_{\mathcal{G}(0)}\right]\right)>0$;

(iii) $E\left[\left\|\mathcal{L}_{\mathcal{G}(0)}\right\|^{2}\right]<\infty$.\\
Then, there exists  a random vector $z^{*}$ taking values in $\mathcal{X}^{*}$, such that
$\lim_{k\rightarrow\infty}x_{i}(k)=z^{*}$\ a.s.,\ $i=1,\cdots,N.$
\end{corollary}

\subsection{Discussion on the assumptions}

\begin{remark}

\Cref{subgradient} is weaker than the existing assumptions on the differentiability of local cost functions and the boundedness of (sub) gradients in distributed convex optimizations (\cite{ThinhT}, \cite{JakoveticD1}, \cite{AlghunaimSA}, \cite{SahuAK}-\cite{LiuS}).

The ways to deal with the convex cost functions with bounded or Lipschitz continuous (sub)gradients employ the boundness or Lipschitz continuity of the (sub)gradients, respectively (\cite{ThinhT}, \cite{AlghunaimSA}, \cite{SahuAK}-\cite{LiuS}).
In \cite{SahuAK}, the gradients of local cost functions satisfy Lipschitz continuity, in which, the key step of analyzing the mean square error
between the average state of nodes and the optimal solution is to obtain
the recursive inequality by adding and subtracting the gradients
of the global cost function at the average state and the optimal
solution, and then using the Lipschitz continuity of gradients.
That is, the mean square error  at the next time can be controlled by that at the
previous time and the consensus error. However, this can not be obtained for the case with the  linearly growing subgradients. Also, different from \cite{Srivastava},  the subgradients are not required to be bounded and the inequality (28) in \cite{Srivastava} does not hold.

As a result, the existing methods are no longer applicable. In fact, the inner product of the subgradients and the error between local optimizers' states and the global optimal solution inevitably exists in  the  recursive inequality of the conditional mean square error, which leads the nonegative supermartingale convergence theorem not to be used directly.
We first estimate the mean square increasing rate of the states in \Cref{lemma3}, and then substitute this rate into the recursive inequality  (\ref{itetooptimal1}) of the conditional mean square error between the state and the global optimal solution.

\end{remark}

\vskip 0.2cm
For the case with $\mu$-strongly convex local cost functions, the conditions on step sizes can be weaken. We have the following result.

\vskip 0.2cm
\begin{theorem}\label{strcoconver}

 For the optimization problem \eqref{model} and the algorithm \eqref{algorithm}-\eqref{algorithm2}, if the local cost functions $f_{i}(\cdot), \ i=1,\ldots,N$ are $\mu$-strongly convex, under the same assumptions  as \Cref{theorem3.1} with Conditions (C3)-(C5) replaced by (C3)'-(C5)', then
\begin{itemize}
 \item[(i)]$\lim_{k\to \infty}
  E\left[\|x_{i}(k)-z^{*}\|^{2}\right]=0,\ i=1,\ldots,N$;
  \item[(ii)]$ \lim_{k \to \infty}x_{i}(k)= z^{*}$, $i=1,\ldots,N$ a.s.,
 \end{itemize}
where $z^{*}$  is the unique optimal solution of \eqref{model} and
 \begin{itemize}[\itemindent=8pt]
  \global\setlength{\labelsep}{3pt}
  \item[(C3)'] $\sum_{k=0}^{\infty}\alpha^{\frac{3}{2}}(k)c^{-\frac{1}{2}}(k)<\infty$;
 \item[(C4)'] $\lim_{k\rightarrow\infty}\frac{\alpha(k)}{c(k)}=0$;
  \item[(C5)']for any given positive integer $h$,
      $\frac{\alpha(k)-\alpha(k+h)}{\alpha(k+h)}
      =o(c(k+h)),\ k \to \infty.$
\end{itemize}
\end{theorem}
The proof of the above result is put in Appendix A.


\subsection{Convergence rates for the case with strongly  convex local cost functions}
\label{sec:rate}

{If the local cost functions $f_{i}(\cdot)$, $i=1,\ldots,N$  in problem \eqref{model} are $\mu$-strongly convex, then there is a unique optimal solution for the optimization problem \eqref{model}.  We denote the unique optimal solution by $z^{*}\in\mathbb{R}^{n}$. The following theorem gives the convergence rates that the local optimizers' states converge to $z^{*}$ in mean square. The proof of the theorem is given in  Appendix \ref{app:rate}.}

\begin{theorem}
\label{theorem3.2}
{For the optimization problem \eqref{model}, the algorithm \eqref{algorithm}-\eqref{algorithm2} with step sizes $c(k)=\frac{c_{0}}{(k+1)^{\gamma_{1}}}$, $\alpha(k)=\frac{\alpha_{0}}{(k+1)^{\gamma_{2}}}$, where $\gamma_{1}\in(0.5,1)$, $\gamma_{2}\in(\gamma_{1},1]$, $c_{0}>0$, $\alpha_{0}>0$, assume that}

{(a) Assumptions \ref{subgradient}-\ref{intensity} and  and Assumption \ref{stochasticgraph}  hold and the local cost functions $f_{i}(\cdot),i=1,\ldots,N$ are $\mu$-strongly convex;}

{(b) there exists a positive integer $h$ and positive constants $\theta$ and $\rho_{0}$, such that}

{\ \ \ (b.1) $\inf_{m\geq0}\lambda_{mh}^h\geq\theta$\ a.s.;}

{\ \ \ (b.2) $\sup_{k\geq0}\left[E\left[\|\mathcal{L}_{\mathcal{G}(k)}
\|^{2^{\max\{h,2\}}}\Big|\mathcal{F}(k-1)\right]\right]^{\frac{1}
{2^{\max\{h,2\}}}}\leq\rho_{0}$\ a.s.\\
Then,}
{ the convergence rates of $E[\|X(k)-\mathbf{1}_{N}\otimes z^{*}\|^2]$ for different cases are given as follows:}


{\ \ \ (1) if $3\gamma_{1}>2\gamma_{2}$ and $\gamma_{2}\in(\gamma_{1},1)$, then $\limsup_{k\rightarrow\infty}(k+2)^{\gamma_{2}-\gamma_{1}}E\left[\|X(k+1)-\mathbf{1}_{N}\otimes z^{*}\|^2\right]\leq\frac{2NC_{\varphi_{1}}C_{\varphi_{2}}}{\mu\alpha_{0}};$}

{\ \ \ (2) if $3\gamma_{1}>2\gamma_{2}$, $\gamma_{2}=1$ and $\gamma_{1}+\frac{\mu\alpha_{0}}{2N}>1$, then $\limsup_{k\rightarrow\infty}(k+2)^{1-\gamma_{1}} E\left[\|X(k+1)-\mathbf{1}_{N}\otimes z^{*}\|^2\right]\leq\frac{C_{\varphi_{1}}C_{\varphi_{2}}}{\gamma_{1}+\frac{\mu\alpha_{0}}{2N}-1};$ }

{\ \ \ (3) if $3\gamma_{1}>2\gamma_{2}$, $\gamma_{2}=1$ and $\gamma_{1}+\frac{\mu\alpha_{0}}{2N}=1$, then $\limsup_{k\rightarrow\infty}(k+2)^{1-\gamma_{1}}(\ln(k+2))^{-1}E\big[\|X(k+1)-\mathbf{1}_{N}\otimes z^{*}\|^2\big]\leq C_{\varphi_{1}}C_{\varphi_{2}};$}

{\ \ \ (4) if $3\gamma_{1}>2\gamma_{2}$, $\gamma_{2}=1$ and $\gamma_{1}+\frac{\mu\alpha_{0}}{2N}<1$, then
$E\left[\|X(k)-\mathbf{1}_{N}\otimes z^{*}\|^2\right]=O\Big(k^{-\frac{\mu\alpha_{0}}{2N}}\Big);$}


{\ \ \ (5) if $3\gamma_{1}\leq2\gamma_{2}$ and $\gamma_{2}\in(\frac{3}{2}\gamma_{1},1)$, then $\limsup_{k\rightarrow\infty}(k+2)^{2\gamma_{1}-\gamma_{2}}E\left[\|X(k+1)-\mathbf{1}_{N}\otimes z^{*}\|^2\right]\leq\frac{2NC_{\varphi_{1}}C'_{\varphi_{2}}}{\mu\alpha_{0}};$}

{\ \ \ (6) if $3\gamma_{1}\leq2\gamma_{2}$, $\gamma_{2}=1$ and $\frac{\mu\alpha_{0}}{2N}>2\gamma_{1}-1$, then $\limsup_{k\rightarrow\infty}(k+2)^{2\gamma_{1}-1} E\left[\|X(k+1)-\mathbf{1}_{N}\otimes z^{*}\|^2\right]\leq\frac{C_{\varphi_{1}}C'_{\varphi_{2}}}{1-2\gamma_{1}+\frac{\mu\alpha_{0}}{2N}};$}

{\ \ \ (7) if $3\gamma_{1}\leq2\gamma_{2}$, $\gamma_{2}=1$ and $\frac{\mu\alpha_{0}}{2N}=2\gamma_{1}-1$, then $\limsup_{k\rightarrow\infty}(k+2)^{2\gamma_{1}-1}(\ln(k+2))^{-1}E\big[\|X(k+1)-\mathbf{1}_{N}\otimes z^{*}\|^2\big]\leq C_{\varphi_{1}}C'_{\varphi_{2}};$}

{\ \ \ (8) if $3\gamma_{1}\leq2\gamma_{2}$, $\gamma_{2}=1$ and $\frac{\mu\alpha_{0}}{2N}<2\gamma_{1}-1$, then
$E[\|X(k)-\mathbf{1}_{N}\otimes z^{*}\|^2]=O\Big(k^{-\frac{\mu\alpha_{0}}{2N}}\Big)$,\\
where
$C_{\varphi_{1}}=
\exp\Big(\max\{2(\rho^2_{0}+8\sigma^2C_{\xi}\rho_{1}),
4(2\sigma_{\zeta}+3\sigma^2_{d})\}\sum_{k=0}^{\infty}(c^2(t)+\alpha^2(t))\Big)$,  
$C_{\varphi_{2}}=4\frac{\alpha^2_{0}}{c_{0}}\Bigg(1+
2\bigg(2C_{V1}\\ \Big(\frac{1}{N}\sigma^2_{d}\Big(\frac{2NC_{d}^2}{\mu^2}
+1\Big)+C^2_{d}\Big)\bigg)^{\frac{1}{2}}
\Bigg)$,
$\ C'_{\varphi_{2}}=4\Bigg(8b^2C_{\xi}\rho_{1}+\frac{2\alpha^2_{0}}{c_{0}^2}\sqrt{2C_{V2}
\Big(\frac{1}{N}\sigma^2_{d}\Big(\frac{2NC_{d}^2}{\mu^2}+1\Big)+C^2_{d}\Big)}+1\Bigg)c^2_{0}$,
$C_{V1}=\frac{\eta^hh^2C_{\rho}\Big(2\sigma^2_{d}\Big(\frac{2NC_{d}^2}{\mu^2}+1\Big)
+2NC^2_{d}\Big)}{\theta^2}$,
$\ C_{V2}=\frac{\eta^hh^2C_{\rho}\Big(2\sigma^2_{d}\Big(\frac{2NC_{d}^2}{\mu^2}+1\Big)
+2NC^2_{d}\Big)\alpha_{0}^{2}}{\theta^2c_{0}^{3}}+$
 $\frac{4\eta^hh^2\rho_{1}C_{\rho}C_{\xi}\Big(4\sigma^2\Big(\frac{2NC_{d}^2}{\mu^2}+1\Big)
+2b^2\Big)}{\theta}$,
$\eta=2\big(1+c^2_{0}(\rho^2_{0}+8\sigma^2C_{\xi}\rho_{1})\big)$,
$C_{\rho}=\Big\{2^{2(h-1)}\sum_{l=0}^{2(h-1)}M_{2(h-1)}^l\rho_{0}^{2l}\Big\}^{\frac{1}{2}}$,
 $M_{2(h-1)}^l$ denotes the combinatorial number of choosing $l$ elements from $2(h-1)$.}
\end{theorem}

\vskip 0.2cm

\begin{remark}
The convergence rates of deterministic distributed subgradient optimization algorithms with strongly convex cost functions were studied in \cite{LiuS}, \cite{NedicA2}-\cite{XiC},
where the convergence rate of $|f(\bar{x}(k))-f^{*}|$ is $O\left(\frac{\log(k)}{\sqrt{k}}\right)$ in \cite{NedicA2}-\cite{XiC},
and the convergence rate of $\|X(k)-\mathbf{1}_{N}\otimes z^{*}\|$ is $O\left(\frac{1}{\sqrt{k}}\right)$ in \cite{LiuS}.
Different from \cite{LiuS}, \cite{NedicA2}-\cite{XiC}, we consider distributed stochastic subgradient descent algorithms with communication and subgradient noises over random communication graphs,
and show how various random factors affect the convergence rate of the algorithm in \Cref{theorem3.2}.
In \cite{JakoveticD1}, the convergence rates of the distributed stochastic gradient descent algorithm with precise communications were analyzed under the conditions that the communication graphs are i.i.d. and the mean graph is connected and undirected.
The convergence rate of the mean square error was given as $O\left(\frac{1}{k}\right)$ with the step sizes $c(k)=\frac{c_{0}}{(1+k)^{\gamma_{1}}}$ and $\alpha(k)=\frac{\alpha_{0}}{k+1}$, where $\gamma_{1}\in[0,0.5]$.
Here,  we consider the algorithm with noisy communications. For attenuating the noises, the step sizes in \cite{JakoveticD1} are not applicable.
In \Cref{theorem3.2}, $\gamma_{1}$ is greater than $0.5$, which leads to a slower convergence than $O\left(\frac{1}{k}\right)$.
Compared with \cite{JakoveticD1},  in \Cref{theorem3.2}, we provide a systematic convergence rate analysis for distributed stochastic optimization algorithms with noisy and general random graphs (i.e. which may be non-stationary) for the case with typical step sizes
and provide the explicit limit bounds of the  convergence rates in terms of algorithm and network parameters.
\end{remark}

\vskip 0.2cm

\begin{remark}
\Cref{theorem3.2} reveals an interesting dependence of the convergence rates on the step size parameters $\gamma_{1}$, $\gamma_{2}$ and $\alpha_{0}$, the cost function's strongly convex coefficient $\mu$, and the number of the local optimizers $N$. What's more, the result shows that larger parameters $\sigma,b,C_{\xi},\sigma_{d},C_{d},\sigma_{\zeta},h$ and $\rho_{0}$ lead to a slower convergence, where $\sigma_{d}=\max\limits_{1\leq i\leq N}\{\sigma_{di}\}$, $C_{d}=\max\limits_{1\leq i\leq N}\{C_{di}\}$ are the subgradient parameters given in \Cref{subgradient};  $C_{\xi}$ is given in \Cref{noise1}; $\sigma_{\zeta}$ is the intensity coefficient of subgradient noises given in \Cref{noise2}; $\sigma=\max\limits_{1\leq i,j\leq N}\{\sigma_{ji}\}$ and $b=\max\limits_{1\leq i,j\leq N}\{b_{ji}\}$ are the intensity coefficients of communication noises given in \Cref{intensity}; $h$ is the length of the intervals over which the graphs are jointly connected, and $\rho_{0}$ is given in Condition (b.2).
\Cref{theorem3.2} is helpful for choosing appropriate parameters in practical implementation to achieve fast convergence.
\end{remark}

\vskip 0.2cm

\section{Numerical simulation}
\label{sec:numerical}


Consider the LASSO regression problem \eqref{lasso} over a random network with $20$ nodes. Let  $\{\mathcal{G}(k),k\geq0\}$
be an  i.i.d sequence of random communication graphs. All the random weights of edges $\{a_{i,j}(k),\ i,j=1,\ldots,20\}$ are chosen by the following rules. For any nonnegative integer $m$, if $k=4m$, then $\{a_{i,j}(k),\ i,j=1,\ldots,20\}$  are uniformly distributed on $[0,1]$; if $k\neq 4m$, the random weights are uniformly distributed on $[-0.5,0.5]$.
We assume that the communication noises $\{\xi_{ji}(k),i,j=1,\ldots,N,k\geq0\}$ are independent standard normally distributed random variables. And $\{\xi_{ji}(k),i,j=1,\ldots,N,k\geq0\}$, $\{u_{i}(k),i=1,\ldots,N,k\geq0\}$, $\{\nu_{i}(k),i=1,\ldots,N,k\geq0\}$ and $\{\mathcal{G}(k),k\geq0\}$ are assumed to be mutually independent. Let $\kappa=0.1$, $\sigma=0.1$, $b=0.1$,  $n=1$ and $R_{u,i}=1$.
 Let $x_{0}=6$, then $x_{0}-\kappa>0$. Noting that $0\in\partial f(x_{0}-\kappa)$, it follows that $x_{0}-\kappa$ is the unique optimal solution to the problem \eqref{lasso} (\cite{Bertsekas1}).
Let the initial states $x_{i}(0)=0$, $i=1,2,\ldots,20$, the step sizes $c(k)=1/(k+3)^{0.75}$ and $\alpha(k)=3/(k+3)\ln(k+3)$.

Then, the algorithm \eqref{algorithm}-\eqref{algorithm2}
is implemented.  The trajectories of all local optimizers' states are shown in \Cref{fig1} (a), which shows that the states of all local optimizers converge to $x_{0}-\kappa$ asymptotically.
Let the step sizes $c(k)=1/(k+1)^{0.4}$ and $\alpha(k)=3/(k+1)$. The algorithm \eqref{algorithm}-\eqref{algorithm2} with and without communication noises are implemented, respectively.
\Cref{fig1} (b) shows that if there is no communication noise as in \cite{JakoveticD1} (i.e. $ \sigma=0$ and $b=0$), then the mean square error of the algorithm tends to $0$;  if there is communication noise (i.e. $ \sigma=0.1$ and $b=0.1$), then
 the mean square error does not vanish. For the case with communication noise ($ \sigma=0.1$ and $b=0.1$),
 we choose $c(k)=1/(k+1)^{0.7}$ and $\alpha(k)=16/(k+1)$.  \Cref{fig1} (c) shows that
the convergence rate of  $E[\|X(k)-\mathbf{1}_{N}\otimes z^{*}\|^2]$ is  $O(k^{-0.3})$, while  the mean square error does not vanish if the step sizes are $c(k)=1/(k+1)^{0.4}$ and $\alpha(k)=3/(k+1)$.


\begin{figure}[htbp]
    \centering
    \begin{minipage}{0.5\textwidth}
    \centering
    \includegraphics[width=1.2\linewidth,height=6cm]{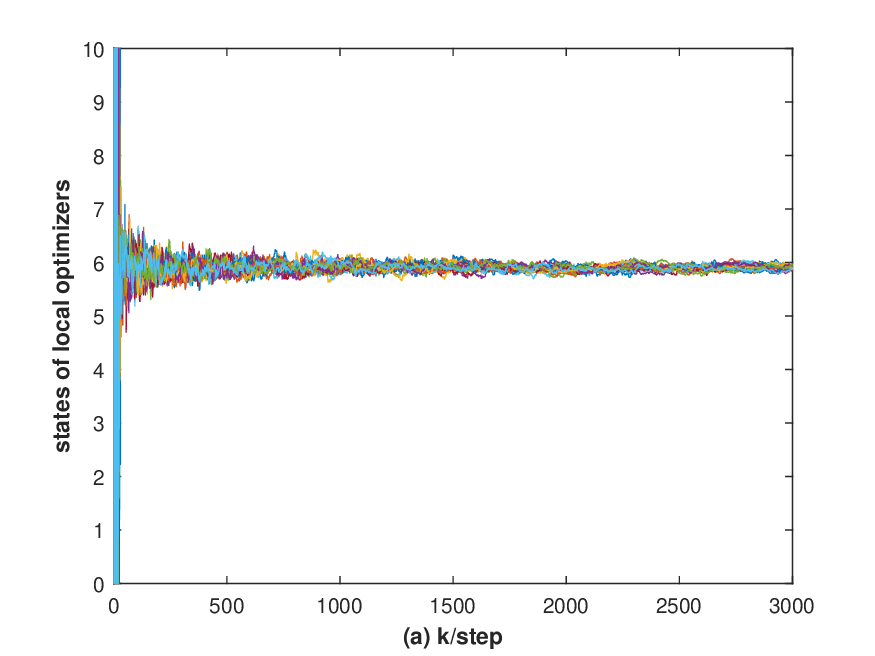}
    \end{minipage}
    \qquad
     \subfigure{\includegraphics[scale=0.5]{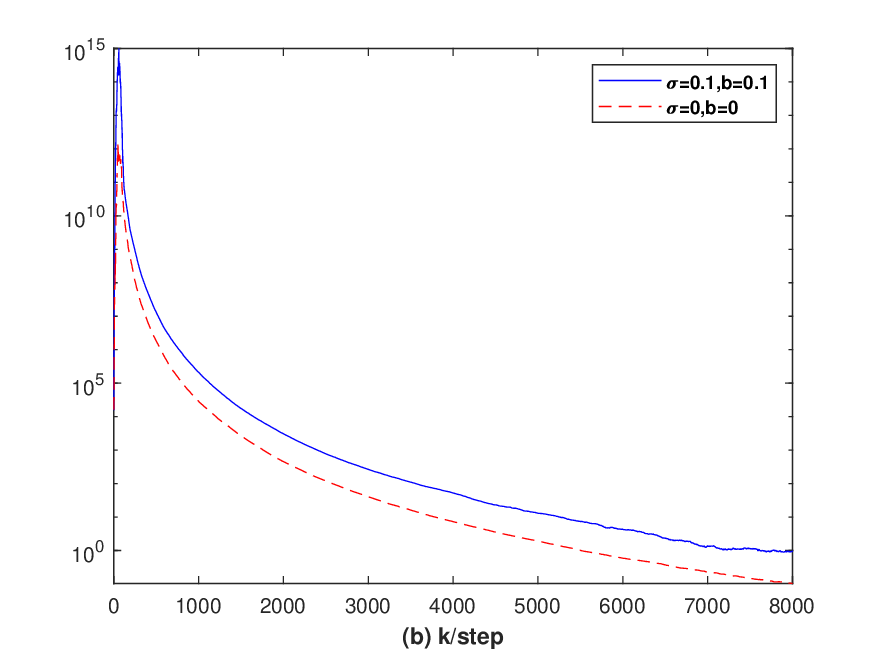}}
     \hspace{-5mm}
	\subfigure{\includegraphics[scale=0.5]{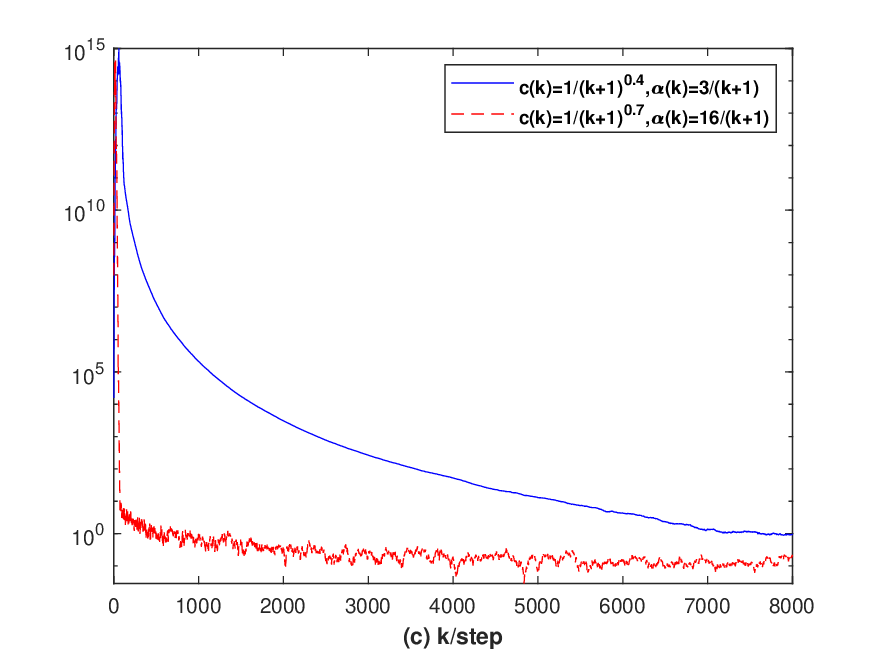}}

\caption{ (a) LASSO regression: trajectories of states; (b) LASSO regression: convergence  of mean square errors with $c(k)=1/(k+1)^{0.4}$ and $\alpha(k)=3/(k+1)$; (c) LASSO regression: the red dashed   line is the trajectory of $k^{0.3}E[\|X(k)-\mathbf{1}_{N}\otimes z^{*}\|^2]$; the blue  solid  line is the trajectory of $E[\|X(k)-\mathbf{1}_{N}\otimes z^{*}\|^2]$.}
\label{fig1}
	\end{figure}

\section{Conclusions}
\label{sec:conclusions}

We have studied the distributed stochastic subgradient  algorithm  for the stochastic optimization by networked nodes to cooperatively minimize a sum of convex cost functions.
  We have proved that if the local subgradient functions grow linearly and the sequence of digraphs is conditionally balanced and uniformly conditionally jointly connected, then proper algorithm step sizes can be designed so that all nodes' states converge to the global optimal solution almost surely.

 The  graph with a generalized weighted adjacency matrix  is often used to describe the competitive and cooperative interaction behaviors  existing in some scenarios of applications.
  So, it is also worth studying the distributed stochastic optimization over the network with the generalized weighted adjacency matrix in the future.

\appendices

\numberwithin{equation}{section}

\section{Proofs of Lemmas in \Cref{sec:main}}
\label{appendix:proof}


\noindent

\vskip 0.2cm

\emph{Proof of \Cref{lemma2}:}

By \eqref{error}, $C_{r}$ inequality and $\|P\otimes I_{n}\|=1$, we have
\begin{align}
\label{V(k+1)}
&V(k+1)\cr
\leq&V(k)-2c(k)\delta^T(k)\frac{\left(\mathcal{L}^T_{\mathcal{G}(k)
}P^T+P\mathcal{L}_{\mathcal{G}(k)}\right)\otimes I_{n}}{2}\delta(k)+2c^2(k)\|\mathcal{L}_{\mathcal{G}(k)}\otimes I_{n}\|^2\|\delta(k)\|^2 \cr
&+2(c(k)D(k)\Psi(k)\xi(k)-\alpha(k)\zeta(k))^T(P\otimes I_{n})((I_{N}-c(k)P\mathcal{L}_{\mathcal{G}(k)})\otimes I_{n})\delta(k)\cr
&+4c^2(k)\|D^{T}(k)D(k)\|\|\Psi(k)\|^{2}\|\xi(k)\|^2+4\alpha^2(k)\|\zeta(k)\|^2\cr
&+2\alpha(k)\|d(k)\|\|\delta(k)\|+3\alpha^{2}(k)\|d(k)\|^2.
\end{align}
Now, we consider the conditional expectation of each term on the right side of \eqref{V(k+1)}. For the 2nd term, by \Cref{stochasticgraph}    and  $\delta(k)\in\mathcal{F}(k-1)$
we have
%
\begin{align}
&E\left[-2c(k)\delta^T(k)\dfrac{(\mathcal{L}^T_{\mathcal{G}(k)}P^T
+P\mathcal{L}_{\mathcal{G}(k)})\otimes I_{n}}{2}\delta(k)\bigg|\mathcal{F}(k-1)\right]\notag\\
=&-2c(k)\delta^T(k)E[\widehat{\mathcal{L}}_{\mathcal{G}(k)}\otimes I_{n}|\mathcal{F}(k-1)]\delta(k) \notag\\
\leq & 0 \ \mbox{a.s.}\label{deltaLdelta}
\end{align}
For the $3$rd term, by $\sup_{k\geq0}\left[E\left[\|\mathcal{L}_{\mathcal{G}(k)}\|^2\big|
\mathcal{F}(k-1)\right]\right]^{\frac{1}{2}}\leq\rho_{0}$\ a.s.,  $\|\mathcal{L}_{\mathcal{G}(k)}\otimes I_{n}\|=\|\mathcal{L}_{\mathcal{G}(k)}\|$ and $\delta(k)\in\mathcal{F}(k-1)$, we get
\begin{align}
\label{Ldelta}
&E\left[2c^2(k)\|\mathcal{L}_{\mathcal{G}(k)}\otimes I_{n}\|^2\|\delta(k)\|^2\big|\mathcal{F}(k-1)\right]\notag\\
=&2c^2(k)E\left[\|\mathcal{L}_{\mathcal{G}(k)}\|^2\big|\mathcal{F}(k-1)\right]\|\delta(k)\|^2\notag\\
\leq&2\rho^2_{0}c^2(k)V(k)\ \mbox{a.s.}
\end{align}
By $\delta(k)\in\mathcal{F}(k-1)$ and \Cref{noise1}, we have
\begin{align}
\label{xidelta}
&E\Big[2c(k)\xi^T(k)\Psi^T(k)D^T(k)(P\otimes I_{n})((I_{N}-c(k)P\mathcal{L}_{\mathcal{G}(k)})\otimes I_{n}) \delta(k)|\mathcal{F}(k-1)\Big]\notag\\
=&2c(k)E[\xi^T(k)|\mathcal{F}(k-1)]\Psi^T(k) E[D^T(k)(P\otimes I_{n})((I_{N}-c(k) P\mathcal{L}_{\mathcal{G}(k)})\otimes I_{n})|\mathcal{F}(k-1)]\delta(k)
=0.
\end{align}
Similarly, by \Cref{noise2}, we have $E\left[2\alpha(k)\zeta^T(k)(P\otimes I_{n})((I_{N}-c(k)P\mathcal{L}_{\mathcal{G}(k)})\otimes I_{n})\delta(k)|\mathcal{F}(k-1)\right]=0$.
Thus, for the $4$th term, combining the above two equalities gives
\begin{align}\label{xiwdelta}
&E\left[2\Big(c(k)D(k)\Psi(k)\xi(k)-\alpha(k)\zeta(k)\Big)^T(P\otimes I_{n})((I_{N}-c(k)P\mathcal{L}_{\mathcal{G}(k)})\otimes I_{n})\delta(k)|\mathcal{F}(k-1)\right]=0.
\end{align}
By \Cref{intensity} and $C_{r}$ inequality, we have
\begin{align}
&\|\Psi(k)\|^2\notag\\
\leq&\max\limits_{1\leq i,j\leq N}\big[2\sigma^2\|x_{j}(k)-x_{i}(k)\|^2+2b^2\big]\notag\\
\leq&4\sigma^2\max\limits_{1\leq i,j\leq N}\big[\|x_{j}(k)-\bar{x}(k)\|^2+\|x_{i}(k)-\bar{x}(k)\|^2\big]+2b^2\notag\\
\leq &4\sigma^2\|X(k)-\mathbf{1}^T_{N}\otimes \bar{x}(k)\|^2+2b^2. \label{Ee3}
\end{align}
By
$$|\mathcal{E}_{\mathcal{G}(k)}|\max_{1\leq i,j\leq N}a_{ij}^2(k)\leq N(N-1)\max_{1\leq i,j\leq N}a^2_{ij}(k) \leq N(N-1)\|\mathcal{L}_{\mathcal{G}(k)}\|^2_{F}\leq N^{2}(N-1)\|\mathcal{L}_{\mathcal{G}(k)}\|^2,$$  and $\sup_{k\geq0}\left[E\left[\|\mathcal{L}_{\mathcal{G}(k)}\|^2\big|
\mathcal{F}(k-1)\right]\right]^{\frac{1}{2}}\leq\rho_{0}$
 a.s., there exists $\rho_{1}>0,$ such that
 \begin{align}
 E\left[|\mathcal{E}_{\mathcal{G}(k)}|\max_{1\leq i,j\leq N}a^2_{ij}(k)\big|\mathcal{F}(k-1)\right] \leq \rho_{1} \ \mbox{a.s.} \notag
 \end{align}
 Then, for the 5th  term, by (\ref{Ee3}), \Cref{noise1}, $\|P\otimes I_{n}\|=1$  and  $\|D^{T}(k)D(k)\|=\lambda_{max}(D^{T}(k)D(k))= \max_{1\leq i \leq N} \lambda_{max}(a_{i}(k)a^{T}_{i}(k))= \max_{1\leq i \leq N}\operatorname{Tr}(a_{i}^{T}(k)a_{i}(k))\leq \big|\mathcal{E}_{\mathcal{G}(k)}\big| \max\limits_{1\leq i,j\leq N}a^2_{ij}(k)$, we have
\begin{align}
&E\left[4c^2(k)\|D^{T}(k)D(k)\|\|\Psi(k)\|^{2}\|\xi(k)\|^2\big|\mathcal{F}(k-1)\right]\notag\\
=&4c^2(k)\|\Psi(k)\|^2E\left[\|D^{T}(k)D(k)\||\mathcal{F}(k-1)\right] E\left[\left\|\xi(k)\right\|^2\big|\mathcal{F}(k-1)\right]\notag\\
\leq&4c^2(k)C_{\xi} (4\sigma^2V(k)+2b^2) E\left[\big|\mathcal{E}_{\mathcal{G}(k)}\big|\max\limits_{1\leq i,j\leq N}a^2_{ij}(k)\big|\mathcal{F}(k-1)\right]\notag\\
\leq&16\sigma^2C_{\xi}\rho_{1}c^2(k)V(k)+8b^2C_{\xi}\rho_{1}c^2(k)\ \mbox{a.s.}\label{xixi}
\end{align}
 For the $6$th term, by \Cref{noise2}, we have
\begin{align}
&E\left[4\alpha^2(k) \left\|\zeta(k)\right\|^2\big|\mathcal{F}(k-1)\right]
\leq  4\sigma_{\zeta}\alpha^2(k)\|X(k)\|^2+4C_{\zeta}\alpha^2(k)\ \mbox{a.s.}\label{ww}
\end{align}
For the last term,
by \Cref{subgradient} and $C_{r}$ inequality, we have
\begin{align}
&E\left[3\alpha^{2}(k)\|d(k)\|^2|\mathcal{F}(k-1)\right]\notag\\
=& 3\alpha^2(k) \|d(k)\|^2
\leq6\alpha^2(k)(\sigma^2_{d}  \|X(k)\|^2+NC^2_{d}).\label{Edd}
\end{align}
This together with  \eqref{V(k+1)}-\eqref{ww} leads to
(\ref{itconvk}).

In the following part, we prove \Cref{lemma2} (ii). Noting that $\mathcal{L}_{\mathcal{G}(k)}\mathbf{1}_{N}=\mathbf{0}_{N}$ and by \eqref{compact}, we get
\begin{align}
&\|X(k+1)-\mathbf{1}_{N}\otimes x\|^2\notag\\
\leq&\|X(k)-\mathbf{1}_{N}\otimes x\|^2-2c(k)(X(k)-\mathbf{1}_{N}\otimes x)^T\frac{1}{2}\left(\mathcal{L}^T_{\mathcal{G}(k)}
+\mathcal{L}_{\mathcal{G}(k)}\right)\otimes I_{n}(X(k)-\mathbf{1}_{N}\otimes x)\notag\\
&+2c^2(k)\|\mathcal{L}_{\mathcal{G}(k)}\otimes I_{n}\|^2\|X(k)-\mathbf{1}_{N}\otimes x\|^2\notag\\
&+2(c(k)D(k)\Psi(k)\xi(k)-\alpha(k)\zeta(k))^T ((I_{N}-
c(k)\mathcal{L}_{\mathcal{G}(k)})\otimes I_{n})(X(k)-\mathbf{1}_{N}\otimes x)\notag\\
&+4c^2(k)\|D^{T}(k)D(k)\|\|\Psi(k)\|^{2}\|\xi(k)\|^2+4\alpha^{2}(k)
\|\zeta(k)\|^{2} \notag\\
&-2\alpha(k)d^T(k)(X(k)-\mathbf{1}_{N}\otimes x)+3\alpha^2(k)\|d(k)\|^2.\notag
\end{align}
Then, taking conditional expectation of each term on the right side of the above inequality, similar to the proof of \eqref{V(k+1)}-\eqref{ww}, we have
 (\ref{itconvk1}).
$\hfill\blacksquare$

\emph{Proof of \Cref{lemma3}:}
By $C_{r}$ inequality and (\ref{Edd}),  we have
\begin{align}
-2\alpha(k)d^T(k)X(k)
\leq\alpha(k)(2\sigma^2_{d}+1)\|X(k)\|^2
+2NC^2_{d}\alpha(k),\notag
\end{align}
which together with \Cref{lemma2} (ii) leads to
\begin{align}
\label{Xk+1}
&E\big[\|X(k+1)\|^2\big]\cr
\leq&\big(1+\alpha(k)(2\sigma^2_{d}+1)+2c^2(k)(\rho^2_{0}+8\sigma^2C_{\xi}\rho_{1})
+4\alpha^2(k) (2\sigma_{\zeta}+3\sigma^2_{d})\big)E\big[\|X(k)\|^2\big]+2NC^2_{d}\alpha(k)\cr
&+8b^2C_{\xi}\rho_{1}c^2(k)+2\alpha^2(k)
\big(2C_{\zeta}+3NC^2_{d}).
\end{align}
By Conditions (C1)-(C2), there exists a positive integer $k_{1}$, such that $\alpha^2(k)\leq \alpha(k)$ and $c^2(k)\leq \alpha(k)$,  $\forall\ k\geq k_{1}$. Thus, from \eqref{Xk+1},  we have
$$
E\left[\|X(k+1)\|^2\right]
\leq\left(1+C_{0}\alpha(k)\right)E\left[\|X(k)\|^2\right]+\widetilde{C}_{0}\alpha(k),\ \forall \ k\geq k_{1},
$$
where $\widetilde{C}_{0}=8NC^2_{d}+8b^2C_{\xi}\rho_{1}+4C_{\zeta}$.
This  gives
\begin{align}
&E\left[\left\|X(k+1)\right\|^2\right]\notag\\
\leq&\prod\limits_{i=k_{1}}^k\Big[1+C_{0}\alpha(i)\Big]
E\left[\left\|X(k_{1})\right\|^2\right]
+\widetilde{C}_{0}\sum\limits_{i=k_{1}}^k\prod\limits_{j=i+1}^k
\big(1+C_{0}\alpha(j)\big)\alpha(i)\cr
\leq&\exp\left(C_{0}\sum_{t=k_{1}}^k\alpha(t)\right)E\left[\left\|X(k_{1})\right\|^2\right]+
\widetilde{C}_{0}\sum\limits_{i=k_{1}}^k\alpha(i)
\exp\left(C_{0}\sum_{j=i+1}^k\alpha(j)\right)\cr
\leq&\beta(k+1)\left(E\left[\left\|X(k_{1})\right\|^2\right]+\widetilde{C}_{0}
\sum_{i=k_{1}}^{\infty}\alpha(i) \exp\left(-C_{0}\sum_{j=0}^i\alpha(j)\right)\right),\ \forall \ k\geq k_{1}.\notag
\end{align}
This together with Condition (C3)  leads to (\ref{incresecon}).
$\hfill\blacksquare$

\emph{Proof of \Cref{lemmac}:}
 In the following part, we will prove \Cref{lemmac} (i). By \Cref{lemma3}, we get
\begin{equation}
E\left[\left\|X(j)\right\|^2\right]\leq C_{1}\beta(j)\leq C_{1}\beta((m+1)h),\ \forall \ mh\leq j\leq (m+1)h.\notag
\end{equation}
Then,  taking $\tau(mh)=\alpha(mh)$  in  \Cref{lemmac(i)}, by  condition (b.1), noting that $\beta((m+1)h)\geq1$ and  $\alpha(j)\leq \alpha(mh)$, $c(j)\leq c(mh)$, $mh\leq j\leq (m+1)h$, there exist $C_{4}>0$ and $m_{0}>0$, such that
\begin{align}
\label{Vmh4}
&E[V((m+1)h)]\cr
\leq&(1+\alpha(mh))[1-2\theta c((m+1)h)+C_{4} c^2((m+1)h)]  E[V(mh)]\notag\\
&+4h\rho_{1}C_{\rho}C_{\xi} \big(4\sigma^2C_{1} \beta((m+1)h)+2b^2\big)\left( \sum_{j=mh}^{(m+1)h-1}c^2(j)\right)+4hC_{\rho}
\Big(C_{\zeta}  +\sigma_{\zeta}C_{1} \beta((m+1)h)
 \Big)\notag\\
& \times \sum_{j=mh}^{(m+1)h-1}\alpha^2(j)+\big(2\sigma^2_{d}C_{1}
\beta((m+1)h)+2NC^2_{d}\big)   \left(\frac{1}{\alpha(mh)}+2\right) \left(hC_{\rho}
\sum_{j=mh}^{(m+1)h-1}\alpha^2(j)  \right)\notag\\
\leq&\big(1-2\theta c((m+1)h)+q_{0}(mh)\big)E[V(mh)]+p(mh), \ \forall \  m\geq m_{0},
\end{align}
where
$q_{0}(mh)=C_{4} c^2((m+1)h)+\alpha(mh)\big(1-2\theta c((m+1)h)+C_{4} c^2((m+1)h)\big),$
$p(mh)=C_{5}\alpha(mh)\beta((m+1)h)+\big(2C_{5}+4h^2C_{\rho}$ $\big(\sigma_{\zeta}C_{1}+C_{\zeta}\big)\big)\alpha^2(mh)\beta((m+1)h)
+\big(4h^2\rho_{1}C_{\rho}C_{\xi}(4\sigma^2C_{1}+2b^2)\big)c^2(mh)\beta((m+1)h),$
$C_{5}=2h^2C_{\rho}(\sigma^2_{d}C_{1}+NC^2_{d}).$
From Conditions (C1) and (C4), we obtain $q_{0}(mh)=o(c((m+1)h))$. Thus, there exists a positive integer $m_{3}$, such that
\begin{equation}\label{hatq}
0<2\theta c((m+1)h)-q_{0}(mh)\leq1,\ \forall\ m\geq m_{3}.
\end{equation}
Let $\Pi(k)=\frac{c(k)V(k)}{\alpha(k)\beta(k)}$. By \eqref{Vmh4}, \eqref{hatq} and Condition (C1), we get
\begin{align}
&E[\Pi((m+1)h)]\notag\\
\leq&\dfrac{c((m+1)h)\alpha(mh)\beta(mh)}{c(mh)\alpha((m+1)h)
\beta((m+1)h)}[1-2\theta c((m+1)h)+q_{0}(mh)]E[\Pi(mh)]\notag\\
&+\dfrac{c((m+1)h)}{\alpha((m+1)h)\beta((m+1)h)}p(mh)\notag\\
\leq&\dfrac{\alpha(mh)\beta(mh)}{\alpha((m+1)h)\beta((m+1)h)}[1-2\theta c((m+1)h)+q_{0}(mh)] E[\Pi(mh)]\notag\\
&+\dfrac{c((m+1)h)}{\alpha((m+1)h)\beta((m+1)h)}p(mh),\ \forall \ m\geq\max\{m_{0},m_{3}\}.\label{pi11}
\end{align}
By Conditions (C1),  (C4) and (C5), we have
 $
\lim_{k\rightarrow\infty}\frac{\alpha(k+1)\beta(k+1)}{\alpha(k)\beta(k)}
=1-\lim_{k\rightarrow\infty}\frac{\alpha(k)\beta(k)
-\alpha(k+1)\beta(k+1)}{\alpha(k)\beta(k)}
=1,$
which implies
\begin{align}
\label{limalpha0}
&\lim\limits_{k\rightarrow\infty}\frac{\alpha(k)\beta(k)}{\alpha(k+h)\beta(k+h)}\notag\\
=&\lim\limits_{k\rightarrow\infty}\frac{\alpha(k)\beta(k)}{\alpha(k+1)\beta(k+1)}
\times\cdots\lim\limits_{k\rightarrow\infty}\frac{\alpha(k+h-1)\beta(k+h-1)}{\alpha(k+h)\beta(k+h)}=1.
\end{align}
This implies that there exists a constant $C_{6}>0$, such that
\begin{equation}\label{alphah}
\frac{\alpha(k)\beta(k)}{\alpha(k+h)\beta(k+h)}\leq C_{6},\ \forall\ k\geq0.
\end{equation}
From Condition (C5), there exists a positive integer $k_{2}$ and a positive constant $C_{7}$, such that
\begin{align}
&\alpha(k)\beta(k)-\alpha(k+h)\beta(k+h)\notag\\
=&\sum_{j=k}^{k+h-1}
\big(\alpha( j )\beta(  j )-\alpha( j+1)\beta(j+1)\big)\notag\\
\leq& C_{7}\sum_{j=k}^{k+h-1} \alpha^{2}( j ) \beta^{2}( j  )
\leq  hC_{7}\alpha^{2}(k)\beta^{2}(k),\ k\geq k_{2},\notag
\end{align}
which together with \eqref{alphah} leads to
\begin{align}
 \frac{\alpha(k)\beta(k)}{\alpha(k+h)\beta(k+h)}
=& 1+\frac{\alpha(k)\beta(k)-\alpha(k+h)\beta(k+h)}{\alpha(k+h)\beta(k+h)}\notag\\
\leq & 1+hC_{6}C_{7}\alpha(k)\beta(k),\ k\geq k_{2}.\notag
\end{align}
This together with \eqref{hatq} gives
\begin{align}
&\dfrac{\alpha(mh)\beta(mh)}{\alpha((m+1)h)\beta((m+1)h)}\big(1-2\theta c((m+1)h)+q_{0}(mh)\big)\notag\\
\leq&(1+hC_{6}C_{7}\alpha(mh)\beta(mh))(1-2\theta c((m+1)h)+q_{0}(mh))\notag\\
\leq & 1-q_{1}(mh), \ \forall \ m\geq \max\left\{m_{3},\lceil k_{2}h^{-1}\rceil\right\}, \notag
\end{align}
where
$
q_{1}(mh)=2\theta c((m+1)h)-q_{0}(mh)-hC_{6}C_{7}\alpha(mh)   \beta(mh)
 \big(2\theta c((m+1)h)-1-q_{0}(mh)\big).
$
This together with \eqref{pi11} leads to, for any $m\geq\max\{m_{0},m_{3},\lceil k_{2}h^{-1}\rceil\}$,
\begin{align}
E[\Pi((m+1)h)]\leq&(1-q_{1}(mh))E[\Pi(mh)]+\dfrac{c((m+1)h)
p(mh)}{\alpha((m+1)h)\beta((m+1)h)}. \label{pi2}
\end{align}
By Conditions (C1) and (C4), we have
$$\lim_{m\rightarrow\infty} \frac{\beta(mh) \alpha(mh)}{c((m+1)h)}  =0.$$
Then, $q_{0}(mh)+hC_{5}C_{7}  \alpha(mh)\beta(mh) \big(2\theta c((m+1)h)-1-q_{0}(mh)\big)=$
$o(c((m+1)h))$. Thus, by Conditions (C1) and (C4), there exists a positive integer $m_{4}$, such that
\begin{align}\label{qmh}
\sum\limits\limits_{m=0}^{\infty}q_{1}(mh)=\infty, \ 0<q_{1}(mh)\leq1,\ \forall\ m\geq m_{4}.
\end{align}
By Condition (C1),
we have  $\lim_{k\rightarrow \infty} \frac{\beta(k+h)}{\beta(k)}=1$. Then, from \eqref{limalpha0}, we obtain
$
\lim_{k\rightarrow\infty} \frac{\alpha(k)}{\alpha(k+h)}
=1.
$
This together with Conditions (C1)-(C2) leads to $$\lim_{m\rightarrow\infty}\frac{p(mh)}{\alpha((m+1)h)\beta((m+1)h)}
=C_{5}.$$
Then, by Condition (C4),  we have
\ban
\lim\limits_{m\rightarrow\infty}\frac{c((m+1)h)p(mh)}{\alpha((m+1)h)\beta((m+1)h)q_{1}(mh)}
=
\frac{C_{5}}{2\theta}.
\ean
This together with \eqref{pi2}, \eqref{qmh} and
Lemma 1.2.25 in \cite{Guo} leads to
\begin{align}\label{limpi1}
\limsup_{m\rightarrow\infty}E[\Pi(mh)]
\leq & \lim\limits_{m\rightarrow\infty}\frac{c((m+1)h)p(mh)}
{\alpha((m+1)h)\beta((m+1)h)q_{1}(mh)}
=\frac{C_{5}}{2\theta}.
\end{align}
By $C_{r}$ inequality and (\ref{Edd}), we get
\begin{align}
&2\alpha(k)E\left[\left\|d(k)\right\|\left\|\delta(k)\right\|\right]\notag\\
\leq& E\left[\left\|\alpha(k)d(k)\right\|^2\right]+E\left[\left\|\delta(k)\right\|^2\right]\notag\\
\leq& \alpha^2(k)E\left[2\sigma^2_{d}\left\|X(k)\right\|^2+2NC^2_{d}\right]+E[V(k)],\notag
\end{align}
which together with
\Cref{lemma2} (i) and \Cref{lemma3} leads to
\begin{align}
E[V(k+1)]\leq&2\big(1+c^2(k)(\rho^2_{0}+8\sigma^2C_{\xi}\rho_{1})\big)E[V(k)]
+2\alpha^2(k) C_{1}(2\sigma_{\zeta}+4\sigma^2_{d})\beta(k)+8b^2C_{\xi}  \rho_{1}c^2(k)\notag\\ &
+2\alpha^2(k)(2C_{\zeta}+4NC^2_{d}),\ k\geq 0.\notag
\end{align}
From \eqref{limpi1}, the above inequality, Conditions (C1)-(C3) and \Cref{lemma0}, we have
$$\limsup\limits_{k\rightarrow\infty}\dfrac{c(k+1)E[V(k+1)]}{\alpha(k+1)\beta(k+1)}
\leq\dfrac{\eta^hC_{5}}{2\theta}<\infty,$$
where $\eta=2\big(1+c^2(0) (\rho^2_{0} +8\sigma^2C_{\xi}\rho_{1})\big)$, which means that there exists a constant $C_{8}>0$, such that
$$
 E[V(k)]\leq C_{8}\frac{\alpha(k)\beta(k)}{c(k)}, \forall\  k\geq 0.
$$
 Thus, by Condition (C4) with $C=4C_{0}$, there exists a constant $C_{9}>0$, such that, for any $   k\geq 0$,  $$E[V(k)]\leq C_{8}\alpha(k)\beta^4(k)\beta^{-3}(k)c^{-1}(k)\leq C_{8}C_{9}\beta^{-3}(k),$$
that is, (i) holds.

Now, we will prove \Cref{lemmac} (ii).
By Conditions (C1) and (C4), Corollary 4.1.2 in \cite{YSChow} and \Cref{lemma3}, we have
\begin{align}
E\left[\sum_{k=0}^{\infty}\alpha^2(k)\|X(k)\|^2\right]=&
\sum_{k=0}^{\infty}\alpha^2(k)E\left[\left\|X(k)\right\|^2\right]\notag\\
\leq  & C_{1}\sum_{k=0}^{\infty}\alpha^2(k)\beta(k)<\infty.\notag
 \end{align}
By the non-negativity of $\sum_{k=0}^{\infty}\alpha^2(k)\|X(k)\|^2$, we have
\begin{align}
\sum_{k=0}^{\infty}\alpha^2(k)\|X(k)\|^2 <\infty \text{ a.s.}\label{EV(k+1)100}
\end{align}
By $\beta(k)\geq1$, H\"{o}lder inequality, (\ref{Edd}), Lemmas \ref{lemma3}-\ref{lemmac} (i), Corollary 4.1.2 in \cite{YSChow} and Condition (C3), we have
\begin{align}
&E\left[\sum_{k=0}^{\infty}2\alpha(k)\|d(k)\|\|\delta(k)\|\right]\notag\\
\leq & \sum_{k=0}^{\infty}2\alpha(k)\left[E\left[\|d(k)\|^2\right]\right]^{\frac{1}{2}} \left[E\left[\|\delta
(k)\|^2\right]\right]^{\frac{1}{2}}\notag\\
\leq &  \sum_{k=0}^{\infty} 2\alpha(k)\left[2NC^2_{d}+2\sigma^2_{d} E\left[\|X(k)\|^2\right]\right]^{\frac{1}{2}}
\left[E[V(k)]\right]^{\frac{1}{2}}\notag\\
\leq & \sum_{k=0}^{\infty}2\alpha(k)\left[2\sigma^2_{d}C_{1}\beta(k)+2NC^2_{d}\right]^{\frac{1}{2}}(\beta^{-3}(k))^{\frac{1}{2}}
C_{2}^{\frac{1}{2}}\notag\\
\leq & 2(2(\sigma^2_{d}C_{1}+NC^2_{d})C_{2})^{\frac{1}{2}}
\sum_{k=0}^{\infty}\alpha(k) \beta^{-1}(k)<\infty.\notag
\end{align}
Then, by the non-negativity of $\sum_{k=0}^{\infty}2\alpha(k)\|d(k)\|\|\delta(k)\|$, it follows that
\begin{align}
\sum_{k=0}^{\infty}2\alpha(k)\|d(k)\|\|\delta(k)\|<\infty \text{ a.s.}\notag
\end{align}
Then, by Condition (C1) and (\ref{EV(k+1)100}), we have
\begin{align}
&\sum_{k=0}^{\infty}\Big(  8b^2C_{\xi}\rho_{1}c^2(k)
+2\alpha^2(k)(2\sigma_{\zeta}+3\sigma^2_{d}) \|X(k)\|^2+2\alpha^2(k)
 (2C_{\zeta}+3NC^2_{d})\notag\\
 &
+2\alpha(k)\left\|d(k)\right\|\left\|\delta(k)\right\| \Big)<\infty \ \text{a.s.}\notag
\end{align}
Then,  by \Cref{lemma2} (i), Condition (C1) and  Theorem 1 in \cite{Robbins}, we obtain that $V(k)$ converges to a  random variable, $k\rightarrow\infty$ a.s., which together with (i) gives (ii).
\ $\hfill\blacksquare$

\emph{Proof of \Cref{corollary3.2}:}
From $\{\mathcal{G}(k),k\geq0\}\in\Gamma_{3}$, we know that $\{\mathcal{G}(k),k\geq0\}\in\Gamma_{1}$, and $E\left[\widehat{\mathcal{L}}_{\mathcal{G}(k)}\right]$ is positive semi-definite. By the independence of $\{\mathcal{G}(k),k\geq0\}$, we have $E[\mathcal{A}_{\mathcal{G}(k)}|\mathcal{F}(k-1)]=E[\mathcal{A}_{\mathcal{G}(k)}]$ and $E[\mathcal{L}_{\mathcal{G}(k)}|\mathcal{F}(k-1)]=E[\mathcal{L}_{\mathcal{G}(k)}]$. From $\delta(k)\in\mathcal{F}(k-1)$ and $(X(k)-\mathbf{1}_{N}\otimes x)\in\mathcal{F}(k-1)$, we have
\begin{align}
&E\left[\delta^T(k)\dfrac{\left(\mathcal{L}^T_{\mathcal{G}(k)}P^T
+P\mathcal{L}_{\mathcal{G}(k)}\right)\otimes I_{n}}{2}\delta(k)\Bigg|\mathcal{F}(k-1)\right]\cr
=&\delta^T(k)E\left[\dfrac{\left(\mathcal{L}^T_{\mathcal{G}(k)}P^T
+P\mathcal{L}_{\mathcal{G}(k)}\right)\otimes I_{n}}{2}\Bigg
|\mathcal{F}(k-1)\right]\delta(k)\cr
=&\delta^T(k)\dfrac{\Big(E\left[\mathcal{L}^T_{\mathcal{G}(k)}\right]
+E\left[\mathcal{L}_{\mathcal{G}(k)}\right]\Big)\otimes I_{n}}{2}\delta(k)
\geq0,
\end{align}
and
\begin{align}
&E\bigg[(X(k)-\mathbf{1}_{N}\otimes x)^T\frac{\left(\mathcal{L}^T_{\mathcal{G}(k)}+\mathcal{L}_{\mathcal{G}(k)}\right)\otimes I_{n}}{2}(X(k)-\mathbf{1}_{N}\otimes x)\Bigg|\mathcal{F}(k-1)\bigg]\cr
=&(X(k)-\mathbf{1}_{N}\otimes x)^T\left(E\left[\widehat{\mathcal{L}}_{\mathcal{G}(k)}|\mathcal{F}(k-1)\right]\otimes I_{n}\right) (X(k)-\mathbf{1}_{N}\otimes x)\cr
=&(X(k)-\mathbf{1}_{N}\otimes x)^TE\left[\widehat{\mathcal{L}}_{\mathcal{G}(k)}\otimes I_{n}\right](X(k)-\mathbf{1}_{N}\otimes x)
\geq0.
\end{align}
Then, similar to the proof of \Cref{lemma2}, \Cref{lemma2} still holds. From \Cref{lemma2} (ii), we can get \Cref{lemma3}. Let $\rho_{2}=\sup _{k \geq 0} \Big[E\left[\left\|\mathcal{L}_{\mathcal{G}(k)}\right\|^{2}\right]\Big]^{\frac{1}{2}}$. Since $\mathcal{L}_{\mathcal{G}(i)}$ is independent of $\mathcal{L}_{\mathcal{G}(j)}$, $i\neq j$, we do not have to use the conditional H\"{o}lder inequality. Here, by the Lyapunov inequality and Condition (iii), we have $\sup _{k \geq 0} E\left[\left\|\mathcal{L}_{\mathcal{G}(k)}\right\|\right]\leq \sup _{k \geq 0} \Big[E\left[\left\|\mathcal{L}_{\mathcal{G}(k)}\right\|^{2}\right]\Big]^{\frac{1}{2}}=\rho_{2}$. Then, similar to (A.9), we obtain
\begin{align}
&E\left[\left\|\Phi^T((m+1)h,mh)\Phi((m+1)h,mh)-I_{N}+\sum\limits_{i=mh}^{(m+1)h-1}c(i)(\mathcal{L}^T_{\mathcal{G}(i)}P^T+P\mathcal{L}_{\mathcal{G}(i)})
 \right\|\Bigg|\mathcal{F}(mh-1)\right]\cr
\leq&C_{4}c^2((m+1)h),\ m\geq m_{0},
\end{align}
where $C_{4}=C_{3}[(1+\rho_{2})^{2h}-1-2h\rho_{2}]$. Also, by the independence of $\{\mathcal{G}(k),k\geq0\}$ and Condition (ii), similar to (A.11) in \Cref{lemmac(i)}, we have
\begin{align}
&E\left[\delta^T(mh)\left[\sum\limits_{i=mh}^{(m+1)h-1}c(i)\left(P\mathcal{L}_{\mathcal{G}(i)}
+\mathcal{L}^T_{\mathcal{G}(i)}P\right)\otimes I_{n}\right]\delta(mh)\right]\cr
=&E\left[E\left[\delta^T(mh)\left[\sum\limits_{i=mh}^{(m+1)h-1}c(i)\left(P\mathcal{L}_{\mathcal{G}(i)}
+\mathcal{L}^T_{\mathcal{G}(i)}P\right)\otimes I_{n}\right] \delta(mh)\Bigg|\mathcal{F}(mh-1)\right]\right]\cr
=&2E\left[\delta^T(mh)\left[\sum\limits_{i=mh}^{(m+1)h-1}c(i)E\left[
\widehat{\mathcal{L}}_{
\mathcal{G}(i)}\otimes I_{n}\right]
\right]\delta(mh)\right]\cr
\geq&2c((m+1)h)E\left[\delta^T(mh)\left[\sum\limits_{i=mh}^{(m+1)h-1}E\left[\widehat{\mathcal{L}}_{\mathcal{G}(i)}\otimes I_{n}\right]\right]\delta(mh)\right]\cr
\geq&2c((m+1)h)\inf _{m \geq 0}\lambda_{2}\left(\sum_{i=m h}^{(m+1) h-1} E\left[\widehat{\mathcal{L}}_{\mathcal{G}(i)}\right]\right)E[V(mh)].
\end{align}
Then, similar to the proof of \Cref{lemmac}, \Cref{lemmac} still holds.
Then, similar to the proof of \Cref{theorem3.1}, we get \Cref{corollary3.2}.
$\hfill\blacksquare$

\vskip 0.2cm

To prove \Cref{strcoconver}, we need the following lemmas.

\begin{lemma}
\label{lemma3(iii)}
{For the  convex optimization problem \eqref{model}, consider the algorithm \eqref{algorithm}-\eqref{algorithm2}  with step sizes satisfying Conditions
(C1)-(C2).  If the local cost function $f_{i}(\cdot)$ in problem \eqref{model} is $\mu$-strongly convex, Assumptions \ref{subgradient}-\ref{intensity} and Assumption \ref{stochasticgraph} hold, and there exists a positive constant $\rho_{0}$, such that \\ $\sup_{k\geq0}\left[E\left[\|\mathcal{L}_{\mathcal{G}(k)}\|^2\Big|\mathcal{F}(k-1)
\right]\right]^{\frac{1}{2}}\leq\rho_{0}$\ a.s., then there exists a  positive integer $k_{0}$, such that
\begin{align}
\label{xboundedstr1}
\sup_{k\geq k_{0}}E\left[\left\|X(k)\right\|^2\right]\leq \frac{2NC_{d}^2}{\mu^2}+1.
\end{align}
Especially, if $c(k)=\frac{c_{0}}{(k+1)^{\gamma_{1}}}$, $\alpha(k)=\frac{\alpha_{0}}{(k+1)^{\gamma_{2}}}$, where $\gamma_{1}\in(0.5,1)$, $\gamma_{2}=1$, then
\begin{align}
\label{xboundedstr2}
\sup_{k\geq \lceil\alpha_{0}\mu\rceil-1}E\left[\left\|X(k)\right\|^2\right]\leq C_{X1},
\end{align}
where $C_{X1}=C_{\varphi_{1}}(\lceil\alpha_{0}\mu\rceil)^{\mu\alpha_{0}}
E\left[\|X(\lceil\alpha_{0}\mu\rceil-1)\|^2\right]
+C_{\varphi_{1}}NC_{d}^22^{\mu\alpha_{0}+2}\mu^{-2}
+C_{\varphi_{1}}C_{\varphi_{3}}$,
$C_{\varphi_{1}}=\exp\big(\max\{2(\rho^2_{0}+8\sigma^2 C_{\xi}\rho_{1}),
4(2\sigma_{\zeta}+3\sigma^2_{d})\}\sum_{k=0}^{\infty}(c^2(k)+\alpha^2(k))\big)$  and $C_{\varphi_{3}}=\max\{8b^2C_{\xi}\rho_{1},2\big(2C_{\zeta}+3NC^2_{d}\big)\}\sum_{k=0}^{\infty}(c^2(k)+\alpha^2(k))$.}
\end{lemma}

\vskip 0.2cm

\begin{proof}
{Noting that $f_{i}(\cdot)$ is $\mu$-strongly convex, by $d_{f_{i}}(k)\in\partial f_{i}(x_{i}(k))$ and \Cref{subgradient}, we have
\begin{align}
&-d^T_{f_{i}}(k)(x_{i}(k)-x)\notag\\
\leq&f_{i}(x)-f_{i}
(x_{i}(k))-\frac{\mu}{2}\|x_{i}(k)-x\|^2\cr
\leq& d^T_{f_{i}}(x)(x-x_{i}(k))-\frac{\mu}{2}\|x-x_{i}(k)\|^2-\frac{\mu}{2}\|x_{i}(k)-x\|^2\cr
\leq& \frac{1}{2\tau}d^T_{f_{i}}(x)d_{f_{i}}(x)+\frac{\tau}{2}(x-x_{i}(k))^T(x-x_{i}(k))-\mu\|x_{i}(k)-x\|^2\cr
\leq& \frac{1}{\tau}\left(\sigma^2_{d}\|x\|^2+C_{d}^2\right)-\frac{1}{2}(2\mu-\tau)\|x_{i}(k)-x\|^2,\ \forall\  x\in \mathbb{R}^n,\ \forall\ \tau\in(0,2\mu),\ i=1,\ldots,N,
\end{align}
where the third $"\leq"$ is obtained from the inequality $p^Tq\leq\frac{1}{2\tau}\|p\|^2+\frac{\tau}{2}\|q\|^2,\ \forall\ 0< \tau <2\mu$. Thus,
\begin{align}
&-2\alpha(k)E\left[d^T(k)(X(k)-\mathbf{1}_{N}\otimes x)\right]\cr
=&-2\alpha(k)E\left[\sum\limits_{i=1}^{N}d^T_{f_{i}}(k)(x_{i}(k)-x)\right]\cr
\leq&2\alpha(k)E\left[\sum\limits_{i=1}^{N}
\left[\frac{1}{\tau}\left(\sigma^2_{d}\|x\|^2+C_{d}^2\right)-\frac{1}{2}(2\mu-
\tau)\|x_{i}(k)-x\|^2\right]\right]\cr
\leq&\alpha(k)\dfrac{2N}{\tau}\left(\sigma^2_{d}\|x\|^2+C_{d}^2\right)
-\alpha(k)(2\mu-\tau)E\left[\left\|X(k)-\mathbf{1}_{N}\otimes x\right\|^2\right],
\end{align}
which together with \Cref{lemma2} (ii) and $x=\mathbf{0}_{n}$ gives
\begin{align}\label{Exkxx}
E\left[\left\|X(k+1)\right\|^2\right]
\leq&\left(1-\alpha(k)(2\mu-\tau)+2c^2(k)\left(\rho^2_{0}+8\sigma^2C_{\xi}\rho_{1}\right)
+4\alpha^2(k)\left(2\sigma_{\zeta}+3\sigma^2_{d}\right)\right)E\left[\left\|X(k)
\right\|^2\right]\notag\\
&+8b^2C_{\xi}\rho_{1}c^2(k)+2\alpha^2(k)\left(2C_{\zeta}+3NC^2_{d}\right)
+\alpha(k)\dfrac{2N}{\tau}C_{d}^2,\ \forall\ k\geq0.
\end{align}}
{Then, we will prove \eqref{xboundedstr1}. From Condition (C2), we have
\begin{align}\label{Exkxx4}
\lim\limits_{k\rightarrow\infty}\frac{8b^2C_{\xi}\rho_{1}c^2(k)+2\alpha^2(k)\big(2C_{\zeta}+3NC^2_{d}\big)
+\alpha(k)\dfrac{2N}{\tau}C_{d}^2}
{\alpha(k)(2\mu-\tau)-2c^2(k)(\rho^2_{0}+8\sigma^2C_{\xi}\rho_{1})
-4\alpha^2(k)(2\sigma_{\zeta}+3\sigma^2_{d})}
=\frac{2NC_{d}^2}{\tau(2\mu-\tau)}.
\end{align}
By $(2\mu-\tau)>0$, Conditions (C1) and (C2), there exists a positive integer $k_{1}$, such that $0<\Big(\alpha(k)(2\mu-\tau)-2c^2(k)(\rho^2_{0}+8\sigma^2C_{\xi}\rho_{1})
-4\alpha^2(k)(2\sigma_{\zeta}+3\sigma^2_{d})\Big)\leq 1,\ \forall\  k\geq k_{1}$
and $\sum\limits_{k=0}^{\infty}\Big(\alpha(k)(2\mu-\tau)-2c^2(k)(\rho^2_{0}+8\sigma^2C_{\xi}\rho_{1})
-4\alpha^2(k)$ $(2\sigma_{\zeta}+3\sigma^2_{d})\Big)=\infty$.
Then by Lemma 1.2.25 in \cite{Guo}, (\ref{Exkxx}) and (\ref{Exkxx4}) , it follows that $\limsup_{k\rightarrow\infty}E\left[\left\|X(k)\right\|^2\right]\leq\dfrac{2NC_{d}^2}{\tau(2\mu-\tau)}.$
Let $\tau=\mu$, and then we get \eqref{xboundedstr1}.}

{In the following part, we will prove \eqref{xboundedstr2}.
From \eqref{Exkxx} with $\tau=\mu$, we have
\begin{align}\label{Exkxx1}
E\left[\left\|X(k+1)\right\|^2\right]
\leq\left(1-\alpha(k)\mu+\varphi_{1}(k)\right)E\left[\left\|X(k)\right\|^2\right]
+\alpha(k)\dfrac{2N}{\mu}C_{d}^2+\varphi_{3}(k),\ \forall\ k\geq0,
\end{align}
where $\varphi_{1}(k)=2c^2(k)(\rho^2_{0}+8\sigma^2C_{\xi}\rho_{1})
+4\alpha^2(k)\left(2\sigma_{\zeta}+3\sigma^2_{d}\right)$ and $\varphi_{3}(k)=8b^2C_{\xi}\rho_{1}c^2(k)+2\alpha^2(k)\big(2C_{\zeta}+3NC^2_{d}\big)$.
Let $T_{0}=\lceil\alpha_{0}\mu\rceil-1$, then $1-\alpha(k)\mu+\varphi_{1}(k)>0,\ \forall\ k\geq T_{0}$. Thus, from \eqref{Exkxx1}, we have
\begin{align}
\label{Exx}
&E\left[\left\|X(k+1)\right\|^2\right]\cr
\leq&\prod\limits_{t=T_{0}}^{k}\left(1-\alpha(t)\mu+
\varphi_{1}(t)\right)E\left[\|X(T_{0})\|^2\right]+
\sum_{s=T_{0}}^{k}\left(\alpha(s)\dfrac{2N}{\mu}C_{d}^2
+\varphi_{3}(s)\right)\prod\limits_{t=s+1}^{k}
\Big(1-\alpha(t)\mu+\varphi_{1}(t)\Big)\cr
\leq&\exp\left(\sum_{t=T_{0}}^{k}\varphi_{1}(t)\right)
\exp\left(\sum_{t=T_{0}}^{k}-\alpha(t)\mu\right)E\left[\left\|X(T_{0})\right\|^2\right]\cr
&+\sum_{s=T_{0}}^{k}\left(\alpha(s)\dfrac{2N}{\mu}
C_{d}^2+\varphi_{3}(s)\right)\exp\left(\sum_{t=s+1}^{k}\varphi_{1}(t)\right)\exp\left(\sum_{t=s+1}^{k}-\alpha(t)\mu\right)\cr
\leq&C_{\varphi_{1}}\exp\left(\sum_{t=T_{0}}^{k}-\alpha(t)\mu
\right)E\left[\left\|X(T_{0})\right\|^2\right]\cr
&+\sum_{s=T_{0}}^{k}\left(\alpha(s)\dfrac{2N}{\mu}C_{d}^2+\varphi_{3}(s)\right)C_{\varphi_{1}}\exp\left(\sum_{t=s+1}^{k}-\alpha(t)\mu\right),\ k\geq T_{0}.
\end{align}
Since $\alpha(k)=\frac{\alpha_{0}}{k+1}$, we have $\sum_{t=T_{0}}^{k}\alpha(t)\geq\alpha_{0}(\ln(k+2)-\ln(T_{0}+1))$, which together with \eqref{Exx} gives
\begin{align}
\label{Exx1}
&E\left[\left\|X(k+1)\right\|^2\right]\cr
\leq&C_{\varphi_{1}}\exp\left(-\mu\alpha_{0}(\ln(k+2)-\ln(T_{0}+1))\right)
E\left[\left\|X(T_{0})\right\|^2\right]\cr
&+\sum_{s=T_{0}}^{k}\left(\alpha(s)\frac{2N}{\mu}C_{d}^2
+\varphi_{3}(s)\right)C_{\varphi_{1}}\exp\left(-\mu\alpha_{0}(\ln(k+2)-\ln(s+2))\right)\cr
\leq&C_{\varphi_{1}}(T_{0}+1)^{\mu\alpha_{0}}E\left[\left\|X(T_{0})\right\|^2\right](k+2)^{-\mu\alpha_{0}}\cr
&+C_{\varphi_{1}}\frac{2N}{\mu}C_{d}^2(k+2)^{-\mu\alpha_{0}}\sum_{s=T_{0}}^{k}\frac{\alpha_{0}(s+2)^{\mu\alpha_{0}}}{s+1}
+C_{\varphi_{1}}C_{\varphi_{3}}\cr
\leq&C_{\varphi_{1}}(T_{0}+1)^{\mu\alpha_{0}}E\left[\left\|X(T_{0})\right\|^2\right](k+2)^{-\mu\alpha_{0}}\cr
&+C_{\varphi_{1}}\frac{2N}{\mu}C_{d}^22\alpha_{0}(k+2)^{-\mu\alpha_{0}}\sum_{s=T_{0}}^{k}(s+2)^{\mu\alpha_{0}-1}
+C_{\varphi_{1}}C_{\varphi_{3}},\ k\geq T_{0},
\end{align}
where the $2$nd ``$\leq$" is by $(k+2)^{-\mu\alpha_{0}}\sum_{s=T_{0}}^{k}\varphi_{3}(s)(s+2)^{\mu\alpha_{0}}\leq$
 $\sum_{s=0}^{\infty}\varphi_{3}(s)$ and the last ``$\leq$" is by $(k+2)^{-\mu\alpha_{0}}\sum_{s=T_{0}}^{k}\frac{(s+2)^{\mu\alpha_{0}}}{s+1}
=(k+2)^{-\mu\alpha_{0}}\sum_{s=T_{0}}^{k}(s+2)^{\mu\alpha_{0}-1}
(1+\frac{1}{s+1})
\leq2(k+2)^{-\mu\alpha_{0}}\sum_{s=T_{0}}^{k}(s+2)^{\mu\alpha_{0}-1}$.
If $\mu\alpha_{0}=1$, then
\begin{align}\label{mu1}
(k+2)^{-\mu\alpha_{0}}\sum_{s=T_{0}}^{k}(s+2)^{\mu\alpha_{0}-1}=\frac{k-T_{0}}{k+2}\leq 1.
\end{align}
If $\mu\alpha_{0}\neq1$, then
\begin{align}\label{mu2}
(k+2)^{-\mu\alpha_{0}}\sum_{s=T_{0}}^{k}(s+2)^{\mu\alpha_{0}-1}\leq&(k+2)^{-\mu\alpha_{0}}\int_{T_{0}-1}^{k+1}(s+2)^{\mu\alpha_{0}-1}ds\cr
=&(k+2)^{-\mu\alpha_{0}}\frac{1}{\mu\alpha_{0}}\left((k+3)^{\mu\alpha_{0}}-(T_{0}+1)^{\mu\alpha_{0}}\right)\cr
\leq&\frac{1}{\mu\alpha_{0}}\left(\frac{k+3}{k+2}
\right)^{\mu\alpha_{0}}
\leq\frac{2^{\mu\alpha_{0}}}{\mu\alpha_{0}}.
\end{align}
From \eqref{Exx1}-\eqref{mu2}, we have \eqref{xboundedstr2}.
}
\end{proof}

\begin{lemma}\label{lemma2(iii)}
{For the convex optimization problem \eqref{model}, consider the algorithm \eqref{algorithm}-\eqref{algorithm2}. Suppose that the local cost function $f_{i}(\cdot)$ in problem \eqref{model} is $\mu$-strongly convex, Assumptions \ref{subgradient}-\ref{intensity} and Assumption \ref{stochasticgraph} hold, there exists a positive integer $T$ and a positive constant $C_{X}$, such that $\sup_{k\geq T}E\left[\left\|X(k)\right\|^{2}\right]\leq C_{X}$, and there exists a positive constant $\rho_{0}$, such that $\sup\limits_{k\geq0}\left[E\left[\left\|\mathcal{L}_{\mathcal{G}(k)}
\right\|^2\Big|\mathcal{F}(k-1)\right]\right]^{\frac{1}{2}}\leq\rho_{0}$ a.s.  Then,
\begin{align}\label{iii}
E\left[\left\|X(k+1)-\mathbf{1}_{N}\otimes z^{*}\right\|^2\right]
\leq&\Big(1-\frac{\mu}{2N}\alpha(k)+\varphi_{1}(k)\Big)E\left[\left\|X(k)-\mathbf{1}_{N}
\otimes z^{*}\right\|^2\right]+\varphi_{2}(k),\ k\geq  T,
\end{align}
where $z^{*}$ is the unique optimal solution to the problem \eqref{model}, and
\begin{align}
\varphi_{1}(k)=&2c^2(k)(\rho^2_{0}+8\sigma^2C_{\xi}\rho_{1})
+4\alpha^2(k)\left(2\sigma_{\zeta}+3\sigma^2_{d}\right),\notag\\
\varphi_{2}(k)=&8b^2C_{\xi}\rho_{1}c^2(k)+2\alpha^2(k)\big(2C_{\zeta}+3NC^2_{d}+2(3\sigma^2_{d}
+2\sigma_{\zeta})N\|z^{*}\|^2\big)\notag\\
&+\frac{\mu \alpha(k)}{N}E[V(k)]+2\sqrt{2\left(\sigma^2_{d}C_{X}+NC^2_{d}\right)}
\alpha(k)\sqrt{E[V(k)]}.\notag
\end{align}
}
\end{lemma}

\vskip 0.2cm

\begin{proof}
{From \Cref{Edd1},  \Cref{dx} and  H\"{o}lder inequality, we have
\begin{align}
\notag
&-2\alpha(k)E\left[d^T(k)\left(X(k)-\mathbf{1}_{N}\otimes z^{*}\right)\right]\cr
\leq&2\alpha(k)E[f^{*}-f(\bar{x}(k))]+2\alpha(k)\sqrt{N}\left[E\left[(\sigma_{d}\|\bar{x}(k)\|
+C_{d})^2\right]\right]^{\frac{1}{2}}\left[E\left[\left\|\delta(k)\right\|^2\right]\right]^{\frac{1}{2}}\cr
\leq&2\alpha(k)E[f^{*}-f(\bar{x}(k))]+2\sqrt{2\left(\sigma^2_{d}C_{X}+NC^2_{d}\right)}
\alpha(k)\sqrt{E[V(k)]}, \ k\geq T,
\end{align}
 which together with \Cref{lemma2} (ii) leads to
\begin{align}\label{Exkxstar}
&E\left[\left\|X(k+1)-\mathbf{1}_{N}\otimes z^{*}\right\|^2\right]\cr
\leq&\left(1+2c^2(k)\left(\rho^2_{0}+8\sigma^2C_{\xi}\rho_{1}\right)
+4\alpha^2(k)(2\sigma_{\zeta}+3\sigma^2_{d})\right)E\left[\left\|X(k)-\mathbf{1}_{N}\otimes z^{*}\right\|^2\right]\cr
&+8b^2C_{\xi}\rho_{1}c^2(k)+2\alpha^2(k)\big(2C_{\zeta}+3NC^2_{d}+2(3\sigma^2_{d}+2\sigma_{\zeta})N\|z^{*}\|^2\big)\cr
&-2\alpha(k)E[f(\bar{x}(k))-f^{*}]+2\sqrt{2\left(\sigma^2_{d}C_{X}+NC^2_{d}\right)}
\alpha(k)\sqrt{E[V(k)]},\ k\geq T.
\end{align}
Noting that the global cost function $f(\cdot)$ is $\mu$-strongly convex, we derive
\begin{align}\label{ffstar}
-2\alpha(k)E[f(\bar{x}(k))-f^{*}]\leq-\mu\alpha(k)E\left[\left\|\bar{x}(k)
-z^{*}\right\|^2\right].
\end{align}
From the inequality $\|x+y\|^2\leq2\|x\|^2+2\|y\|^2$, we have
\begin{align}
E\left[\left\|X(k)-\mathbf{1}_{N}\otimes z^{*}\right\|^2\right]
=&E\left[\left\|X(k)-\mathbf{1}_{N}\otimes \bar{x}(k)+\mathbf{1}_{N}\otimes \bar{x}(k)-\mathbf{1}_{N}\otimes z^{*}\right\|^2\right]\cr
\leq&2E\left[\left\|X(k)-\mathbf{1}_{N}\otimes \bar{x}(k)\right\|^2\right]+2E\left[\left\|\mathbf{1}_{N}\otimes \bar{x}(k)-\mathbf{1}_{N}\otimes z^{*}\right\|^2\right]\cr
=&2E[V(k)]+2NE\left[\left\|\bar{x}(k)-z^{*}\right\|^2\right].\notag
\end{align}
From the above inequality, we get
\begin{align}
\notag
-E\left[\left\|\bar{x}(k)-z^{*}\right\|^2\right]\leq-\frac{1}{2N}
E\left[\left\|X(k)-\mathbf{1}_{N}\otimes z^{*}\right\|^2\right]+\frac{1}{N}E[V(k)].
\end{align}
By \eqref{Exkxstar}, \eqref{ffstar} and the above inequality, we get \eqref{iii}.}
\end{proof}

\begin{lemma}\label{lemmac(ii1)}
{For the convex optimization problem \eqref{model}, consider the algorithm \eqref{algorithm}-\eqref{algorithm2}.  Suppose that the local cost function $f_{i}(\cdot)$ in problem \eqref{model} is $\mu$-strongly convex and assume that}

{(a) Assumptions \ref{subgradient}-\ref{intensity} and Assumption \ref{stochasticgraph} hold,  Conditions (C1-(C2) and (C3)'-(C5)' hold, and the local cost functions $f_{i}(\cdot),i=1,\ldots,N$ are $\mu$-strongly convex;}

{(b) there exists a positive integer $h$ and positive constants $\theta$ and $\rho_{0}$, such that}

{\ \  (b.1) $\inf_{m\geq0}\lambda_{mh}^h\geq\theta$\  a.s.;}

{\ \  (b.2) $\sup_{k\geq0}\left[E\left[\left\|\mathcal{L}_{\mathcal{G}(k)}\right
\|^{2^{\max\{h,2\}}}\Big|
\mathcal{F}(k-1)\right]\right]^{\frac{1}{2^{\max\{h,2\}}}}\leq\rho_{0}$\  a.s.;\\}
{
 then there exists a constant $\hat{C}_{12}>0$ such that
\begin{align}\label{EVkorder2str211}
E[V(k)]\leq \hat{C}_{12} \frac{\alpha(k)}{c(k)}, \ \forall \ k \geq 0.
\end{align}
}
\end{lemma}
\begin{proof}
From \Cref{lemmac(i)} with $\tau(mh)=\alpha(mh)$, \Cref{lemma3(iii)} and the monotone property of $c(k),\alpha(k)$, we get
\begin{align}\label{Vmh2cstr1}
&E[V((m+1)h)]\cr
\leq&\Big[1-2\theta c((m+1)h)+\hat{q}_{0}(mh)\Big]E[V(mh)]+\hat{p}(mh),\ m\geq\max\left\{m_{0},\left\lceil\frac{k_{0}}{h}\right\rceil\right\},
\end{align}
where
$\hat{q}_{0}(mh)=C_{4} c^2((m+1)h)+\alpha(mh)\big(1-2\theta c((m+1)h)+C_{4} c^2((m+1)h)\big),$
$\hat{p}(mh)=\hat{C}_{6}\alpha(mh)+(2\hat{C}_{6}+\hat{C}_{7})\alpha^2(mh)
+\hat{C}_{8}c^2(mh),$
$\hat{C}_{6}=2h^2C_{\rho}\left(\sigma^2_{d}\left(\frac{2NC_{d}^2}{\mu^2}+1\right)
+NC^2_{d}\right),$
$\hat{C}_{7}=4h^2C_{\rho}\bigg(\sigma_{\zeta}\bigg(\frac{2NC_{d}^2}{\mu^2}+1\bigg)
+C_{\zeta}\bigg),$
$\hat{C}_{8}=4h^2\rho_{1}
C_{\rho}C_{\xi}(4\sigma^2\left(\frac{2NC_{d}^2}{\mu^2}+1\right)+2b^2),$ $C_{4}$ is given by \Cref{lemmac(i)}.
From Conditions (C1) and (C4)', we obtain $\hat{q}_{0}(mh)=o(c((m+1)h))$, thus, there exists a positive integer $\hat{m}_{3}$, such that
\begin{equation}\label{hatq1}
0<2\theta c((m+1)h)-q_{0}(mh)\leq1,\ m\geq \hat{m}_{3}.
\end{equation}
Let $\widetilde{\Pi}(k)=\frac{c(k)V(k)}{\alpha(k)}$.
From Conditions (C1), (C5)' and (\ref{Vmh2cstr1}), we obtain
\begin{align}
\label{pi11strong31}
&E\left[\widetilde{\Pi}((m+1)h)\right]\notag\\
=&E\left[\frac{c((m+1)h)V((m+1)h)}
{\alpha((m+1)h)}\right]\notag\\
\leq & \frac{\alpha(mh)}{\alpha((m+1)h)}\Big[1-2\theta c((m+1)h)+\hat{q}_{0}(mh)\Big]E\left[\widetilde{\Pi}(mh)\right]
+\frac{c((m+1)h)\hat{p}(mh)}{\alpha((m+1)h)}\notag\\
=&\big(1+o(c((m+1)h))\big)\Big[1-2\theta c((m+1)h)+\hat{q}_{0}(mh)\Big]E\left[\widetilde{\Pi}(mh)\right]
+\frac{c((m+1)h)\hat{p}(mh)}{\alpha((m+1)h)}\notag\\
=&\Big[1-2\theta c((m+1)h)+\hat{q}_{2}(mh)\Big]E\left[\widetilde{\Pi}(mh)\right]
+\frac{c((m+1)h)\hat{p}(mh)}{\alpha((m+1)h)},\ m\geq\max\left\{m_{0},\hat{m}_{3},\left\lceil\dfrac{k_{0}}{h}\right\rceil\right\},
\end{align}
where $\hat{q}_{2}(mh)=\hat{q}_{0}(mh)+o(c((m+1)h))\big(1-2\theta c((m+1)h)+\hat{q}_{0}(mh)\big)$. Then $\hat{q}_{2}(mh)=o(c((m+1)h))$. Thus, by Condition (C1), there exists a positive integer $\hat{m}_{2}$, such that
\begin{align}\label{hatq2}
&&\hspace{-1.7cm}0<\hat{q}_{1}(mh)\leq1,\ \forall\ m\geq \hat{m}_{2},\
\text{and}
\sum\limits\limits_{m=0}^{\infty}\hat{q}_{1}(mh)=\infty,
\end{align}
where $\hat{q}_{1}(mh)=2\theta c((m+1)h)-\hat{q}_{2}(mh).$ By conditions (C1) and (C4)', we have $\lim_{m\to \infty} \frac{\alpha(mh)}{\alpha((m+1)h)}=1+\lim_{m\to \infty} \frac{\alpha(mh)-\alpha((m+1)h)}{\alpha((m+1)h)}=1.$ This together with Conditions (C1) and (C2) gives
\begin{align}
\notag
&\lim\limits_{m\rightarrow\infty}\frac{\hat{p}(mh)c((m+1)h)}{\alpha((m+1)h)\hat{q}_{1}(mh)}\notag\\
=&\lim\limits_{m\rightarrow\infty}\frac{c((m+1)h)}{\hat{q}_{1}(mh)}
\lim\limits_{m\rightarrow\infty}\frac{\hat{C}_{6}\alpha(mh)+(2\hat{C}_{6}+\hat{C}_{7})\alpha^2(mh)
+\hat{C}_{8}c^2(mh)}{\alpha((m+1)h)}=\frac{\hat{C}_{6}}{2\theta}.
\end{align}
This together with \Cref{lemma0}, similar to the proof of \Cref{lemmac}, gives (\ref{EVkorder2str211}).

\end{proof}

\emph{Proof of \Cref{strcoconver}:}
 By \Cref{lemma3(iii)}- \Cref{lemmac(ii1)}, there exists $k_{0},\ \hat{C}_{12}$, such that
\begin{align}\label{meanstrongconvex}
E\left[\left\|X(k+1)-\mathbf{1}_{N}\otimes z^{*}\right\|^2\right]
\leq\Big(1-\frac{\mu}{2N}\alpha(k)+\hat{\varphi}_{1}(k)\Big)E\left[\|X(k)
-\mathbf{1}_{N}\otimes z^{*}\|^2\right]+\hat{\varphi}_{2}(k),\ k\geq  k_{0},
\end{align}
where $\hat{\varphi}_{1}(k)=2c^2(k)(\rho^2_{0}+8\sigma^2C_{\xi}\rho_{1})
+4\alpha^2(k)(2\sigma_{\zeta}+3\sigma^2_{d})$, $\hat{\varphi}_{2}(k)=8b^2C_{\xi}\rho_{1}c^2(k)+2\alpha^2(k)\big(2C_{\zeta}
+3NC^2_{d}\\ +2(3\sigma^2_{d}
+2\sigma_{\zeta})N\|z^{*}\|^2\big)
+\frac{\mu \alpha(k)}{N} \hat{C}_{12} \frac{\alpha(k)}{c(k)}
+2\sqrt{2\left(\sigma^2_{d}C_{X}+NC^2_{d}\right)}\hat{C}_{12}^{\frac{1}{2}}
\alpha^{\frac{3}{2}}(k)c^{-\frac{1}{2}}(k).$
By Conditions (C1)-(C2) and  (C3)', there exists $\hat{k}_{0}$, such that $0<\frac{\mu}{2N}\alpha(k)-\hat{\varphi}_{1}(k)<1, \ \forall \ k \geq \hat{k}_{0}$, $\ \sum_{k=0}^{\infty}\Big(\frac{\mu}{2N}\alpha(k)-\hat{\varphi}_{1}(k)\Big)=\infty$ and $\lim_{k \to \infty}\frac{\hat{\varphi}_{2}(k)}{\frac{\mu}{2N}\alpha(k)-\hat{\varphi}_{1}(k)}=0$. Then, by (\ref{meanstrongconvex}) and Lemma 1.2.25 in \cite{Guo} , we have $\lim_{k \to \infty}E\big[\|X(k)-\mathbf{1}_{N}\otimes z^{*}\|^2\big]=0.$ That is, (i) holds.


Similar to the proof of (\ref{itetooptimal1}) in \Cref{theorem3.1}, we have
\begin{align}
&E\left[\left\|X(k+1)-\mathbf{1}_{N}\otimes z^{*}\right\|^2\big|\mathcal{F}(k-1)\right]\cr
\leq&\Big(1+2c^2(k)(\rho^2_{0}+8\sigma^2C_{\xi}\rho_{1})
+4\alpha^2(k)(2\sigma_{\zeta}+3\sigma^2_{d})\Big)\|X(k)-\mathbf{1}_{N}\otimes z^{*}\|^2\cr
&+8b^2C_{\xi}\rho_{1}c^2(k)+2\alpha^2(k)\big(2C_{\zeta}+3NC^2_{d}+2(3\sigma^2_{d}
+2\sigma_{\zeta})N\| z^{*}\|^2\big)\cr
&-2\alpha(k)(f(\bar{x}(k))-f(z^{*}))+2\alpha(k)\sqrt{N}(\sigma_{d}\|\bar{x}(k)
\|+C_{d})\|\delta(k)\| \ \text{a.s.}\label{iteoptimscvez}
\end{align}
 From $\|p+q\|^2\leq 2\|p\|^2+2\|q\|^2$, $p, q\in\mathbb{R}^n$,  \Cref{lemma3(iii)}, Conditions (C1) and (C2), we get
\begin{align}
E\left[(\sigma_{d}\|\bar{x}(k)\|+C_{d})^2\right]
\leq& 2\sup_{k\geq0}E\left[\sigma^2_{d}\|\bar{x}(k)\|^2+C^2_{d}\right]\cr
\leq& 2\sup_{k\geq0}E\left[\frac{1}{N}\sigma^2_{d}\sum\limits\limits_{i=1}^N\|x_{i}(k)\|^2+C^2_{d}\right]\cr
\leq& \frac{2}{N}\sigma^2_{d}\left(\frac{2NC_{d}^2}{\mu^2}+1\right)+2C^2_{d},  \ k \geq k_{0},
\end{align}
where $k_{0}$ is given by  \Cref{lemma3(iii)}.
 By Conditions (C1)-(C2), (C4)'-(C5)' and  \Cref{lemmac(ii1)}, there exists a positive number $\hat{C}_{12}>0$ such that $E[V(k)]\leq \hat{C}_{12} \frac{\alpha(k)}{c(k)},\ \forall\ k \geq0.$
Then, from the above inequality and  H\"{o}lder inequality, we have
\begin{align}
\label{barxalphack}
E\left[(\sigma_{d}\|\bar{x}(k)\|+C_{d})\|\delta(k)\|\right]
\leq&\left[E\left[(\sigma_{d}\|\bar{x}(k)\|+C_{d})^2\right]\right]^{\frac{1}{2}}\left
[E\left[\|\delta(k)\|^2\right]\right]^{\frac{1}{2}}\cr
\leq& \sqrt{\frac{2}{N}\sigma^2_{d}\left(\frac{2NC_{d}^2}{\mu^2}+1\right)+2C^2_{d}}\sqrt{E[V(k)]}\cr
\leq&\sqrt{\frac{2}{N}\sigma^2_{d}\left(\frac{2NC_{d}^2}{\mu^2}+1\right)
+2C^2_{d}}
\hat{C}_{12}^{\frac{1}{2}}\frac{\alpha^{\frac{1}{2}}(k)}{c^{\frac{1}{2}}(k)},\ \ k \geq k_{0}.
\end{align}
From  Condition (C3)', the above inequality and Corollary 4.1.2 in \cite{YSChow}, we have
\begin{align}
&E\left[\sum_{k=0}^{\infty} \alpha(k)(\sigma_{d}\|\bar{x}(k)
\|+C_{d})\|\delta(k)\|\right]\notag\\
=&\sum_{k=0}^{\infty}\alpha(k)E\left[ (\sigma_{d}\|\bar{x}(k)
\|+C_{d})\|\delta(k)\|\right]\notag\\
\leqslant & \sum_{k=0}^{k_{0}}\alpha(k)E\left[ (\sigma_{d}\|\bar{x}(k)
\|+C_{d})\|\delta(k)\|\right]+ \sqrt{\frac{2}{N}\sigma^2_{d}\left(\frac{2NC_{d}^2}{\mu^2}+1\right)
+2C^2_{d}}
\hat{C}_{12}^{\frac{1}{2}}\sum_{k=k_{0}}^{\infty}\frac{\alpha^{\frac{3}{2}}(k)}
{c^{\frac{1}{2}}(k)}<\infty.\notag
\end{align}
By the non-negativity of $\sum_{k=0}^{\infty} \alpha(k)(\sigma_{d}\|\bar{x}(k)
\|+C_{d})\|\delta(k)\|$, then we have
\begin{align}
 \sum_{k=0}^{\infty} \alpha(k)(\sigma_{d}\|\bar{x}(k)
\|+C_{d})\|\delta(k)\|<\infty \text{ a.s.}\label{sumalphakbarxk1}
\end{align}
By Condition (C1), we have $\sum_{k=0}^{\infty} \left(8b^2C_{\xi}\rho_{1}c^2(k)+2\alpha^2(k)\big(2C_{\zeta}+3NC^2_{d}+2(3\sigma^2_{d}
+2\sigma_{\zeta})N\| z^{*}\|^2\big)\right) <\infty$.
This together with (\ref{sumalphakbarxk1}) leads to $\sum_{k=0}^{\infty} \Big(8b^2C_{\xi}\rho_{1}c^2(k)+2\alpha^2(k)\big(2C_{\zeta}+3NC^2_{d}+2(3\sigma^2_{d}
+2\sigma_{\zeta})N\| z^{*}\|^2\big)+2\alpha(k)\sqrt{N}(\sigma_{d}\|\bar{x}(k)
\|+C_{d})\|\delta(k)\|\Big) <\infty$ a.s. Then, Noting that
 $f(\bar{x}(k))-f^{*}\geq0$ and by
Theorem 1 in \cite{Robbins}, Condition (C1) and (\ref{iteoptimscvez}), we obtain that the sequence $\{\|X(k)-\mathbf{1}_{N}\otimes z^{*}\|^2,k\geq0\}$ converges a.s. This together with (i) leads to (ii).
$\hfill\blacksquare$

\section{supplementary Lemmas}
\label{appendix:lemmas}






\begin{lemma}(Integral test)\label{lemma2.1} (\cite{ApostolTA})
Let $f(x)$ be a positive decreasing function defined on $[0,\infty)$ such that $\lim\limits _{x \rightarrow\infty} f(x)=0$.  For $n=1,2, \ldots$, define
\begin{equation*}
s_{n}=\sum_{k=0}^{n} f(k), \quad t_{n}=\int_{0}^{n} f(x) d x, \quad d_{n}=s_{n}-t_{n}.
\end{equation*}
Then we have\\
${\text { i) } 0<f(n+1) \leq d_{n+1} \leq d_{n} \leq f(0),\   n=1,2, \ldots;}$\\
${\text { ii) } \lim\limits _{n \rightarrow \infty} d_{n} \text { exists; }}$\\
${\text { iii) } \left\{s_{n}\right\} \text{converges if and only if the sequence}\ \left\{t_{n}\right\}  \text{converges;}\ }$\\
${\text { iv) } 0 \leq d_{k}-\lim\limits _{n \rightarrow \infty} d_{n} \leq f(k),\ k=1,2, \ldots}$
\end{lemma}

\vskip0.2cm

%


\vskip0.2cm

\begin{lemma}
\label{lemma0}
{Let $\{\Phi(k),k\geq0\}$, $\{\tilde{c}_{1}(k),k\geq0\}$, and $\{\tilde{c}_{2}(k),k\geq0\}$ be real nonnegative sequences satisfying $\Phi(k+1)\leq C_{\Phi 1}\Phi(k)+C_{\Phi 2}\tilde{c}_{2}(k),k\geq0$, where $C_{\Phi 1}$ and $C_{\Phi 2}$ are nonnegative constants and $C_{\Phi 1}\geq1$. Then
\begin{align}
\Phi(k+1)\leq C_{\Phi 1}^{h}\Phi(m_{k}h)+C_{\Phi 1}^{h}C_{\Phi 2}\sum_{i=m_{k}h}^k\tilde{c}_{2}(i),\notag
\end{align}
where $m_{k}=\lfloor\frac{k}{h}\rfloor$, $h$ is a given positive integer. Especially,
if $\tilde{c}_{1}(k)>0$, $\tilde{c}_{2}(k)>0$, $\tilde{c}_{2}(k)\downarrow 0$, $\limsup\limits_{k\rightarrow\infty}\frac{\tilde{c}_{1}(k)}{\tilde{c}_{1}(k+h)}\leq1$, $\lim\limits_{k\rightarrow\infty}\frac{\tilde{c}_{2}(k)}{\tilde{c}_{1}(k)}=0$ , and $\limsup\limits_{m\rightarrow\infty}\frac{\Phi(mh)}{\tilde{c}_{1}(mh)}\leq C_{\Phi 3}$,  then $\limsup\limits_{k\rightarrow\infty}\frac{\Phi(k)}{\tilde{c}_{1}(k)}\leq C_{\Phi 1}^{h}C_{\Phi 3}$.}
\end{lemma}
\begin{proof}
{Let $m_{k}=\lfloor\frac{k}{h}\rfloor$, then $0\leq k-m_{k}h< h$, $\forall\ k\geq 0$. From $\Phi(k+1)\leq C_{\Phi 1}\Phi(k)+C_{\Phi 2}\tilde{c}_{2}(k)$, we have
\begin{align}
\Phi(k+1)\leq& C_{\Phi 1}^{k-m_{k}h+1}\Phi(m_{k}h)+\sum_{i=m_{k}h}^kC_{\Phi 1 }^{k-i}C_{\Phi 2}\tilde{c}_{2}(i)
\leq C_{\Phi 1}^{h}\Phi(m_{k}h)+C_{\Phi 1}^{h}C_{\Phi 2}\sum_{i=m_{k}h}^k\tilde{c}_{2}(i).\notag
\end{align}
Then, we get
\begin{align}
\limsup_{k\rightarrow\infty}\frac{\Phi(k+1)}{\tilde{c}_{1}(k+1)}\leq&C_{\Phi 1}^{h}\limsup_{k\rightarrow\infty}\frac{\Phi(m_{k}h)}{\tilde{c}_{1}(m_{k}h)}
\limsup_{k\rightarrow\infty}\frac{\tilde{c}_{1}(m_{k}h)}{\tilde{c}_{1}(k+1)}+hC_{\Phi 1}^{h}C_{\Phi 2}\limsup_{k\rightarrow\infty}\frac{\tilde{c}_{2}(m_{k}h)}{\tilde{c}_{1}(m_{k}h)}
\limsup_{k\rightarrow\infty}\frac{\tilde{c}_{1}(m_{k}h)}{\tilde{c}_{1}(k+1)}\cr
\leq&C_{\Phi 1}^{h}C_{\Phi 3}.\notag
\end{align}}
\end{proof}

\vskip0.2cm

\begin{lemma}\label{Edd1}
For the convex optimization problem \eqref{model}, consider the algorithm \eqref{algorithm}-\eqref{algorithm2}. If \Cref{subgradient} holds, then
$\|d(k)\|^2\leq2\sigma^2_{d}\|X(k)\|^2+2NC^2_{d}$.
\end{lemma}

\vskip0.2cm

\begin{proof}
By \Cref{subgradient}, we have
$$\|d(k)\|^2=\sum_{i=1}^N\|d_{f_{i}}(k)\|^2\leq \sum_{i=1}^N(\sigma_{di}\|x_{i}(k)\|+C_{di})^2
\leq2\sigma^2_{d}\|X(k)\|^2+2NC^2_{d}.$$ 
\end{proof}

\vskip 0.2cm

%

%

\begin{lemma}\label{Ee322}
For the convex optimization problem \eqref{model}, consider the algorithm \eqref{algorithm}-\eqref{algorithm2}. If \Cref{intensity} holds, then for any given $x\in\mathbb{R}^{n}$, we have
\begin{equation}\label{Psi1}
\|\Psi(k)\|^2\leq4\sigma^2\|X(k)-\mathbf{1}^T_{N}\otimes x\|^2+2b^2.
\end{equation}
\end{lemma}

\begin{proof}
By the definition of $\Psi(k)$ and \Cref{intensity}, we have
\begin{align}
\|\Psi(k)\|^2=&\max\limits_{1\leq i,j\leq N}(\psi_{ji}(x_{j}(k)-x_{i}(k)))^2\notag\\
\leq&\max\limits_{1\leq i,j\leq N}\left[2\sigma^2\|x_{j}(k)-x_{i}(k)\|^2+2b^2\right]\notag\\
\leq&4\sigma^2\max\limits_{1\leq i,j\leq N}\left[\|x_{j}(k)-x\|^2+\|x_{i}(k)-x\|^2\right]+2b^2\notag\\
\leq &4\sigma^2\left\|X(k)-\mathbf{1}^T_{N}\otimes x\right\|^2+2b^2.\notag
\end{align}
Therefore, \eqref{Psi1} holds.
\end{proof}

\begin{lemma}\label{dx}
For the convex optimization problem \eqref{model}, consider the algorithm \eqref{algorithm}-\eqref{algorithm2}. If Assumption \ref{subgradient} holds, then
\begin{align}
&-2\alpha(k)d^T(k)(X(k)-\mathbf{1}_{N}\otimes x^{*})\cr
\leq&-2\alpha(k)(f(\bar{x}(k))-f(x^{*}))+2\alpha(k)\sqrt{N}(\sigma_{d}\|\bar{x}(k)
\|+C_{d})\|\delta(k)\|,\ \forall\ x^{*}\in\mathcal{X}^{*},\ k\geq 0.\notag
\end{align}
\end{lemma}
\begin{proof}
 By $d_{f_{i}}(k)\in\partial f_{i}(x_{i}(k))$ and \Cref{subgradient}, for any given $x^{*}\in\mathcal{X}^{*}$, we have
\begin{align}
-d^T_{f_{i}}(k)(x_{i}(k)-x^{*})\leq& f_{i}(x^{*})-f_{i}(x_{i}(k))\cr
=&f_{i}(x^{*})-f_{i}(\bar{x}(k))+f_{i}(\bar{x}(k))-f_{i}(x_{i}(k))\cr
\leq& f_{i}(x^{*})-f_{i}(\bar{x}(k))+d^T_{f_{i}}(\bar{x}(k))(\bar{x}(k)-x_{i}(k))\cr
\leq& f_{i}(x^{*})-f_{i}(\bar{x}(k))+(\sigma_{d}\|\bar{x}(k)\|+C_{d})\|\bar{x}(k)-x_{i}(k)\|,\notag
\end{align}
which gives
\begin{align}
&-2\alpha(k)d^T(k)(X(k)-\mathbf{1}_{N}\otimes x^{*})\cr
=&-2\alpha(k)\sum_{i=1}^{N}d^T_{f_{i}}(k)(x_{i}(k)-x^{*})\cr
\leq&2\alpha(k)\sum_{i=1}^{N}(f_{i}(x^{*})-f_{i}(\bar{x}(k))
+2\alpha(k)\sum\limits\limits_{i=1}^{N}(\sigma_{d}\|\bar{x}(k)\|+C_{d})\|\bar{x}(k)-x_{i}(k)\|\cr
\leq&2\alpha(k)(f(x^{*})-f(\bar{x}(k))) +2\alpha(k)\sqrt{N}(\sigma_{d}\|\bar{x}(k)
\|+C_{d})\|\delta(k)\|.\notag
\end{align}
\end{proof}

\begin{lemma}\label{lemmac(i)}
For the convex optimization problem \eqref{model} and the algorithm \eqref{algorithm}-\eqref{algorithm2}, assume that

(a) Assumptions \ref{subgradient}-\ref{intensity}, Assumption \ref{stochasticgraph} and Condition (C1) hold;

(b) there exists a positive integer $h$ and positive constants $\theta$ and $\rho_{0}$, such that

\ \  (b.1) $\inf_{m\geq0}\lambda_{mh}^h\geq\theta$\  a.s.;

\ \  (b.2) $\sup_{k\geq0}\left[E\left[\|\mathcal{L}_{\mathcal{G}(k)}\|^{2^{\max\{h,2\}}}|
\mathcal{F}(k-1)\right]\right]^{\frac{1}{2^{\max\{h,2\}}}}\leq\rho_{0}$\  a.s.\\
Then, there exists a positive constant $C_{4}$ such that, for any given positive sequence $\{\tau(k),\ k\geq0\}$,
\begin{align}
\label{Vmh2}
E[V((m+1)h)]
\leq&\left(1+\tau(mh)\right)\Big[1-2\theta c((m+1)h)+C_{4} c^2((m+1)h))\Big]E[V(mh)]\cr
&+\left(\dfrac{1}{\tau(mh)}+2\right)
\left(hC_{\rho}\sum\limits\limits_{j=mh}^{(m+1)h-1}\alpha^2(j)\left(2\sigma^2_{d}
E\left[\left\|X(j)\right\|^2\right]+2NC^2_{d}\right)\right)\cr
&+4\Bigg(h\rho_{1}C_{\rho}C_{\xi}\sum
\limits\limits_{j=mh}^{(m+1)h-1}c^2(j)\left(4\sigma^2
E\left[\left\|X(j)\right\|^2\right]+2b^2\right)\cr
&+hC_{\rho}\sum\limits\limits_{j=mh}^{(m+1)h-1}\alpha^2(j)\left(\sigma_{\zeta}
E\left[\left\|X(j)\right\|^2\right]+C_{\zeta}\right)\Bigg),\ m\geq m_{0}.
\end{align}
\end{lemma}

\begin{proof}
Let $\Phi(m,s)=(I_{N}-c(m-1)P\mathcal{L}_{\mathcal{G}(m-1)})\ldots(I_{N}-c(s)P\mathcal{L}_{\mathcal{G}(s)}),~m> s\geq 0,~\Phi(s,s)=I_{N},~s\geq0$. By \eqref{error} and some iterative calculations, we have
\begin{equation}\label{iteration}
\delta((m+1)h)=(\Phi((m+1)h,mh)\otimes I_{n})\delta(mh)+\tilde{\Lambda}_{m}^{mh}-\tilde{d}_{m}^{mh},
\end{equation}
where
\begin{align}
\tilde{\Lambda}_{m}^{mh}=&\sum_{j=mh}^{(m+1)h-1}(\Phi((m+1)h,j+1)P\otimes I_{n})\big(c(j)D(j)\Psi(j)\xi(j)-\alpha(j)\zeta(j)\big),\label{zeta}\\
\tilde{d}_{m}^{mh}=&\sum\limits_{j=mh}^{(m+1)h-1}\alpha(j)(\Phi((m+1)h,j+1)P\otimes I_{n})d(j).\label{d}
\end{align}
By the definition of $V(k)$, \eqref{iteration} and $-2\left(\tilde{\Lambda}_{m}^{mh}\right)^T\left(\tilde{d}_{m}^{mh}\right)
\leq\left(\tilde{\Lambda}_{m}^{mh}\right)^T\left(\tilde{\Lambda}_{m}^{mh}\right)
+\left(\tilde{d}_{m}^{mh}\right)^T\left(\tilde{d}_{m}^{mh}\right)$, we get
\begin{align}
\label{V((m+1)h)}
&V((m+1)h)\cr
\leq&\delta^T(mh)(\Phi^T((m+1)h,mh)\Phi((m+1)h,mh)\otimes I_{n})\delta(mh)\cr
&+2\left(\tilde{\Lambda}_{m}^{mh}\right)^T\left(\tilde{\Lambda}_{m}^{mh}\right)+2\delta^T(mh)(\Phi^T((m+1)h,mh)\otimes I_{n})\tilde{\Lambda}_{m}^{mh}\cr
&+2\left(\tilde{d}_{m}^{mh}\right)^T\left(\tilde{d}_{m}^{mh}\right)-2\delta^T(mh)(\Phi^T((m+1)h,mh)\otimes I_{n})\tilde{d}_{m}^{mh}.
\end{align}
We now consider the mathematical expectation of each term on the right side of \eqref{V((m+1)h)}. For the first term, by Condition (C1), we know that there exists a  positive integer $m_{0}$ and a positive constant $C_{3}$, such that $c^2(mh)\leq C_{3}c^2((m+1)h),\ \forall\ \ m\geq m_{0}$, and $c(k)\leq1,\ \forall\ \ k\geq m_{0}h$. From Condition (b.2) and the conditional Lyapunov inequality, we have
\begin{equation}\label{laplace}
\sup_{k\geq0}E\left[\|\mathcal{L}_{\mathcal{G}(k)}\|^i|\mathcal{F}(k-1)\right]
\leq\sup_{k\geq0}\left[E\left[\|\mathcal{L}_{\mathcal{G}(k)}\|^{2^h}|\mathcal{F}(k-1)\right]\right]^{\frac{i}{2^h}}\leq\rho_{0}^i\ \mbox{a.s.},\ \forall\  2\leq i\leq2^h.
\end{equation}
By multiplying term by term, applying conditional H\"{o}lder inequality, noting that $c(mh)$ decreases monoto-\\ nously as $m$ increases, and from \eqref{laplace}, we have
\begin{align}
\label{deltadelta}
&E\left[\left\|\Phi^T((m+1)h,mh)\Phi((m+1)h,mh)
-I_{N}+\sum\limits_{i=mh}^{(m+1)h-1}c(i)
\left(\mathcal{L}^T_{\mathcal{G}(i)}P^T+P\mathcal{L}_{\mathcal{G}(i)}\right)
\right\| \Bigg| \mathcal{F}(mh-1)\right]\cr
\leq&C_{4}c^2((m+1)h),\ m\geq m_{0},
\end{align}
where $C_{4}=C_{3}\left[(1+\rho_{0})^{2h}-1-2h\rho_{0}\right]$. From $\delta(mh)\in\mathcal{F}(mh-1)$, we have $V(mh)\in\mathcal{F}(mh-1)$. Noting that $\|A\otimes I_{n}\|=\|A\|$, we have
\begin{align}
&E\bigg[\delta^T(mh)\bigg[\bigg[\Phi^T((m+1)h,mh)\Phi((m+1)h,mh)-I_{N}
+\sum\limits_{i=mh}^{(m+1)h-1}c(i)\left[\mathcal{L}^T_{\mathcal{G}(i)}P^T
+P\mathcal{L}_{\mathcal{G}(i)}\right]\bigg]\otimes I_{n}\bigg]\delta(mh)\bigg]\notag\\
\leq& E\bigg[\Big\|\Phi^T((m+1)h,mh)\Phi((m+1)h,mh)-I_{N}
+\sum\limits_{i=mh}^{(m+1)h-1}c(i)\left(\mathcal{L}^T_{\mathcal{G}(i)}P^T
+P\mathcal{L}_{\mathcal{G}(i)}\right)
\Big\|V(mh)\bigg]\notag\\
=& E\bigg[E\bigg[\bigg\|\Phi^T((m+1)h,mh)\Phi((m+1)h,mh)-I_{N}+
\sum\limits_{i=mh}^{(m+1)h-1}c(i)\Big(\mathcal{L}^T_{\mathcal{G}(i)}P^T\notag\\
&+P\mathcal{L}_{\mathcal{G}(i)}\Big)
\bigg\| \bigg|\mathcal{F}(mh-1)\bigg]V(mh)\bigg]\notag\\
\leq&C_{4}c^2((m+1)h)E[V(mh)],\  m\geq m_{0}.\label{tildedeltadelta1}
\end{align}
Noting that $\mathcal{G}(i|i-1)$ is balanced a.s., it is known that $\mathcal{G}(i|mh-1)$ is balanced a.s., $mh\leq i\leq(m+1)h-1$.
Then by Condition (b.1), we get
\begin{align}
\label{lambda1}
&E\left[\delta^T(mh)\left[\sum\limits_{i=mh}^{(m+1)h-1}c(i)
\left(P\mathcal{L}_{\mathcal{G}(i)}
+\mathcal{L}^T_{\mathcal{G}(i)}P\right)\otimes I_{n}\right]\delta(mh)\right]\cr
=&E\left[E\left[\delta^T(mh)\left[\sum\limits_{i=mh}^{(m+1)h-1}c(i)\left(P\mathcal{L}_{\mathcal{G}(i)}
+\mathcal{L}^T_{\mathcal{G}(i)}P\right)\otimes I_{n}\right]\delta(mh)\bigg|\mathcal{F}(mh-1)\right]\right]\cr
=&2E\left[\delta^T(mh)\left[\sum\limits_{i=mh}^{(m+1)h-1}c(i)
E\left[\widehat{\mathcal{L}}_{\mathcal{G}(i)}\otimes I_{n}
\bigg|\mathcal{F}(mh-1)\right]\right]\delta(mh)\right]\cr
\geq&2c((m+1)h)E\left[\delta^T(mh)\left[\sum\limits_{i=mh}^{(m+1)h-1}
E\left[\widehat{\mathcal{L}}_{\mathcal{G}(i)}\otimes I_{n}
\Big|\mathcal{F}(mh-1)\right]\right]\delta(mh)\right]\cr
\geq&2c((m+1)h)E\left[\lambda_{mh}^hV(mh)\right]\cr
\geq&2c((m+1)h)E\left[\inf\limits_{m\geq0}(\lambda_{mh}^h)V(mh)\right]\cr
\geq&2\theta c((m+1)h)E[V(mh)],
\end{align}
which together with \eqref{tildedeltadelta1} gives
\begin{align}
\label{deltaphidelta}
&E[\delta^T(mh)(\Phi^T((m+1)h,mh)\Phi((m+1)h,mh)\otimes I_{n})\delta(mh)]\cr
=&E\bigg[\delta^T(mh)\bigg[\Big[\Phi^T((m+1)h,mh)
\Phi((m+1)h,mh)-I_{N}+\sum\limits\limits_{i=mh}^{(m+1)h-1}c(i)
\left[\mathcal{L}^T_{\mathcal{G}(i)}P^T+P\mathcal{L}_{\mathcal{G}(i)}\right]\Big]\otimes I_{n}\bigg]\delta(mh)\bigg]\cr
&+E[V(mh)]-E\bigg[\delta^T(mh)\sum\limits\limits_{i=mh}^{(m+1)h-1}c(i)
\Big[\left[\mathcal{L}^T_{\mathcal{G}(i)}P^T
+P\mathcal{L}_{\mathcal{G}(i)}\right]\otimes I_{n}\Big]\delta(mh)\bigg]\cr
\leq&\Big[1-2\theta c((m+1)h)+C_{4}c^2((m+1)h)\Big]E[V(mh)],\  m\geq m_{0}.
\end{align}
For the second term on the right side of \eqref{V((m+1)h)}, by \eqref{zeta}, we have
\begin{equation}\label{zetazeta}
\left(\tilde{\Lambda}_{m}^{mh}\right)^T\left(\tilde{\Lambda}_{m}^{mh}\right)
=\left(\tilde{\xi}_{m}^{mh}-\tilde{\zeta}_{m}^{mh}\right)^T
\left(\tilde{\xi}_{m}^{mh}-\tilde{\zeta}_{m}^{mh}\right)
\leq2\left(\tilde{\xi}_{m}^{mh}\right)^T\left(\tilde{\xi}_{m}^{mh}\right)
+2\left(\tilde{\zeta}_{m}^{mh}\right)^T\left(\tilde{\zeta}_{m}^{mh}\right),
\end{equation}
where $\tilde{\xi}_{m}^{mh}=\sum\limits_{j=mh}^{(m+1)h-1}c(j)(\Phi((m+1)h,j+1)P\otimes I_{n})D(j)\Psi(j)\xi(j)$
and $\tilde{\zeta}_{m}^{mh}$$=$$\sum\limits_{j=mh}^{(m+1)h-1}\alpha(j)(\Phi((m+1)h,$ $j+1)P\otimes I_{n})\zeta(j)$.
Then by $C_r$ inequality, we get
\begin{align}
\label{xixi1}
&E\left[\left(\tilde{\xi}_{m}^{mh}\right)^T\left(\tilde{\xi}_{m}^{mh}\right)\right]\cr
\leq&h\sum\limits_{j=mh}^{(m+1)h-1}c^2(j)E\left[\xi^T(j)\Psi^T(j)
D^T(j)((P\Phi^T((m+1)h,j+1)\Phi((m+1)h,j+1)P)\otimes I_{n})D(j)\Psi(j)\xi(j)\right]\cr
\leq&h\sum\limits_{j=mh}^{(m+1)h-1}c^2(j)E\left[\left\|\Phi^T((m+1)h,j+1)\Phi((m+1)h,j+1)
\right\|\|D(j)\|^2\|\Psi(j)\|^2\|\xi(j)\|^2\right].
\end{align}
From \Cref{noise1}, it follows that
\begin{align}
\label{xixi2}
&E\left[\left\|\Phi^T((m+1)h,j+1)\Phi((m+1)h,j+1)\right\|\|D(j)\|^2\|\Psi(j)\|^2\|\xi(j)\|^2\right]\cr
=&E\left[\|\Psi(j)\|^2E\left[\left\|\Phi^T((m+1)h,j+1)\Phi((m+1)h,j+1)\right\|
\|D(j)\|^2\|\xi(j)\|^2\Big|\mathcal{F}(j-1)\right]\right]\cr
=&E\Big[\|\Psi(j)\|^2E[\|\Phi^T((m+1)h,j+1)\Phi((m+1)h,j+1)\|\|D(j)\|^2|\mathcal{F}(j-1)] E[\|\xi(j)\|^2|\mathcal{F}(j-1)]\Big]\cr
\leq&C_{\xi}E\Big[\|\Psi(j)\|^2E[\|\Phi^T((m+1)h,j+1)\Phi((m+1)h,j+1)\|\cr
&\times\|D(j)\|^2|\mathcal{F}(j-1)]\Big],\  mh\leq j\leq (m+1)h-1.
\end{align}
From Condition (b.2), it is known that there exists a positive constant $\rho_{1}$, such that
\begin{equation}\label{DkDk}
\sup_{k\geq0}\left[E\left[\left\|D(k)\right\|^4\Big|\mathcal{F}(k-1)\right]\right]^{\frac{1}{2}}\leq\rho_{1}\ \mbox{a.s.}
\end{equation}
By Condition (b.2) and \eqref{laplace}, we have
\begin{align}
\label{phiphi}
\left[E\left[\left\|\Phi^T((m+1)h,j+1)\Phi((m+1)h,j+1)
\right\|^2\Big|\mathcal{F}(j-1)\right]\right]^{\frac{1}{2}}\leq C_{\rho},\ mh\leq j \leq (m+1)h-1\ \mbox{a.s.},
\end{align}
where $C_{\rho}=\left\{2^{2(h-1)}\sum\limits_{l=0}^{2(h-1)}M_{2(h-1)}^l\rho_{0}^{2l}\right\}^{\frac{1}{2}}$,  $M_{2(h-1)}^l$ is the combinatorial number of choosing $l$ elements from $2(h-1)$. By \eqref{DkDk}, \eqref{phiphi}, and conditional H\"{o}lder inequality, we get
\begin{align}
\label{phiphiD}
&E\left[\|\Phi^T((m+1)h,j+1)\Phi((m+1)h,j+1)\|\|D(j)\|^2|\mathcal{F}(j-1)\right]\cr
\leq&\left\{E\left[\|\Phi^T((m+1)h,j+1)\Phi((m+1)h,j+1)
\|^{2}|\mathcal{F}(j-1)\right]\right\}^{\frac{1}{2}} \left\{E\left[\|D(j)\|^4|\mathcal{F}(j-1)\right]\right\}^{\frac{1}{2}}\cr
\leq&\rho_{1}C_{\rho}\ \mbox{a.s.}
\end{align}
From \Cref{Ee322}, \eqref{xixi1}, \eqref{xixi2} and \eqref{phiphiD}, we get
\begin{equation}\label{xixi3}
\begin{array}{rcl}
E\left[\left(\tilde{\xi}_{m}^{mh}\right)^T\left(\tilde{\xi}_{m}^{mh}\right)\right]\leq h\rho_{1}C_{\rho}C_{\xi}\sum\limits\limits_{j=mh}^{(m+1)h-1}c^2(j)\left(4\sigma^2
E\left[\|X(j)\|^2\right]+2b^2\right).\\
\end{array}
\end{equation}
From conditional H\"{o}lder inequality, \Cref{noise2} and \eqref{phiphi}, it follows that
\begin{align}
\label{phiomega}
&E\left[\|\Phi^T((m+1)h,j+1)\Phi((m+1)h,j+1)\|\|\zeta(j)\|^2\right]\cr
=&E\left[E\left[\|\Phi^T((m+1)h,j+1)\Phi((m+1)h,j+1)\|\|\zeta(j)
\|^2|\mathcal{F}(j-1)\right]\right]\cr
=&E\left[E\left[\|\Phi^T((m+1)h,j+1)\Phi((m+1)h,j+1)\||\mathcal{F}(j-1)\right]
E\left[\|\zeta(j)\|^2|\mathcal{F}(j-1)\right]\right]\cr
\leq&E\left[\left\{E\left[\|\Phi^T((m+1)h,j+1)\Phi((m+1)h,j+1)
\|^2\Big|\mathcal{F}(j-1)\right]\right\}^{\frac{1}{2}}
\left(\sigma_{\zeta}\|X(j)\|^2+C_{\zeta}\right)\right]\cr
\leq&C_{\rho}\left(\sigma_{\zeta}E\left[\|X(j)\|^2\right]+C_{\zeta}\right),\  mh\leq j\leq (m+1)h-1,
\end{align}
which leads to
\begin{align}
\label{omegaomega}
&E\left[\left(\tilde{\zeta}_{m}^{mh}\right)^T\left(\tilde{\zeta}_{m}^{mh}\right)\right]\cr
\leq&hE\left[\sum\limits\limits_{j=mh}^{(m+1)h-1}\alpha^2(j)\zeta^T(j)\left(\left(P
\Phi^T((m+1)h,j+1)\Phi((m+1)h,j+1)P\right)\otimes I_{n}\right)\zeta(j)\right]\cr
\leq&h\sum\limits\limits_{j=mh}^{(m+1)h-1}\alpha^2(j)
E\left[\left\|\Phi^T((m+1)h,j+1)\Phi((m+1)h,j+1)\right\|\|\zeta(j)\|^2\right]\cr
\leq&hC_{\rho}\sum\limits\limits_{j=mh}^{(m+1)h-1}\alpha^2(j)
\big(\sigma_{\zeta}E\left[\|X(j)\|^2\right]+C_{\zeta}\big).
\end{align}
Thus,  by \eqref{zetazeta}, \eqref{xixi3} and \eqref{omegaomega}, we get
\begin{align}
\label{zetazeta1}
E\left[\left(\tilde{\Lambda}_{m}^{mh}\right)^T\left(\tilde{\Lambda}_{m}^{mh}\right)\right]\leq&2h\rho_{1}C_{\rho}C_{\xi}\sum\limits\limits_{j=mh}^{(m+1)h-1}c^2(j)
\left(4\sigma^2E\left[\|X(j)\|^2\right]+2b^2\right)\cr
&+2hC_{\rho}\sum\limits\limits_{j=mh}^{(m+1)h-1}\alpha^2(j)
\big(\sigma_{\zeta}E\left[\|X(j)\|^2\right]+C_{\zeta}\big).
\end{align}
For the third term on the right side of  \eqref{V((m+1)h)}, by $\delta(mh)\in\mathcal{F}(j-1),~j\geq mh$, \Cref{noise1}, we have
\begin{align}
\label{deltaxi1}
&E\left[\delta^T(mh)((\Phi^T((m+1)h,mh)\Phi((m+1)h,j+1)P)\otimes I_{n})D(j)\Psi(j)\xi(j)\right]\cr
=&E\Big[\delta^T(mh)E\Big[((\Phi^T((m+1)h,mh)\Phi((m+1)h,j+1)P)\otimes I_{n})D(j)\Psi(j)\xi(j)
\Big|\mathcal{F}(j-1)\Big]\Big]\cr
=&E\Big[\delta^T(mh)E\left[((\Phi^T((m+1)h,mh)\Phi((m+1)h,j+1)P)\otimes I_{n})D(j)|\mathcal{F}(j-1)\right] \Psi(j)E[\xi(j)|\mathcal{F}(j-1)]\Big]\cr
=&0,\ mh\leq j\leq(m+1)h-1,\ m\geq0.
\end{align}
Similarly, from \Cref{noise2}, we have $E\left[\delta^T(mh)((\Phi^T((m+1)h,mh)\Phi((m+1)h,j+1)P)\otimes I_{n})\zeta(j)\right]=0$, $mh\leq j\leq(m+1)h-1,\ m\geq0$.
This together with \eqref{zeta} and \eqref{deltaxi1} gives
\begin{equation}\label{deltaxi}
E\left[\delta^T(mh)(\Phi((m+1)h,mh)\otimes I_{n})^T\tilde{\Lambda}_{m}^{mh}\right]=0.
\end{equation}
By \Cref{Edd1}, \eqref{phiphi} and conditional H\"{o}lder inequality, we have
\begin{align}
\label{dphidh2}
&E\left[\left\|\Phi^T((m+1)h,j+1)\Phi((m+1)h,j+1)\right\|\|d(j)\|^2\right]\cr
\leq&E\left[\left\|\Phi^T((m+1)h,j+1)\Phi((m+1)h,j+1)\right\|\left(2\sigma^2_{d}
\|X(j)\|^2+2NC^2_{d}\right)\right]\cr
=&E\bigg[E\Big[\left\|\Phi^T((m+1)h,j+1)\Phi((m+1)h,j+1)\right\|\Big|
\mathcal{F}(j-1)\Big]\left(2\sigma^2_{d}\|X(j)\|^2+2NC^2_{d}\right)\bigg]\cr
\leq&E\bigg[\Big[E\big[\|\Phi^T((m+1)h,j+1)\Phi((m+1)h,j+1)
\|^2\big|\mathcal{F}(j-1)\big]\Big]^{\frac{1}{2}}\left(2\sigma^2_{d}\|X(j)\|^2
+2NC^2_{d}\right)\bigg]\cr
\leq&C_{\rho}\left(2\sigma^2_{d}E\left[\|X(j)\|^2\right]+2NC^2_{d}\right),\ j\geq mh,
\end{align}
where the first ``$=$'' is derived by $X(j)\in\mathcal{F}(j-1)$. This together with \eqref{d}  and the $C_r$-inequality gives
\begin{align}
\label{Lambda3}
&E\left[\left(\tilde{d}_{m}^{mh}\right)^T\left(\tilde{d}_{m}^{mh}\right)\right]\cr
\leq& h\sum\limits\limits_{j=mh}^{(m+1)h-1}\alpha^2(j)E\left[d^T(j)((P^T\Phi^T
((m+1)h,j+1)\Phi((m+1)h,j+1)P)\otimes I_{n})d(j)\right]\cr
\leq&h\sum\limits\limits_{j=mh}^{(m+1)h-1}\alpha^2(j)
E\left[\left\|P^T\Phi^T((m+1)h,j+1)\Phi((m+1)h,j+1)P\right\|\|d(j)\|^2\right]\cr
\leq& hC_{\rho}\sum\limits\limits_{j=mh}^{(m+1)h-1}
\alpha^2(j)\left(2\sigma^2_{d}E\left[\|X(j)\|^2\right]+2NC^2_{d}\right).
\end{align}
For the fifth term on the right side of  \eqref{V((m+1)h)}, from the inequality $p^Tq\leq\frac{1}{2\tau}\|p\|^2+\frac{\tau}{2}\|q\|^2,\ \forall\ \tau>0,\ p,q\in\mathbb{R}^n$, \eqref{deltaphidelta} and \eqref{Lambda3}, it follows that
\begin{align}
&-2E\left[\delta^T(mh)\left(\Phi^T((m+1)h,mh)\otimes I_{n}\right)\tilde{d}_{m}^{mh}\right]\cr
\leq&\tau(mh)E\left[\delta^T(mh)(\Phi^T((m+1)h,mh))(\Phi((m+1)h,mh)\otimes I_{n})\delta(mh)\right]+\dfrac{1}{\tau(mh)}E\left[(\tilde{d}_{m}^{mh})^T(\tilde{d}_{m}^{mh})\right]\cr
\leq&\tau(mh)\Big[1-2\theta c((m+1)h)+C_{4}c^2((m+1)h)\Big]E[V(mh)]\cr
&+\dfrac{hC_{\rho}}{\tau(mh)}\sum\limits_{j=mh}^{(m+1)h-1}\alpha^2(j)
\left(2\sigma^2_{d}E\left[\|X(j)\|^2\right]+2NC^2_{d}\right).\notag
\end{align}
Taking the mathematical expectations on both sides of \eqref{V((m+1)h)}, by \eqref{deltaphidelta}, \eqref{zetazeta1}, \eqref{deltaxi}, \eqref{Lambda3} and the above inequality, we get \eqref{Vmh2}.
\end{proof}

\vskip 0.2cm
\begin{remark}
If Conditions (C1)-(C5) hold, then  Conditions (C3)'-(C5)' hold.
 From condition (C4), we have, for any given positive constant $C>0$,
 \begin{align}
 \lim_{k \to \infty} \frac{\alpha^{\frac{3}{2}}(k)\exp(C \sum_{t=0}^{k} \alpha(t))}{c^{\frac{1}{2}}(k)\alpha(k)}=\lim_{k \to \infty}\left(\frac{\alpha(k)\exp(2C \sum_{t=0}^{k} \alpha(t))}{c(k)}\right)^{\frac{1}{2}}=0.\notag
 \end{align}
From the above equality and Condition (C3), we obtain Condition (C3)'. By Conditions (C1) and (C4), we have, for any given positive constant $C>0$,
\begin{align}
\lim\limits_{k\rightarrow\infty}\frac{\alpha(k)}{c(k)}=
\lim\limits_{k\rightarrow\infty}\frac{\alpha(k)\exp(C \sum_{t=0}^{k} \alpha(t))}{c(k)}\lim\limits_{k\rightarrow\infty}\frac{1}{\exp(C \sum_{t=0}^{k} \alpha(t))}=0.\notag
\end{align}
That is,  Condition (C4)' holds. From Condition (C5) and similar to the proof of \Cref{lemmac},  we have, for any given positive integer $h$,
\begin{align}
\alpha(k)\beta(k)-\alpha(k+h)\beta(k+h)\leq hC_{6}C_{7}\alpha(k)\beta(k)\alpha(k+h)\beta(k+h),\ k\geq k_{2},\label{speedpdalphabeta}
\end{align}
where $\beta(k)$ and $k_{2}$ are given by \Cref{lemma3} and (\ref{pi2}) in \Cref{lemmac}.
By Condition (C1), there exists $\widetilde{k}_{2}>0$ such that $0<hC_{0}\alpha(k)<1,\ \forall \ k \geqslant \widetilde{k}_{2}$. Then
\begin{align}
\frac{\beta(k+h)}{\beta(k)}-1 =\exp\left(C_{0}\sum_{t=k+1}^{k+h}\alpha(t)\right)-1\leq \exp\left(hC_{0}\alpha(k+1)\right)-1 \leq h C_{0}\alpha(k+1), \ k \geq \widetilde{k}_{2}.\notag
\end{align}
By (\ref{speedpdalphabeta}) and the above inequality, we have
\begin{align}
\frac{\alpha(k)-\alpha(k+h)}{\alpha(k+h)}=&
\frac{\alpha(k)\beta(k)-\alpha(k+h)\beta(k+h)}{\alpha(k+h)\beta(k)}
+\frac{\alpha(k+h)\beta(k+h)-\alpha(k+h)\beta(k)}{\alpha(k+h)\beta(k)}\notag\\
\leq & hC_{6}C_{7}\alpha(k)\beta(k+h)+h C_{0}\alpha(k+1),\ k \geq \max\{k_{2},\ \widetilde{k}_{2}\}.\label{speedpdalphabeta1}
\end{align}
Noting that $1\leq\dfrac{\beta(k+h)}{\beta(k)}\leq\exp(hC_{0}\alpha(k))$ and $\alpha(k)\downarrow0$, we have $\lim\limits_{k\rightarrow\infty}\dfrac{\beta(k+h)}{\beta(k)}=1$. Thus, by Conditions (C1), (C4) and (C4)', it follows that
\begin{align}
&\limsup_{k \to \infty} \frac{hC_{6}C_{7}\alpha(k)\beta(k+h)+h C_{0}\alpha(k+1)}{c(k+h)}\notag\\
=& hC_{6}C_{7}\limsup_{k \to \infty} \frac{\alpha(k)\beta(k)}{c(k)}\limsup_{k \to \infty} \frac{\beta(k+h)}{\beta(k)}\limsup_{k \to \infty} \frac{c(k)}{c(k+h)}+ h C_{0}\limsup_{k \to \infty} \frac{\alpha(k+1)}{c(k+1)}\limsup_{k \to \infty} \frac{c(k+1)}{c(k+h)}=0. \notag
\end{align}
This together with (\ref{speedpdalphabeta1}) leads to Condition (C5)'.
\end{remark}

\section{Proofs of  \Cref{theorem3.2}}
\label{app:rate}
To prove \Cref{theorem3.2}, we first prove the following lemmas.

\vskip 0.2cm

\begin{lemma}\label{lemmac(iii)}
{For the convex optimization problem \eqref{model}, consider the algorithm \eqref{algorithm}-\eqref{algorithm2}  with step sizes $c(k)=\frac{c_{0}}{(k+1)^{\gamma_{1}}}$, $\alpha(k)=\frac{\alpha_{0}}{(k+1)^{\gamma_{2}}}$, where $\gamma_{1}\in(0.5,1)$, $\gamma_{2}\in(\gamma_{1},1]$, $c_{0}>0$, $\alpha_{0}>0$.  Suppose that the local cost function $f_{i}(\cdot)$ in problem \eqref{model} is $\mu$-strongly convex and assume that}

{(a) Assumptions \ref{subgradient}-\ref{intensity} and Assumption \ref{stochasticgraph} hold;}

{(b) there exists a positive integer $h$ and positive constants $\theta$ and $\rho_{0}$, such that}

{\ \  (b.1) $\inf_{m\geq0}\lambda_{mh}^h\geq\theta$\  a.s.;}

{\ \  (b.2) $\sup_{k\geq0}\left[E\left[\left\|\mathcal{L}_{\mathcal{G}(k)}\right\|^{2^{\max\{h,2\}}}|
\mathcal{F}(k-1)\right]\right]^{\frac{1}{2^{\max\{h,2\}}}}\leq\rho_{0}$\  a.s.;\\}
{
(i) if $3\gamma_{1}>2\gamma_{2}$, then
\begin{align}\label{EVkorder2str1}
\limsup_{k\rightarrow\infty}\frac{c^2(k)E[V(k)]}{\alpha^2(k)}\leq C_{V1};
\end{align}
(ii) if $3\gamma_{1}\leq2\gamma_{2}$, then
\begin{align}\label{EVkorder2str2}
\limsup_{k\rightarrow\infty}\frac{E[V(k)]}{c(k)}\leq C_{V2}.
\end{align}
Especially, if $\gamma_{1}\in(0.5,\frac{2}{3})$, and $\gamma_{2}=1$, then
\begin{align}\label{EVkorder2str3}
E[V(k)]\leq C_{V3}c(k),\ \forall\ k\geq \lceil\theta c_{0}\rceil-1,
\end{align}
where $C_{V1}$ and $C_{V2}$ are given in \Cref{theorem3.2}, $C_{p}=\left(\frac{1}{\theta}+2\right)\bigg(h^2C_{\rho}\Big(2NC^2_{d}+
2\sigma^2_{d}C_{X1}\Big)\bigg)\frac{\alpha_{0}^2}{c_{0}^{3}}(1+h)$
$+4h^2\rho_{1}$ $\times C_{\rho}C_{\xi}\Big(4\sigma^2C_{X1}
+2b^2\Big)
+4h^2C_{\rho}\Big(\sigma_{\zeta}C_{X1}+C_{\zeta}\Big) \frac{\alpha_{0}^2}{c_{0}^2}$, $C_{V3}=\frac{\eta^hC_{q}C_{p}c_{0}^2}
{c_{0}\theta\Big(1+2h\Big)^{-\gamma_{1}}-\gamma_{1}h}+\eta^{h}C_{q}\exp\left(\frac{c_{0}\theta((\hat{m}+1)h+1
)^{1-\gamma_{1}}}{h(1-\gamma_{1})}\right)$ $\times E[V(\hat{m}h)]\exp\left(-\frac{c_{0}\theta(2h+1)^{1-\gamma_{1}}}{h(1-\gamma_{1})}\right)
$,
$C_{q}=\max\left\{0, (1+c_{0}\theta)C_{4}-2\theta^2\right\}\sum_{m=0}^{\infty}c^2((m+1)h)$,
$C_{4}=(1+h)[(1+\rho_{0})^{2h}-1-2h\rho_{0}]$, $\hat{m}=\max\left\{\lceil\frac{\lceil\alpha_{0}\mu\rceil-1}{h}\rceil,\lceil\frac{\lceil\theta c_{0}\rceil-1}{h}\rceil-1\right\}$.}
\end{lemma}

\begin{proof}
{
From \Cref{lemmac(i)} with $\tau(mh)=\tau_{0}c((m+1)h),\tau_{0}\in(0,2\theta)$, \Cref{lemma3(iii)} and the monotone property of $c(k),\alpha(k)$, we get
\begin{align}\label{Vmh2cstr}
&E[V((m+1)h)]\cr
\leq&\left(1+\tau_{0}c((m+1)h)\right)\Big[1-2\theta c((m+1)h)+C_{4} c^2((m+1)h))\Big]E[V(mh)]\cr
&+\left(\dfrac{1}{\tau_{0}c((m+1)h)}+2\right)\left(h^2C_{\rho}\Big(2\sigma^2_{d}\Big(\frac{2NC_{d}^2}{\mu^2}+1\Big)
+2NC^2_{d}\Big)\right)\alpha^2(mh)\cr
&+4\Bigg(h^2\rho_{1}C_{\rho}C_{\xi}\Big(4\sigma^2\Big(\frac{2NC_{d}^2}{\mu^2}+1\Big)
+2b^2\Big)c^2(mh)+h^2C_{\rho}\Big(\sigma_{\zeta}\Big(\frac{2NC_{d}^2}{\mu^2}+1\Big)+C_{\zeta}\Big)\alpha^2(mh)\Bigg)\cr
=&\Big[1-(2\theta-\tau_{0}) c((m+1)h)+\tilde{q}_{0}(mh)\Big]E[V(mh)]
+\tilde{p}(mh),\ m\geq\max\left\{m_{0},\left\lceil\frac{k_{0}}{h}\right\rceil\right\},
\end{align}
where $\tilde{q}_{0}(mh)=(C_{4}-2\tau_{0}\theta) c^2((m+1)h)+\tau_{0} C_{4}c^3((m+1)h)$, $\tilde{p}(mh)=\left(\dfrac{1}{\tau_{0}c((m+1)h)}+2\right)
\bigg(h^2C_{\rho}\times $\\ $\Big(2NC^2_{d}+$
$2\sigma^2_{d}\Big(\frac{2NC_{d}^2}{\mu^2}+1\Big)\Big)\bigg)\alpha^2(mh)
+4h^2\rho_{1}C_{\rho}C_{\xi}\Big(4\sigma^2\Big(\frac{2NC_{d}^2}{\mu^2}+1\Big)
+2b^2\Big)c^2(mh)
+4h^2C_{\rho}\Big(\sigma_{\zeta}\Big(\frac{2NC_{d}^2}{\mu^2}+\\ 1\Big)
+C_{\zeta}\Big)\alpha^2(mh)$. By $2\theta-\tau_{0}>0$ and $\tilde{q}_{0}(mh)=o(c((m+1)h))$, then there exists a positive integer $m_{1}$, such that
\begin{align}\label{hatq0}
0<(2\theta-\tau_{0}) c((m+1)h)-\tilde{q}_{0}(mh))\leq1,\ \forall\ m\geq m_{1}.
\end{align}}
\ \ \ {(i).
Let $\widetilde{\Pi}(k)=\dfrac{c^2(k)V(k)}{\alpha^2(k)}$. From \eqref{Vmh2cstr}, \eqref{hatq0} and the monotonically decreasing property of $c(k)$,  we have
\begin{align}\label{pi11strong1}
E\left[\widetilde{\Pi}((m+1)h)\right]
\leq&\dfrac{\alpha^2(mh)}{\alpha^2((m+1)h)}\Big[1-(2\theta-\tau_{0}) c((m+1)h)+\tilde{q}_{0}(mh)\Big]E\left[\widetilde{\Pi}(mh)\right]\cr
&+\dfrac{c^2((m+1)h)}{\alpha^2((m+1)h)}\tilde{p}(mh),\ \ m\geq\max\left\{m_{0},m_{1},\left\lceil\dfrac{k_{0}}{h}\right\rceil\right\}.
\end{align}
From $\alpha(k)=\dfrac{\alpha_{0}}{(k+1)^{\gamma_{2}}}$, we have
\begin{align}\label{alphastrong}
\dfrac{\alpha^2(mh)}{\alpha^2((m+1)h)}
=\left(1+\dfrac{h}{\alpha_{0}}\alpha(mh)\right)^{2\gamma_{2}}
\leq\left(1+\dfrac{\gamma_{2}h}{\alpha_{0}}\alpha(mh)\right)^2
\leq1+C_{\alpha}\alpha(mh),
\end{align}
where $C_{\alpha}=\dfrac{\gamma_{2}h}{\alpha_{0}}(2+\gamma_{2}h)$, the first $"\leq"$ is by the inequality $(1+x)^{\gamma}\leq1+\gamma x$, $\forall\ \gamma\in(0,1]$, and the second is by $\alpha(mh)\leq\alpha_{0}$. From \eqref{pi11strong1} and \eqref{alphastrong}, we obtain, for any $m\geq\max\left\{m_{0},m_{1},\left\lceil\dfrac{k_{0}}{h}\right\rceil\right\}$,
\begin{align}\label{pi11strong2}
&E[\widetilde{\Pi}((m+1)h)]\cr
\leq&\Big[1-(2\theta-\tau_{0}) c((m+1)h)+\tilde{q}_{1}(mh)\Big]E[\widetilde{\Pi}(mh)]+\dfrac{c^2((m+1)h)}{
\alpha^2((m+1)h)}\tilde{p}(mh),
\end{align}
where $\tilde{q}_{1}(mh)=\tilde{q}_{0}(mh)+C_{\alpha}\alpha(mh)\Big[1-(2\theta-\tau_{0}) c((m+1)h)+\tilde{q}_{0}(mh)\Big]$.
Noting that $c(k)=\dfrac{c_{0}}{(k+1)^{\gamma_{1}}}$, $\alpha(k)=\dfrac{\alpha_{0}}{(k+1)^{\gamma_{2}}}$, we have  $\lim_{k\rightarrow\infty}\dfrac{\alpha(k)}{c(k+h)}=0$, thus
$$\lim_{m\rightarrow\infty}\dfrac{\tilde{q}_{1}(mh)}{c((m+1)h)}
=\lim_{m\rightarrow\infty}\dfrac{\tilde{q}_{0}(mh)}{c((m+1)h)}+
\lim_{m\rightarrow\infty}\dfrac{C_{\alpha}\alpha(mh)}{c((m+1)h)}
\big(1-(2\theta-\tau_{0})c((m+1)h)+\tilde{q}_{0}(mh)\big)=0,$$
i.e. $\tilde{q}_{1}(mh)=o(c((m+1)h))$.
Then, by $c(k)=\dfrac{c_{0}}{(k+1)^{\gamma_{1}}}$ and $2\theta-\tau_{0}>0$, there exists a positive integer $m_{2}$ such that
\begin{align}
0<(2\theta-\tau_{0}) c((m+1)h)-\tilde{q}_{1}(mh)\leq1,\ \forall\ m\geq m_{2} \text{ and}\
\sum\limits\limits_{m=0}^{\infty}[(2\theta-\tau_{0}) c((m+1)h)-\tilde{q}_{1}(mh)]=\infty. \label{hatq1str}
\end{align}
From $3\gamma_{1}>2\gamma_{2}$, we have $c^2(k)=o\left(\frac{\alpha^{2}(k)}{c(k)}\right)$.
Noting that $\lim\limits_{m\rightarrow\infty}\frac{\alpha(k)}{\alpha(k+h)}=1$, from \eqref{pi11strong2}, \eqref{hatq1str}, $\tilde{q}_{1}(mh)=o(c((m+1)h))$ and Lemma 1.2.25 in \cite{Guo}, we obtain
\begin{align}\label{limsuppi1}
\limsup\limits_{k\rightarrow\infty}E[\widetilde{\Pi}(mh)]
\leq&\lim\limits_{m\rightarrow\infty}\frac{\frac{c^2((m+1)h)}{\alpha^2((m+1)h)}\tilde{p}(mh)}
{(2\theta-\tau_{0}) c((m+1)h)-\tilde{q}_{1}(mh)}\cr
&=\frac{h^2C_{\rho}\Big(2\sigma^2_{d}\Big(\frac{2NC_{d}^2}{\mu^2}+1\Big)
+2NC^2_{d}\Big)}{\tau_{0}(2\theta-\tau_{0})}.
\end{align}
By  (\ref{Edd}), we get
\begin{align}
&2\alpha(k)E\left[\left\|d(k)\right\|\left\|\delta(k)\right\|\right]\cr
\leq& E\left[\|\alpha(k)d(k)\|^2\right]+E\left[\|\delta(k)\|^2\right]\cr
\leq& \alpha^2(k)E\left[2\sigma^2_{d}\|X(k)\|^2+2NC^2_{d}\right]+E[V(k)],\notag
\end{align}
which together with
\Cref{lemma2} (i) and \Cref{lemma3(iii)} gives
\begin{align}\label{EV}
E[V(k+1)]\leq&2\big(1+c^2(k)(\rho^2_{0}+8\sigma^2C_{\xi}\rho_{1})\big)E[V(k)]+8b^2C_{\xi}\rho_{1}c^2(k)\cr
&+2\alpha^2(k)\Big((2\sigma_{\zeta}+4\sigma^2_{d})
\Big(\frac{2NC_{d}^2}{\mu^2}+1\Big)
+\left(2C_{\zeta}+4NC^2_{d}\right)\Big),\ k\geq k_{0}.
\end{align}
From the above inequality, \eqref{limsuppi1} and \Cref{lemma0}, we have
\begin{align}
\limsup_{k\rightarrow\infty}\frac{c^2(k+1)E[V(k+1)]}{\alpha^2(k+1)}
\leq\frac{\eta^hh^2C_{\rho}\Big(2\sigma^2_{d}\Big(\frac{2NC_{d}^2}{\mu^2}+1\Big)
+2NC^2_{d}\Big)}{\tau_{0}(2\theta-\tau_{0})}.\notag
\end{align}
From the arbitrariness of $\tau_{0}$ in $(0,2\theta)$ and  the above inequality, we get \eqref{EVkorder2str1}.}

{(ii). If $3\gamma_{1}\leq2\gamma_{2}$, then $\frac{\alpha^{2}(k)}{c(k)}=O(c^2(k))$. Let $\widehat{\Pi}(k)=\frac{V(k)}{c(k)}$. By $c(k)=\frac{c_{0}}{(k+1)^{\gamma_{1}}}$, we have
\begin{align}
\frac{c(k)}{c(k+h)}
=\frac{\frac{c_{0}}{(k+1)^{\gamma_{1}}}}{\frac{c_{0}}{(k+1+h)^{\gamma_{1}}}}
=\Big(1+\frac{h}{k+1}\Big)^{\gamma_{1}}
\leq1+\frac{\gamma_{1}h}{k+1}=1+o(c(k)).\notag
\end{align}
Then, from \eqref{Vmh2cstr} and \eqref{hatq0}, we obtain
\begin{align}\label{pi11strong3}
&E\left[\widehat{\Pi}((m+1)h)\right]\notag\\
=& E\left[\frac{V((m+1)h)}{c((m+1)h)}\right]\cr
\leq&\frac{c(mh)}{c((m+1)h)}\left[1-(2\theta-\tau_{0}) c((m+1)h)+\tilde{q}_{0}(mh)\right]E\left[\widehat{\Pi}(mh)\right]
+\frac{\tilde{p}(mh)}{c((m+1)h)}\cr
\leq&\big(1+o(c(mh))\big)\Big[1-(2\theta-\tau_{0}) c((m+1)h)+\tilde{q}_{0}(mh)\Big]E\left[\widehat{\Pi}(mh)\right]
+\frac{\tilde{p}(mh)}{c((m+1)h)}\cr
=&\Big[1-(2\theta-\tau_{0}) c((m+1)h)+\tilde{q}_{2}(mh)\Big]E\left[\widehat{\Pi}(mh)\right]+\frac{\tilde{p}(mh)}{c((m+1)h)},\ m\geq\max\big\{m_{0},m_{1},\big\lceil\dfrac{k_{0}}{h}\big\rceil\big\},
\end{align}
where $\tilde{q}_{2}(mh)=\tilde{q}_{0}(mh)+o(c(mh))\big(1-(2\theta-\tau_{0}) c((m+1)h)+\tilde{q}_{0}(mh)o(c(mh))\big)$, then $\tilde{q}_{2}(mh)=o(c((m+1)h))$. Thus, there exists a positive integer $m_{2}$, such that
\begin{align}\label{hatq2}
0<(2\theta-\tau_{0}) c((m+1)h)-\tilde{q}_{2}(mh)\leq1,\ \forall\ m\geq m_{2},\
\text{and}
\sum\limits\limits_{m=0}^{\infty}[(2\theta-\tau_{0}) c((m+1)h)-\tilde{q}_{2}(mh)]=\infty.
\end{align}
Since $c(k)=\frac{c_{0}}{(k+1)^{\gamma_{1}}}$,  $\alpha(k)=\frac{\alpha_{0}}{(k+1)^{\gamma_{2}}}$, and from the definition of $\tilde{p}(mh)$ and  $3\gamma_{1}\leq2\gamma_{2}$, $\gamma_{1}<\gamma_{2}$, we have
\begin{align}\label{p2mh}
&\lim\limits_{m\rightarrow\infty}\frac{\tilde{p}(mh)}{c^{2}((m+1)h)}\cr
=&h^2C_{\rho}\Big(2\sigma^2_{d}\Big(\frac{2NC_{d}^2}{\mu^2}+1\Big)
+2NC^2_{d}\Big)\lim\limits_{m\rightarrow\infty}\frac{\alpha^2(mh)}{\tau_{0} c^{3}((m+1)h)}\cr
&+4h^2\rho_{1}C_{\rho}C_{\xi}\Big(4\sigma^2\Big(\frac{2NC_{d}^2}{\mu^2}+1\Big)
+2b^2\Big)\lim\limits_{m\rightarrow\infty}\frac{c^2(mh)}{c^{2}((m+1)h)}\cr
&+4h^2C_{\rho}\bigg(\Big(\sigma^2_{d}+\sigma_{\zeta}\Big)\Big(\frac{2NC_{d}^2}{\mu^2}+1\Big)
+NC^2_{d}+C_{\zeta}\bigg)\lim\limits_{m\rightarrow\infty}\frac{\alpha^2(mh)}{c^{2}((m+1)h)}\cr
\leq&h^2C_{\rho}\Big(2\sigma^2_{d}\Big(\frac{2NC_{d}^2}{\mu^2}+1\Big)
+2NC^2_{d}\Big)\frac{\alpha_{0}^{2}}{\tau_{0}c_{0}^{3}}+4h^2\rho_{1}C_{\rho}C_{\xi}\Big(4\sigma^2\Big(\frac{2NC_{d}^2}{\mu^2}+1\Big)
+2b^2\Big).
\end{align}
By \eqref{pi11strong3}-\eqref{p2mh} and Lemma 1.2.25 in \cite{Guo} , we derive
\begin{align}
\limsup\limits_{k\rightarrow\infty}E[\widehat{\Pi}(mh)]
\leq&\lim\limits_{m\rightarrow\infty}\frac{\frac{\tilde{p}(mh)}{c((m+1)h)}}
{(2\theta-\tau_{0}) c((m+1)h)-\tilde{q}_{2}(mh)}\cr
\leq&\frac{h^2C_{\rho}\Big(2\sigma^2_{d}\Big(\frac{2NC_{d}^2}{\mu^2}+1\Big)
+2NC^2_{d}\Big)\alpha_{0}^{2}}{(2\theta-\tau_{0})\tau_{0}c_{0}^{3}}+\frac{4h^2\rho_{1}C_{\rho}C_{\xi}\Big(4\sigma^2\Big(\frac{2NC_{d}^2}{\mu^2}+1\Big)
+2b^2\Big)}{2\theta-\tau_{0}},\notag
\end{align}
which together with \eqref{EV} and \Cref{lemma0} gives \eqref{EVkorder2str2}.}

{In the following part, we will prove \eqref{EVkorder2str3}. From \eqref{Vmh2cstr} with $\tau_{0}=\theta$, we have
\begin{align}\label{v1}
E[V((m+1)h)]
\leq\Big[1-\theta c((m+1)h)+\tilde{q}_{0}(mh)\Big]E[V(mh)]
+\tilde{p}(mh),\ m\geq\left\lceil\frac{T_{0}}{h}\right\rceil,
\end{align}
where $\tilde{q}_{0}(mh)=\max\{0,(C_{4}-2\theta^2) c^2((m+1)h)+\theta C_{4}c^3((m+1)h)\}$, $\tilde{p}(mh)=\left(\dfrac{1}{\theta c((m+1)h)}+2\right)\times$ $\bigg(h^2C_{\rho}\Big(2NC^2_{d}+$
$2\sigma^2_{d}C_{X1}\Big)\bigg)\alpha^2(mh)
+4h^2\rho_{1}C_{\rho}C_{\xi}\Big(4\sigma^2C_{X1}
+2b^2\Big)c^2(mh)
+4h^2C_{\rho}\Big(\sigma_{\zeta}C_{X1}+C_{\zeta}\Big)\alpha^2(mh)$.
Let $T_{1}=\lceil\theta c_{0}\rceil-1$, then $1-\theta c((m+1)h)+\tilde{q}_{0}(mh)\geq 0,\ \forall\ m\geq \lceil\frac{T_{1}}{h}\rceil-1$. Let $\hat{m}=\max\{\lceil\frac{T_{0}}{h}\rceil,\lceil\frac{T_{1}}{h}\rceil-1\}$, from \eqref{v1}, we have
\begin{align}
&E[V((m+1)h)]\notag\\
\leq&\prod\limits_{t=\hat{m}}^{m}\big(1-\theta c((t+1)h)+\tilde{q}_{0}(th)\big)E[V(\hat{m}h)]+\sum_{s=\hat{m}}^{m}\tilde{p}(sh)\prod\limits_{t=s+1}^{m}\big(1-\theta c((t+1)h)+\tilde{q}_{0}(th)\big)\cr
\leq&\exp\left(\sum_{t=\hat{m}}^{m}\tilde{q}_{0}(th)\right)\exp\left(\sum_{t=\hat{m}}^{m}-\theta c((t+1)h)\right)E[V(\hat{m}h)]\cr
&+\sum_{s=\hat{m}}^{m}\tilde{p}(sh)\exp\left(\sum_{t=s+1}^{m}\tilde{q}_{0}(th)\right)\exp\left(\sum_{t=s+1}^{m}-\theta c((t+1)h)\right)\cr
\leq&C_{q}\exp\left(-\theta\sum_{t=\hat{m}}^{m} c((t+1)h)\right)E[V(\hat{m}h)]+\sum_{s=\hat{m}}^{m}\tilde{p}(sh)C_{q}\exp\left(-\theta\sum_{t=s+1}^{m} c((t+1)h)\right),\ m\geq\hat{m}.\notag
\end{align}
Since $c(k)=\frac{c_{0}}{(k+1)^{\gamma_{1}}}$, we have $\sum_{t=a}^{m}c((t+1)h)\geq\int_{a}^{m+1}c(t)dt=$ $\frac{c_{0}}{h(1-\gamma_{1})}(((m+2)h+1)^{1-\gamma_{1}}-
((a+1)h+1)^{1-\gamma_{1}})$, which together with the above inequality gives
\begin{align}\label{v2}
&E\left[V((m+1)h)\right]\cr
\leq&C_{q}\exp\left(\frac{c_{0}\theta((\hat{m}+1)h+1
)^{1-\gamma_{1}}}{h(1-\gamma_{1})}\right)E\left[V(\hat{m}h)\right]
\exp\left(-\frac{c_{0}\theta((m+2)h+1)^{1-\gamma_{1}}}{h(1-\gamma_{1})}\right)\cr
&+\sum_{s=\hat{m}}^{m}\tilde{p}(sh)C_{q}\exp\left(\frac{c_{0}\theta((s+2)h+1)^{1-\gamma_{1}}}{h(1-\gamma_{1})}\right)
\exp\left(-\frac{c_{0}\theta((m+2)h+1)^{1-\gamma_{1}}}{h(1-\gamma_{1})}\right),\ m\geq\hat{m}.
\end{align}
From $c(k)=\frac{c_{0}}{(k+1)^{\gamma_{1}}}$, $\alpha(k)=\frac{\alpha_{0}}{k+1}$, we have
\begin{align}
\label{v3}
\tilde{p}(mh)\leq C_{p}c^2(mh),
\end{align}
where $C_{p}=\left(\frac{1}{\theta}+2\right)\bigg(h^2C_{\rho}\Big(2NC^2_{d}+
2\sigma^2_{d}C_{X1}\Big)\bigg)\frac{\alpha_{0}^2}{c_{0}^{3}}(1+h)$
$+4h^2\rho_{1}C_{\rho}C_{\xi}\Big(4\sigma^2C_{X1}
+2b^2\Big)
+4h^2C_{\rho}\Big(\sigma_{\zeta}C_{X1}+C_{\zeta}\Big)\frac{\alpha_{0}^2}{c_{0}^2}$.
From \eqref{v3}, $\sum_{s=\hat{m}}^{m}\tilde{p}(sh)C_{q}\exp\left(\frac{c_{0}\theta((s+2)h+1)^{1-\gamma_{1}}}{h(1-\gamma_{1})}\right)\leq$
$\int_{\hat{m}}^{m+1}\tilde{p}(sh)C_{q}\exp\left(\frac{c_{0}\theta((s+2)h+1)^{1-\gamma_{1}}}{h(1-\gamma_{1})}\right)ds$ and Cauchy Integral Mean Value Theorem, we have there exists  $m^{\prime}\in[\hat{m},m+1]$, such that
\begin{align}
&(mh+1)^{\gamma_{1}}\sum_{s=\hat{m}}^{m}\tilde{p}(sh)C_{q}\exp\left(\frac{c_{0}\theta((s+2)h+1)^{1-\gamma_{1}}}{h(1-\gamma_{1})}\right)
\exp\left(-\frac{c_{0}\theta((m+2)h+1)^{1-\gamma_{1}}}{h(1-\gamma_{1})}\right)\cr
\leq&\frac{C_{q}C_{p}\int_{\hat{m}}^{m+1}\frac{c_{0}^2}{(sh+1)^{2\gamma_{1}}}
\exp\left(\frac{c_{0}\theta((s+2)h+1)^{1-\gamma_{1}}}{h(1-\gamma_{1})}\right)ds}
{(mh+1)^{-\gamma_{1}}\exp\left(\frac{c_{0}
\theta((m+2)h+1)^{1-\gamma_{1}}}{h(1-\gamma_{1})}\right)
-(\hat{m}h+1)^{-\gamma_{1}}\exp\left(\frac{c_{0}\theta((\hat{m}+2)h+1)^{1-\gamma_{1}}}{h(1-\gamma_{1})}\right)}\cr
=&\frac{C_{q}C_{p}\frac{c_{0}^2}{(m'h+1)^{2\gamma_{1}}}\exp\left(\frac{c_{0}\theta((m'+2)h+1)^{1-\gamma_{1}}}{h(1-\gamma_{1})}\right)}
{(c_{0}\theta((m'+2)h+1)^{-\gamma_{1}}(m'h+1)^{-\gamma_{1}}-\gamma_{1}h(m'h+1)^{-\gamma_{1}-1})\exp\left(\frac{c_{0}\theta((m'+2)h+1)^{1-\gamma_{1}}}{h(1-\gamma_{1})}\right)}\cr
=&\frac{C_{q}C_{p}c_{0}^2}
{\left(c_{0}\theta\Big(\frac{(m'+2)h+1}{m'h+1}\Big)^{
-\gamma_{1}}-\gamma_{1}h(m'h+1)^{\gamma_{1}-1}\right)}\cr
\leq&\frac{C_{q}C_{p}c_{0}^2}
{c_{0}\theta (1+2h )^{-\gamma_{1}}-\gamma_{1}h}.\notag
\end{align}
From \eqref{v2} and the above inequality,  \eqref{EV}  and \Cref{lemma0}, we get \eqref{EVkorder2str3}.}
\end{proof}

\noindent
{{\bfseries Proof of \Cref{theorem3.2}: }}


{In the following part, we will analyze the convergence rate of $E\left[\left\|X(k)-\mathbf{1}_{N}\otimes x^{*}\right\|^2\right]$ based on \Cref{lemma2(iii)}. Let $\tilde{k}=k_{0}$. From  \Cref{lemma2(iii)} and \Cref{lemma3(iii)}, we have
\begin{align}
\label{str2}
&E\left[\left\|X(k+1)-\mathbf{1}_{N}\otimes z^{*}\right\|^2\right]\cr
\leq&\prod\limits_{t=\tilde{k}}^{k}\Big(1-\frac{\mu}{2N}
\alpha(t)+\varphi_{1}(t)\Big)E\left[\left\|X(\tilde{k})-\mathbf{1}_{N}\otimes z^{*}\right\|^2\right]+\sum_{s=\tilde{k}}^{k}\varphi_{2}(s)\prod\limits_{t=s+1}^{k}\Big(1-\frac{\mu}{2N}\alpha(t)+\varphi_{1}(t)\Big)\cr
\leq&\exp\left(\sum\limits\limits_{t=\tilde{k}}^{k}\varphi_{1}(t)\right)
\exp\left(-\frac{\mu}{2N}\sum\limits\limits_{t=\tilde{k}}^{k}\alpha(t)\right)
E\left[\left\|X(\tilde{k})-\mathbf{1}_{N}\otimes z^{*}\right\|^2\right]\cr
&+\sum_{s=\tilde{k}}^{k}\varphi_{2}(s)\exp\left(\sum\limits\limits_{t=s+1}^{k}
\varphi_{1}(t)\right)
\exp\left(-\frac{\mu}{2N}\sum\limits\limits_{t=s+1}^{k}\alpha(t)\right)\cr
\leq&C_{\varphi_{1}}\exp\left(-\frac{\mu}{2N}\sum\limits\limits_{t=\tilde{k}}
^{k}\alpha(t)\right)
E\left[\left\|X(\tilde{k})-\mathbf{1}_{N}\otimes z^{*}\right\|^2\right]\cr
&+C_{\varphi_{1}}\sum_{s=\tilde{k}}^{k}\varphi_{2}(s)\exp
\left(-\frac{\mu}{2N}\sum\limits\limits_{t=s+1}^{k}\alpha(t)\right),\ k\geq \tilde{k},
\end{align}
where $C_{\varphi_{1}}=\sum\limits\limits_{k=0}^{\infty}\varphi_{1}(k)$. Since $c(k)=\frac{c_{0}}{(k+1)^{\gamma_{1}}}$,  $\alpha(k)=\frac{\alpha_{0}}{(k+1)^{\gamma_{2}}}$, we know that $C_{\varphi_{1}}$ exists.}

{By $\alpha(k)=\frac{\alpha_{0}}{(k+1)^{\gamma_{2}}}$, we have
$\int_{s+1}^{k+1}\alpha(x)dx=\sum\limits\limits_{t=s+1}^{k}\int_{t}^{t+1}\alpha(x)dx\leq\sum\limits\limits_{t=s+1}^{k}\alpha(t)$,
which means that
\begin{align}\label{str3}
&\exp\left(-\frac{\mu}{2N}\sum\limits\limits_{t=s+1}^{k}\alpha(t)\right)\cr
\leq&\exp\left(-\frac{\mu}{2N}\int_{s+1}^{k+1}\frac{\alpha_{0}}
{(t+1)^{\gamma_{2}}}dt\right)\cr
\leq&
\left\{
   \begin{aligned}
   &\exp\left(\frac{\mu\alpha_{0}}{2N(1-\gamma_{2})}(s+2)^{1-\gamma_{2}}\right)
\exp\left(\frac{-\mu\alpha_{0}}{2N(1-\gamma_{2})}(k+2)^{1-\gamma_{2}}\right),\ \gamma_{2}\neq1, \cr
   &\exp\left(-\frac{\mu\alpha_{0}}{2N}\ln\left(\frac{k+2}{s+2}\right)\right)
=(s+2)^{\frac{\mu\alpha_{0}}{2N}}(k+2)^{-\frac{\mu\alpha_{0}}{2N}},\ \ \ \ \ \ \gamma_{2}=1.
   \end{aligned}
   \right.
\end{align}}

{($\bullet$) If $3\gamma_{1}>2\gamma_{2}$, by \Cref{lemmac(iii)} (i), we have
\begin{align}
&\limsup_{k\rightarrow\infty}\frac{\varphi_{2}(k)c(k)}{\alpha^2(k)}\cr
=&8b^2C_{\xi}\rho_{1}\limsup_{k\rightarrow\infty}\frac{c^3(k)}{\alpha^{2}(k)}
+\limsup_{k\rightarrow\infty}2c(k)\big(2C_{\zeta}+3NC^2_{d}+2(3\sigma^2_{d}+2\sigma_{\zeta})N\|z^{*}\|^2\big)
\cr
&+\limsup_{k\rightarrow\infty}\frac{\mu}{N}\frac{c^2(k)E[V(k)]}{\alpha^2(k)}
\frac{\alpha(k)}{c(k)}+2\sqrt{2\left(\frac{1}{N}\sigma^2_{d}\left(
\frac{2NC_{d}^2}{\mu^2}+1\right)+C^2_{d}\right)}\limsup_{k\rightarrow\infty}\sqrt{\frac{c^2(k)E[V(k)]}{\alpha^2(k)}}\cr
\leq&2\sqrt{2C_{V1}\left(\frac{1}{N}\sigma^2_{d}\left(\frac{2NC_{d}^2}{\mu^2}+1\right)
+C^2_{d}\right)}.\notag
\end{align}
From $c(k)=\frac{c_{0}}{(k+1)^{\gamma_{1}}}$,  $\alpha(k)=\frac{\alpha_{0}}{(k+1)^{\gamma_{2}}}$ and the above inequality, there exists a positive integer $k_{5}$, such that
\begin{align}\label{varphi20}
\varphi_{2}(k)\leq C_{\varphi_{2}}(k+2)^{\gamma_{1}-2\gamma_{2}},\ k\geq k_{5},
\end{align}
where $C_{\varphi_{2}}=4\frac{\alpha^2_{0}}{c_{0}}\Big(1+2\sqrt{2C_{V1}
\Big(\frac{1}{N}\sigma^2_{d}\Big(\frac{2NC_{d}^2}{\mu^2}+1\Big)+C^2_{d}\Big)}\Big)
$.}

{(1)\ $\gamma_{2}\in(\gamma_{1},1)$. From \eqref{str3} and \eqref{varphi20}, we have
\begin{align}
\label{str4}
&\sum_{s=\tilde{k}}^{k}\varphi_{2}(s)\exp\Big(-\frac{\mu}{2N}\sum\limits\limits_{t=s+1}^{k}\alpha(t)\Big)\cr
\leq&\sum_{s=\tilde{k}}^{k}C_{\varphi_{2}}(s+2)^{\gamma_{1}-2\gamma_{2}}
\exp\Big(\frac{\mu\alpha_{0}}{2N(1-\gamma_{2})}(s+2)^{1-\gamma_{2}}\Big)
\exp\Big(\frac{-\mu\alpha_{0}}{2N(1-\gamma_{2})}(k+2)^{1-\gamma_{2}}\Big)\cr
=&C_{\varphi_{2}}\exp\Big(\frac{-\mu\alpha_{0}}
{2N(1-\gamma_{2})}(k+2)^{1-\gamma_{2}}\Big)
\sum_{s=\tilde{k}}^{k}g(s),\ k\geq k_{5},
\end{align}
where $g(s)=(s+2)^{\gamma_{1}-2\gamma_{2}}\exp\Big(\frac{\mu\alpha_{0}}{2N(1-\gamma_{2})}(s+2)^{1-\gamma_{2}}\Big)$, then there exists  $k_{6}>0$, such that for any $s>k_{6}$, $g(s)$ is monotonically increasing. From Integral Mean Value Theorem, we have $\int_{s}^{s+1}g(t)dt\geq g(s)$, thus,
\begin{align}
\sum_{s=k_{6}}^{k}g(s)\leq\sum_{s=k_{6}}^{k}\int_{s}^{s+1}g(t)dt
=\int_{k_{6}}^{k+1}g(t)dt.\notag
\end{align}
Then, by L'Hospital's Rule, we get
\begin{align}
\label{str6}
&\lim\limits_{k\rightarrow\infty}\frac{\int_{k_{6}}^{k+1}g(t)dt}
{(k+2)^{\gamma_{1}-\gamma_{2}}\exp\Big(\frac{\mu\alpha_{0}}{2N(1-\gamma_{2})}(k+2)^{1-\gamma_{2}}\Big)}\cr
=&\lim\limits_{k\rightarrow\infty}\frac{(k+3)^{\gamma_{1}-2\gamma_{2}}
\exp\Big(\frac{\mu\alpha_{0}}{2N(1-\gamma_{2})}(k+3)^{1-\gamma_{2}}\Big)}
{\Big((\gamma_{1}-\gamma_{2})(k+2)^{\gamma_{2}-1}+\frac{\mu\alpha_{0}}
{2N}\Big)(k+2)^{\gamma_{1}-2\gamma_{2}}\exp\Big(\frac{\mu\alpha_{0}}{2N(1-\gamma_{2})}(k+2)^{1-\gamma_{2}}\Big)}\cr
=&\frac{1}
{\lim\limits_{k\rightarrow\infty}\Big((\gamma_{1}-\gamma_{2})(k+2)^{\gamma_{2}-1}+\frac{\mu\alpha_{0}}
{2N}\Big)}\lim\limits_{k\rightarrow\infty}\Big(\frac{k+3}{k+2}\Big)^{\gamma_{1}-2\gamma_{2}}\cr
&\times\lim\limits_{k\rightarrow\infty}\exp\Big(\frac{\mu\alpha_{0}}{2N(1-\gamma_{2})}\big((k+3)^{1-\gamma_{2}}-(k+2)^{1-\gamma_{2}}\big)\Big)\cr
=&\frac{2N}{\mu\alpha_{0}},
\end{align}
where the last ``$=$" is by $\lim\limits_{k\rightarrow\infty}\big((k+3)^{1-\gamma_{2}}-(k+2)^{1-\gamma_{2}}\big)$
$=\lim\limits_{k\rightarrow\infty}(k+2)^{-\gamma_{2}}\big((k+2)^{\gamma_{2}}(k+3)^{1-\gamma_{2}}-(k+2)\big)$
$\leq\lim\limits_{k\rightarrow\infty}(k+2)^{-\gamma_{2}}\big((k+3)^{\gamma_{2}}(k+3)^{1-\gamma_{2}}-(k+2)\big)$
$=\lim\limits_{k\rightarrow\infty}(k+2)^{-\gamma_{2}}=0$.
From \eqref{str4}-\eqref{str6}, we have
\begin{align}
\label{str411}
&\limsup_{k\rightarrow\infty} (k+2)^{\gamma_{2}-\gamma_{1}}\sum_{s=\tilde{k}}^{k}\varphi_{2}(s)\exp\Big(-\frac{\mu}{2N}\sum\limits\limits_{t=s+1}^{k}\alpha(t)\Big)\cr
\leq&\limsup_{k\rightarrow\infty}C_{\varphi_{2}}(k+2)^{\gamma_{2}-\gamma_{1}}\exp\Big(\frac{-\mu\alpha_{0}}
{2N(1-\gamma_{2})}(k+2)^{1-\gamma_{2}}\Big)
\sum_{s=\tilde{k}}^{k}g(s)\cr
\leq&\limsup_{k\rightarrow\infty}\frac{C_{\varphi_{2}}\int_{\max\{\tilde{k},k_{6}\}}^{k+1}g(t)dt}
{(k+2)^{\gamma_{1}-\gamma_{2}}\exp\Big(\frac{\mu\alpha_{0}}{2N(1-\gamma_{2})}(k+2)^{1-\gamma_{2}}\Big)}
+\limsup_{k\rightarrow\infty}\frac{(k+2)^{\gamma_{2}-\gamma_{1}}C_{\varphi_{2}}\sum_{s=\tilde{k}}^{\max\{\tilde{k},k_{6}\}}g(s)}
{\exp\Big(\frac{\mu\alpha_{0}}{2N(1-\gamma_{2})}(k+2)^{1-\gamma_{2}}\Big)}\cr
=&\frac{2NC_{\varphi_{2}}}{\mu\alpha_{0}}.
\end{align}
From $\int_{\tilde{k}}^{k+1}\alpha(x)dx=\sum\limits\limits_{t=\tilde{k}}^{k}\int_{t}^{t+1}\alpha(x)dx\leq\sum\limits\limits_{t=\tilde{k}}^{k}\alpha(t)$, we have
\begin{align}
\label{str42}
&\limsup_{k\rightarrow\infty}(k+2)^{\gamma_{2}-\gamma_{1}}\exp\Big(-\frac{\mu}{2N}\sum\limits\limits_{t=\tilde{k}}^{k}\alpha(t)\Big)\cr
\leq&\limsup_{k\rightarrow\infty}(k+2)^{\gamma_{2}-\gamma_{1}}\exp\Big(-\frac{\mu}{2N}\int_{\tilde{k}}^{k+1}\frac{\alpha_{0}}{(t+1)^{\gamma_{2}}}dt\Big)\cr
=&\exp\Big(\frac{\mu\alpha_{0}}{2N(1-\gamma_{2})}(\tilde{k}+1)^{1-\gamma_{2}}\Big)
\limsup_{k\rightarrow\infty}\frac{(k+2)^{\gamma_{2}
-\gamma_{1}}}{\exp\Big(\frac{\mu\alpha_{0}}{2N
(1-\gamma_{2})}(k+2)^{1-\gamma_{2}}\Big)}
=0.
\end{align}
From \eqref{str2}, \eqref{str411} and \eqref{str42}, we get \Cref{theorem3.2} (1).}

{If $\gamma_{2}=1$,
from \eqref{str3} and \eqref{varphi20}, we get
\begin{align}
\label{str9}
&\sum_{s=\tilde{k}}^{k}\varphi_{2}(s)\exp\Big(-\frac{\mu}{2N}\sum\limits\limits_{t=s+1}^{k}\alpha(t)\Big)\cr
\leq&C_{\varphi_{2}}(k+2)^{-\frac{\mu\alpha_{0}}{2N}}
\sum_{s=\tilde{k}}^{k}(s+2)^{\gamma_{1}-2}(s+2)^{\frac{\mu\alpha_{0}}{2N}}\cr
\leq&C_{\varphi_{2}}(k+2)^{-\frac{\mu\alpha_{0}}{2N}}\int_{\tilde{k}-1}^{k+1}(s+2)^{\gamma_{1}-2+\frac{\mu\alpha_{0}}{2N}}ds\cr
\leq&
\left\{
   \begin{aligned}
   &C_{\varphi_{2}}(k+2)^{-\frac{\mu\alpha_{0}}{2N}}\ln(k+3),\ \ \gamma_{1}+\frac{\mu\alpha_{0}}{2N}=1, \cr
   &\frac{C_{\varphi_{2}}}{\gamma_{1}+\frac{\mu\alpha_{0}}{2N}-1}(k+2)^{-\frac{\mu\alpha_{0}}{2N}}
\Big((k+3)^{\gamma_{1}+\frac{\mu\alpha_{0}}{2N}-1}-(\tilde{k}+1)^{\gamma_{1}+\frac{\mu\alpha_{0}}{2N}-1}\Big),\ \gamma_{1}+\frac{\mu\alpha_{0}}{2N}\neq1,
   \end{aligned}
   \right.
\end{align}
where the 2nd ``$\leq$'' is by $\sum_{s=a}^{b}g(s)\leq\int_{a-1}^{b+1}g(s)ds$, and $g(s)=(s+2)^{\gamma_{1}+\frac{\mu\alpha_{0}}{2N}-2}$ is a nonnegative monotonic function.}

{(2)\ If $\gamma_{1}+\frac{\mu\alpha_{0}}{2N}>1$, then $\frac{C_{\varphi_{2}}}{\gamma_{1}+\frac{\mu\alpha_{0}}{2N}-1}>0$ and $1-\gamma_{1}
-\frac{\mu\alpha_{0}}{2N}<0$. By \eqref{str9}, we obtain
\begin{align}
\sum_{s=\tilde{k}}^{k}\varphi_{2}(s)\exp\Big(-\frac{\mu}{2N}\sum\limits\limits_{t=s+1}^{k}\alpha(t)\Big)
\leq&\frac{C_{\varphi_{2}}}{\gamma_{1}+\frac{\mu\alpha_{0}}{2N}-1}(k+2)^{-\frac{\mu\alpha_{0}}{2N}}
(k+3)^{\gamma_{1}+\frac{\mu\alpha_{0}}{2N}-1},\ k\geq k_{5},
\end{align}
which together with \eqref{str2} and \eqref{str3} gives
\begin{align}
&\limsup_{k\rightarrow\infty}(k+2)^{1-\gamma_{1}} E\left[\left\|X(k+1)-\mathbf{1}_{N}\otimes z^{*}\right\|^2\right]\cr
\leq&\limsup_{k\rightarrow\infty}C_{\varphi_{1}}\Big((\tilde{k}+1)^{\frac{\mu\alpha_{0}}{2N}}(k+2)^{1-\gamma_{1}
-\frac{\mu\alpha_{0}}{2N}}\Big)E\left[\left\|X(\tilde{k})-\mathbf{1}_{N}\otimes z^{*}\right\|^2\right]\notag\\
&+\frac{ C_{\varphi_{1}}C_{\varphi_{2}}}{\gamma_{1}+\frac{\mu\alpha_{0}}{2N}-1}
\limsup_{k\rightarrow\infty}\Big(\frac{k+3}{k+2}\Big)^{\gamma_{1}
+\frac{\mu\alpha_{0}}{2N}-1}
= \frac{ C_{\varphi_{1}}C_{\varphi_{2}}}{\gamma_{1}+\frac{\mu\alpha_{0}}{2N}-1}.
\end{align}}
\ \ {(3)\ If $\gamma_{1}+\frac{\mu\alpha_{0}}{2N}=1$, by \eqref{str2}, \eqref{str3} and \eqref{str9}, we have
\begin{align}
&\limsup_{k\rightarrow\infty}(k+2)^{1-\gamma_{1}}(\ln(k+2))^{-1}
E\left[\left\|X(k+1)-\mathbf{1}_{N}\otimes z^{*}\right\|^2\right]\cr
\leq&\limsup_{k\rightarrow\infty}C_{\varphi_{1}}\Big((\tilde{k}+1)^{\frac{\mu\alpha_{0}}{2N}}
(k+2)^{1-\gamma_{1}-\frac{\mu\alpha_{0}}{2N}}(\ln(k+2))^{-1}\Big)
E\left[\left\|X(\tilde{k})-\mathbf{1}_{N}\otimes z^{*}\right\|^2\right]\cr
&+\limsup_{k\rightarrow\infty}C_{\varphi_{1}}C_{\varphi_{2}}(k+2
)^{1-\gamma_{1}-\frac{\mu\alpha_{0}}{2N}}\ln(k+3)(\ln(k+2))^{-1}
= C_{\varphi_{1}}C_{\varphi_{2}}.
\end{align}}
\ \ {(4)\ If $\gamma_{1}+\frac{\mu\alpha_{0}}{2N}<1$, then $\frac{C_{\varphi_{2}}}{\gamma_{1}+\frac{\mu\alpha_{0}}{2N}-1}<0$. By \eqref{str9}, we obtain
\begin{align}
\label{str11++}
\sum_{s=\tilde{k}}^{k}\varphi_{2}(s)\exp\Big(-\frac{\mu}{2N}\sum\limits\limits_{t=s+1}^{k}\alpha(t)\Big)
\leq&\frac{C_{\varphi_{2}}}{1-\gamma_{1}-\frac{\mu\alpha_{0}}{2N}}(k+2)^{-\frac{\mu\alpha_{0}}{2N}}
\Big((\tilde{k}+1)^{\gamma_{1}+\frac{\mu\alpha_{0}}{2N}-1}-(k+3)^{\gamma_{1}+\frac{\mu\alpha_{0}}{2N}-1}\Big)\cr
\leq&\frac{C_{\varphi_{2}}}{1-\gamma_{1}-\frac{\mu\alpha_{0}}{2N}}(k+2)^{-\frac{\mu\alpha_{0}}{2N}}
(\tilde{k}+1)^{\gamma_{1}+\frac{\mu\alpha_{0}}{2N}-1}\cr
\leq&\frac{C_{\varphi_{2}}}{1-\gamma_{1}-\frac{\mu\alpha_{0}}{2N}}(k+2)^{-\frac{\mu\alpha_{0}}{2N}},\ k\geq k_{5},
\end{align}
where the last $"\leq"$ is obtained from $(\tilde{k}+1)^{\gamma_{1}+\frac{\mu\alpha_{0}}{2N}-1}<1$. Therefore, from \eqref{str2}, \eqref{str3} and \eqref{str11++}, we derive
\begin{align}
\limsup_{k\rightarrow\infty}(k+2)^{\frac{\mu\alpha_{0}}{2N}}
E\left[\left\|X(k+1)-\mathbf{1}_{N}\otimes z^{*}\right\|^2\right]
\leq&C_{\varphi_{1}}(\tilde{k}+1)^{\frac{\mu\alpha_{0}}{2N}}
E\left[\left\|X(\tilde{k})-\mathbf{1}_{N}\otimes z^{*}\right\|^2\right]
+\frac{C_{\varphi_{1}}C_{\varphi_{2}}}{1-\gamma_{1}-\frac{\mu\alpha_{0}}{2N}},\notag
\end{align}
from above, we get \Cref{theorem3.2} (4).}

{($\bullet\bullet$) If $3\gamma_{1}\leq2\gamma_{2}$, by \Cref{lemmac(iii)} (ii), we have
\begin{align}\label{tildevarphi20}
\limsup_{k\rightarrow\infty}\frac{\varphi_{2}(k)}{c^2(k)}
=&8b^2C_{\xi}\rho_{1}+2\big(2C_{\zeta}+3NC^2_{d}+2(3\sigma^2_{d}+2\sigma_{\zeta})N\|z^{*}\|^2\big)
\limsup_{k\rightarrow\infty}\frac{\alpha^2(k)}{c^2(k)}\cr
&+\limsup_{k\rightarrow\infty}\frac{\mu}{N}\frac{E[V(k)]}{c(k)}\frac{\alpha(k)}{c(k)}\cr
&+2\sqrt{2\Big(\frac{1}{N}\sigma^2_{d}\Big(\frac{2NC_{d}^2}{\mu^2}+1\Big)+C^2_{d}\Big)}
\limsup_{k\rightarrow\infty}\sqrt{\frac{E[V(k)]}{c(k)}\frac{\alpha^2(k)}{c^3(k)}}\cr
\leq&8b^2C_{\xi}\rho_{1}+\frac{2\alpha^2_{0}}{c_{0}^2}\sqrt{2C_{V2}
\left(\frac{1}{N}\sigma^2_{d}\left(\frac{2NC_{d}^2}{\mu^2}+1\right)+C^2_{d}\right)}.
\end{align}
From $c(k)=\frac{c_{0}}{(k+1)^{\gamma_{1}}}$ and \eqref{tildevarphi20}, there exists a positive integer $k_{7}$, such that
\begin{align}\label{varphi22}
\varphi_{2}(k)\leq C'_{\varphi_{2}}(k+2)^{-2\gamma_{1}},\ k\geq k_{7},
\end{align}
where $C'_{\varphi_{2}}=
4\left(8b^2C_{\xi}\rho_{1}+\frac{2\alpha^2_{0}}{c_{0}^2}\sqrt{2C_{V2}
\left(\frac{1}{N}\sigma^2_{d}\left(\frac{2NC_{d}^2}{\mu^2}+1\right)
+C^2_{d}\right)}+1\right)c^2_{0}$. Let $\hat{k}=\max\{\tilde{k},k_{7}\}$.}

{(5)\ If $\gamma_{2}\in(\frac{3}{2}\gamma_{1},1)$, by \eqref{str3} and \eqref{varphi22}, we have
\begin{align}
&\sum\limits\limits\limits_{s=\hat{k}}^{k}\varphi_{2}(s)\exp\Big(-\frac{\mu}{2N}\sum\limits\limits_{t=s+1}^{k}\alpha(t)\Big)\cr
\leq&\sum\limits\limits\limits_{s=\hat{k}}^{k}C'_{\varphi_{2}}(s+2)^{-2\gamma_{1}}\exp\Big(\frac{\mu\alpha_{0}}
{2N(1-\gamma_{2})}(s+2)^{1-\gamma_{2}}\Big)\exp\Big(\frac{-\mu\alpha_{0}}
{2N(1-\gamma_{2})}(k+2)^{1-\gamma_{2}}\Big)\cr
=&C'_{\varphi_{2}}\exp\Big(\frac{-\mu\alpha_{0}}
{2N(1-\gamma_{2})}(k+2)^{1-\gamma_{2}}\Big)\sum\limits\limits\limits_{s=\hat{k}}^{k}(s+2)^{-2\gamma_{1}}
\exp\Big(\frac{\mu\alpha_{0}}{2N(1-\gamma_{2})}(s+2)^{1-\gamma_{2}}\Big),\ k\geq \hat{k}.
\end{align}
Similar to \eqref{str6}-\eqref{str42}, we have
\begin{align}\label{str71}
\limsup_{k\rightarrow\infty}(k+1)^{2\gamma_{1}-\gamma_{2}}E\left[\left\|X(k+1)-\mathbf{1}_{N}\otimes z^{*}\right\|^2\right]\leq\frac{2NC_{\varphi_{1}}C'_{\varphi_{2}}}{\mu\alpha_{0}}.
\end{align}}
{If $\gamma_{2}=1$, by \eqref{str3} and \eqref{varphi22}, we derive
\begin{align}\label{str91}
&\sum\limits\limits\limits_{s=\hat{k}}^{k}\varphi_{2}(s)\exp\Big(-\frac{\mu}{2N}\sum\limits\limits_{t=s+1}^{k}\alpha(t)\Big)\cr
\leq&C'_{\varphi_{2}}(k+2)^{-\frac{\mu\alpha_{0}}{2N}}
\sum\limits\limits\limits_{s=\hat{k}}^{k}(s+2)^{-2\gamma_{1}}(s+2)^{\frac{\mu\alpha_{0}}{2N}}\cr
\leq&C'_{\varphi_{2}}(k+2)^{-\frac{\mu\alpha_{0}}{2N}}
\int_{\hat{k}-1}^{k+1}(s+1)^{-2\gamma_{1}+\frac{\mu\alpha_{0}}{2N}}ds\cr
\leq&\left\{
   \begin{aligned}
   &C'_{\varphi_{2}}(k+2)^{-\frac{\mu\alpha_{0}}{2N}}\ln(k+2),\ \ \frac{\mu\alpha_{0}}{2N}=2\gamma_{1}-1, \cr
   &\frac{C'_{\varphi_{2}}}{1-2\gamma_{1}+\frac{\mu\alpha_{0}}{2N}}(k+2)^{-\frac{\mu\alpha_{0}}{2N}}
\left((k+2)^{1-2\gamma_{1}+\frac{\mu\alpha_{0}}{2N}}-\hat{k}^{1-2\gamma_{1}
+\frac{\mu\alpha_{0}}{2N}}\right),\ \frac{\mu\alpha_{0}}{2N}\neq2\gamma_{1}-1.
   \end{aligned}
   \right.
\end{align}
Similar to the proof of \Cref{theorem3.2} (2), (3) and (4), we get \Cref{theorem3.2} (6), (7) and (8).}
\QEDA

\section{Verification for the example in Section 2}
\label{appendix:example}

Denote $u(k)=(u^T_{1}(k),\ldots,u^T_{N}(k))^T,$ $\nu(k)=(\nu_{1}(k),\ldots,\nu_{N}(k))^T$. Suppose that $\{\xi(k),k\geq0\}$, $\{u(k),k\geq0\}$, $\{\nu(k),k\geq0\}$ and $\{\mathcal{A}_{\mathcal{G}(k)},k\geq0\}$ are mutually independent.

Firstly, we will verify that \Cref{subgradient} holds. By (\ref{lasso}), we have
\begin{align}
&E\left[\ell_{i}(x;u_{i}(k),p_{i}(k))\right]\cr
=&\frac{1}{2}E\left[\|u^T_{i}(k)x_{0}+\nu_{i}(k)-u^T_{i}(k)x\|^2\right]\cr
=&\frac{1}{2}E\left[\|u^T_{i}(k)(x_{0}-x)+\nu_{i}(k)\|^2\right]\cr
=&\frac{1}{2}E\left[(x_{0}-x)^Tu_{i}(k)u^T_{i}(k)(x_{0}-x)+2\nu_{i}(k)u^T_{i}(k)(x_{0}-x)
+\nu_{i}(k)\nu_{i}(k)\right]\cr
=&\frac{1}{2}\left[(x-x_{0})^TR_{u,i}(x-x_{0})+\sigma_{\nu,i}\right],\notag
\end{align}
{then,} $\nabla E[\ell_{i}(x;\mu_{i}(k))]=R_{u,i}(x-x_{0})$, and the subgradient  of the local risk function $f_{i}(x)$ is given by
\begin{equation}\label{riskd}
d_{f_{i}}(x)=R_{u,i}(x-x_{0})+d_{R_{i}}(x),\ \forall\  d_{R_{i}}(x)\in\partial R_{i}(x).
\end{equation}
By the definition of $d_{R_{i}}(x)$, it is known that $\|d_{R_{i}}(x)\|\leq \sqrt{n}\kappa,\ \forall\ d_{R_{i}}(x)\in\partial R_{i}(x)$. Hence, it can be obtained from \eqref{riskd} that $$\|d_{f_{i}}(x)\|=\|R_{u,i}(x-x_{0})+d_{R_{i}}(x)\|\leq\|R_{u,i}\|\|x\|
+\|R_{u,i}\|\|x_{0}\|+\sqrt{n}\kappa,$$
thus, \Cref{subgradient} holds. The subgradients of the local cost functions are required to be bounded in \cite{Srivastava,NedicA,NedicA2,LiuS} which can not cover the case above, while our assumption covers both $L_{2}$-regularization and $L_{1}$-regularization.

Secondly, we will verify that \Cref{noise1} and \Cref{noise2} hold.
For \eqref{lasso}, the subgradients of local risk functions are measured with noises, i.e.
\ban
\tilde{d}_{f_{i}}(x_{i}(k))=d_{f_{i}}(x_{i}(k))+\zeta_{i}(k),
\ean
where \begin{equation}\label{observenoise}
\zeta_{i}(k)=(u_{i}(k)u^T_{i}(k)-R_{u,i})(x_{i}(k)-x_{0})-u_{i}(k)\nu_{i}(k)
\end{equation}
is the subgradient measurement noise of the $i$th optimizer.

Let $\mathcal{F}(k)=\sigma\{\xi_{ji}(t),\ u_{i}(t),\ \nu_{i}(t),\ \mathcal{A}_{\mathcal{G}(t)},\ 0\leq t\leq k,\ 1\leq i,j\leq N\},\ k\geq0,\ \mathcal{F}(-1)=\{{\O},\Omega\}$. It can be derived from the algorithm \eqref{algorithm}-\eqref{algorithm2} that $(x_{i}(k)-x_{0})\in\mathcal{F}(k-1)\subseteq\mathcal{F}(k),i=1,\ldots,N$, and by \eqref{observenoise}, we obtain $\zeta_{i}(k)\in\mathcal{F}(k)$, so $\{\zeta(k),\mathcal{F}(k),k\geq0\}$ is an adapted process. Note that $\{u_{i}(k),k\geq0\}$ is i.i.d., $\{u_{i}(k),k\geq0\}$, $\{\xi_{ji}(k),k\geq0\}$, $\{\mathcal{A}_{\mathcal{G}(k)},k\geq0\}$ and $\{\nu_{i}(k),k\geq0\}$ are mutually independent. Then, $\sigma\{u_{i}(k)\}$ and $\mathcal{F}(k-1)$ are mutually independent. Similarly, $\sigma\{\nu_{i}(k)\}$ and $\mathcal{F}(k-1)$ are also mutually independent. Hence, from \eqref{observenoise}, we have
\begin{align}
&E[\zeta_{i}(k)|\mathcal{F}(k-1)]\notag\\
=&E[(u_{i}(k)u^T_{i}(k)-R_{u,i})(x_{i}(k)-x_{0})-u_{i}(k)\nu_{i}(k)|\mathcal{F}(k-1)]\notag\\
=&E[(u_{i}(k)u^T_{i}(k)-R_{u,i})|\mathcal{F}(k-1)](x_{i}(k)-x_{0})-E[u_{i}(k)
\nu_{i}(k)|\mathcal{F}(k-1)]\notag\\
=&(E[u_{i}(k)u^T_{i}(k)]-R_{u,i})(x_{i}(k)-x_{0})-E[u_{i}(k)]E[\nu_{i}(k)]\notag\\
=&0\ \ \mbox{a.s.},\ \forall\ k\geq0,\ i=1,\ldots,N.\notag
\end{align}
Thus, $\{\zeta(k),\mathcal{F}(k),k\geq0\}$ is a martingale difference sequence. By \eqref{observenoise}, we have
\begin{align}
\label{ex1}
&E\left[\zeta^T_{i}(k)\zeta_{i}(k)|\mathcal{F}(k-1)\right]\cr
=&E\Big[\left(u_{i}(k)u^T_{i}(k)-R_{u,i})(x_{i}(k)-x_{0})-u_{i}(k)\nu_{i}(k)\right)^T\cr
&\times\Big((u_{i}(k)u^T_{i}(k)-R_{u,i})(x_{i}(k)-x_{0})-u_{i}(k)\nu_{i}(k)\Big)\Big
|\mathcal{F}(k-1)\Big]\cr
=&E\Big[(x_{i}(k)-x_{0})^T(u_{i}(k)u^T_{i}(k)-R_{u,i})^T(u_{i}(k)u^T_{i}(k)-R_{u,i})(x_{i}(k)-x_{0})\cr
&-2\nu_{i}(k)u^T_{i}(k)(u_{i}(k)u^T_{i}(k)-R_{u,i})(x_{i}(k)-x_{0})+(\nu_{i}(k))^2u^T_{i}(k)u_{i}(k)\Big|\mathcal{F}(k-1)\Big]\ \mbox{a.s.}
\end{align}
Noting that $\sigma(u_{i}(k))$ and $\mathcal{F}(k-1)$ are mutually independent, by $x_{i}(k)\in\mathcal{F}(k-1)$, we have
\begin{align}
\label{ex2}
&E\left[(x_{i}(k)-x_{0})^T(u_{i}(k)u^T_{i}(k)-R_{u,i})^T(u_{i}(k)u^T_{i}(k)-R_{u,i})(x_{i}(k)-x_{0})
\Big|\mathcal{F}(k-1)\right]\cr
=&(x_{i}(k)-x_{0})^TE\Big[(u_{i}(k)u^T_{i}(k)-R_{u,i})^T(u_{i}(k)u^T_{i}(k)-R_{u,i})
\Big|\mathcal{F}(k-1)\Big](x_{i}(k)-x_{0})\cr
=&(x_{i}(k)-x_{0})^TE\Big[(u_{i}(k)u^T_{i}(k)-R_{u,i})^T(u_{i}(k)u^T_{i}(k)-R_{u,i})\Big](x_{i}(k)-x_{0})\ \mbox{a.s.}
\end{align}
Noting that $u_{i}(k)$ and $\nu_{i}(k)$ are mutually independent, by $x_{i}(k)\in\mathcal{F}(k-1)$ and $E[\nu_{i}(k)]=0$, we have
\begin{align}
\label{ex3}
&E\left[-2\nu_{i}(k)u^T_{i}(k)(u_{i}(k)u^T_{i}(k)-R_{u,i})(x_{i}(k)-x_{0})
\Big|\mathcal{F}(k-1)\right]\cr
=&-2E\left[\nu_{i}(k)\Big|\mathcal{F}(k-1)\right]E\left[u^T_{i}(k)(u_{i}(k)u^T_{i}(k)-R_{u,i})
\Big|\mathcal{F}(k-1)\right](x_{i}(k)-x_{0})\cr
=&-2E\Big[\nu_{i}(k)\Big]E\left[u^T_{i}(k)(u_{i}(k)
u^T_{i}(k)-R_{u,i})\right](x_{i}(k)-x_{0})
=0\ \mbox{a.s.}
\end{align}
It follows from the definitions of $u_{i}(k)$ and $\nu_{i}(k)$ that
\begin{align}
\label{ex4}
E\left[(\nu_{i}(k))^2u^T_{i}(k)u_{i}(k)\Big|\mathcal{F}(k-1)\right]
=&E\left[(\nu_{i}(k))^2u^T_{i}(k)u_{i}(k)\right]\cr
=&E\left[(\nu_{i}(k))^2\right]E\left[u^T_{i}(k)u_{i}(k)\right]\cr
=&\sigma^2_{i,\nu}\mathrm{Tr}(R_{u,i})\ \mbox{a.s.}
\end{align}
Substituting \eqref{ex2}-\eqref{ex4} into \eqref{ex1} gives
\begin{align}
&E[\zeta^T_{i}(k)\zeta_{i}(k)|\mathcal{F}(k-1)]\notag\\
=&(x_{i}(k)-x_{0})^TE\Big[(u_{i}(k)u^T_{i}(k)-R_{u,i})^T(u_{i}(k)u^T_{i}(k)-R_{u,i})\Big](x_{i}(k)-x_{0})
+\sigma^2_{i,\nu}\mathrm{Tr}(R_{u,i})\notag\\
\leq&2E\Big[\|u_{i}(k)u^T_{i}(k)-R_{u,i}\|^2\Big]\|x_{i}(k)\|^2
+2E\Big[\|u_{i}(k)u^T_{i}(k)-R_{u,i}\|^2\Big]\|x_{0}\|^2+\sigma^2_{i,\nu}|\mathrm{Tr}(R_{u,i})| \ \mbox{a.s.}\notag
\end{align}
Denote
$\sigma_{\zeta}=\max\limits_{1\leq i\leq N}\Big\{2E\Big[\|u_{i}(k)u^T_{i}(k)-R_{u,i}\|^2\Big]\Big\}$
and
$C_{\zeta}=N\max\limits_{1\leq i\leq N}\Big\{2E\Big[\|u_{i}(k)u^T_{i}(k)-R_{u,i}\|^2\Big]\|x_{0}\|^2
+\sigma^2_{i,\nu}|\mathrm{Tr}(R_{u,i})|\Big\}$. Then we have
\begin{equation}\label{wkwk}
E\Big[\zeta^T(k)\zeta(k)\Big|\mathcal{F}(k-1)\Big]
=\sum_{i=1}^NE\Big[\zeta^T_{i}(k)\zeta_{i}(k)\Big|\mathcal{F}(k-1)\Big]\leq\sigma_{\zeta}\|X(k)\|^2+C_{\zeta}\ \mbox{a.s.}
\end{equation}

Noting that $\{\xi(k),k\geq0\}$, $\{u(k),k\geq0\}$, $\{\nu(k),k\geq0\}$ and $\{\mathcal{A}_{\mathcal{G}(k)},k\geq0\}$ are mutually independent,  by Lemma A.1 in \cite{LiT},
we obtain that $\sigma\{\xi(k),\xi(k+1),\ldots\}$ and $\sigma\{\mathcal{A}_{\mathcal{G}(k)},\mathcal{A}_{\mathcal{G}(k+1)},\ldots\}$  are conditionally independent given $\mathcal{F}(k-1)$, $\forall\ k\geq0$, which means that $\sigma\{\mathcal{A}_{\mathcal{G}(k)},\mathcal{A}_{\mathcal{G}(k+1)},\ldots\}$ and $\sigma\{\xi(k)\}$ are conditionally independent given $\mathcal{F}(k-1)$, i.e. $\{\xi(k), k\geq0\}$ satisfies \Cref{noise1}.

By \eqref{observenoise}, we get $\sigma\{\zeta(k)\}\subseteq\sigma\{u_{i}(k),\nu_{i}(k),x_{i}(k),1\leq i\leq N\}$. {Then}, by $\sigma\{x_{i}(k),1\leq i\leq N\}\subseteq\mathcal{F}(k-1)$, we have $\sigma\{\zeta(k)\}\subseteq\sigma\big\{\sigma\{u_{i}(k),\nu_{i}(k),1\leq i\leq N\}\cup\mathcal{F}(k-1)\big\}$. Therefore,
\begin{equation}\label{227}
\begin{array}{rcl}
\sigma\Big\{\sigma\{\zeta(k)\}\cup\mathcal{F}(k-1)\Big\}
\subseteq\sigma\Big\{\sigma\{u_{i}(k),\nu_{i}(k),1\leq i\leq N\}\cup\mathcal{F}(k-1)\Big\}.
\end{array}
\end{equation}
Noting that $\{\xi(k),k\geq0\}$, $\{u(k),k\geq0\}$, $\{\nu(k),k\geq0\}$ and $\{\mathcal{A}_{\mathcal{G}(k)},k\geq0\}$ are mutually independent, $\{u(k),k\geq0\}$ and $\{\nu(k),k\geq0\}$ are i.i.d., we have $\sigma\{u_{i}(k),\nu_{i}(k),1\leq i\leq N\}$ is independent of  $\sigma\big\{\sigma\{\mathcal{A}_{\mathcal{G}(k)},$ $\mathcal{A}_{\mathcal{G}(k+1)},\ldots\}\cup\mathcal{F}(k-1)\big\}$. By Corollary 7.3.2 in \cite{YSChow},
we have $\sigma\{\mathcal{A}_{\mathcal{G}(k)},\mathcal{A}_{\mathcal{G}(k+1)},\ldots\}$ and $\sigma\{u_{i}(k),\nu_{i}(k),\\ 1\leq i\leq N\}$  are conditionally independent given $\mathcal{F}(k-1)$. Then, by Theorem 7.3.1 in \cite{YSChow}
we obtain that for all $A\in\sigma\{\mathcal{A}_{\mathcal{G}(k)},\mathcal{A}_{\mathcal{G}(k+1)},\ldots\}$, \begin{equation}\label{228}
\begin{array}{rcl}
P\Big\{A\Big|\sigma\Big\{\sigma\{u_{i}(k),\nu_{i}(k),1\leq i\leq N\}\cup\mathcal{F}(k-1)\Big\}\Big\}=P\{A|\mathcal{F}(k-1)\}.
\end{array}
\end{equation}
By \eqref{227} and \eqref{228}, we have
\begin{align}
&P\Big\{A\Big|\sigma\Big\{\sigma\{\zeta(k)\}\cup\mathcal{F}(k-1)\Big\}\Big\}\cr
=&E\Big[\mathbf{1}_{A}\Big|\sigma\Big\{\sigma\{\zeta(k)\}\cup\mathcal{F}(k-1)\Big\}\Big]\cr
=&E\bigg[E\Big[\mathbf{1}_{A}\Big|\sigma\Big\{\sigma\{u_{i}(k),\nu_{i}(k),1\leq i\leq N\}\cup\mathcal{F}(k-1)\Big\}\Big]\Big|\sigma\Big\{\sigma\{\zeta(k)\}\cup\mathcal{F}(k-1)\Big\}\bigg]\cr
=&E\bigg[E\Big[\mathbf{1}_{A}\Big|\mathcal{F}(k-1)\Big]\Big|\sigma\Big\{\sigma\{\zeta(k)\}\cup\mathcal{F}(k-1)\Big\}\bigg]\cr
=&P\Big\{A\Big|\mathcal{F}(k-1)\Big\}.\notag
\end{align}
Furthermore, by Theorem 7.3.1 in \cite{YSChow},
we obtain that $\sigma\{\zeta(k)\}$
and $\sigma\{\mathcal{A}_{\mathcal{G}(k)},\mathcal{A}_{\mathcal{G}(k+1)},\ldots\}$ are conditionally independent given $\mathcal{F}(k-1)$, which together with \eqref{wkwk} gives that $\{\zeta(k),k\geq0\}$ satisfies \Cref{noise2}.

\section{Verification of Remark 2.2}
\label{appendix:step}
\begin{proof}
For the step sizes $\alpha(k)$, $c(k)$ defined in  Remark 2.2, it is easily verified that Condition (C1) holds. By \Cref{lemma2.1}, we know that $\lim_{k\rightarrow\infty}(\sum_{t=0}^{k}\alpha(t)-\int_{0}^{k}\alpha(t)dt)$ exists. Hence, there exists a positive constant $\tilde{\alpha}_{2}$, such that
$$\int_{0}^{k}\alpha(t)dt-\tilde{\alpha}_{2}\leq\sum\limits\limits_{t=0}^{k}\alpha(t)\leq\int_{0}^{k}\alpha(t)dt+\tilde{\alpha}_{2}.$$
By \Cref{examplealphack}, we have
$$\int_{0}^{k}\alpha(t)dt=\frac{\alpha_{1}}{1-\tau_{1}}(\ln^{1-\tau_{1}}(k+3)-\ln^{1-\tau_{1}}(3)).$$
Therefore,
$$\frac{\alpha_{1}}{1-\tau_{1}}\left(\ln^{1-\tau_{1}}(k+3)-\ln^{1-\tau_{1}}(3)\right)-
\tilde{\alpha}_{2}\leq \sum\limits\limits_{t=0}^{k}\alpha(t)\leq\frac{\alpha_{1}}
{1-\tau_{1}}\left(\ln^{1-\tau_{1}}(k+3)
-\ln^{1-\tau_{1}}(3)\right)+\tilde{\alpha}_{2},$$
then, for any given positive constant $C$, we have
\begin{equation}\label{C1}
\alpha_{3}\exp\Big(\tilde{\alpha}_{0}\ln^{1-\tau_{1}}(k+3)\Big)\leq \exp\left(C\sum\limits\limits_{t=0}^{k}\alpha(t)\right)
\leq\alpha_{4}\exp\Big(\tilde{\alpha}_{0}\ln^{1-\tau_{1}}(k+3)\Big),
\end{equation}
where $\tilde{\alpha}_{0}=\frac{C\alpha_{1}}{1-\tau_{1}}$,
$\alpha_{3}=\exp\left(-C\left(\frac{\alpha_{1}\ln^{1-\tau_{1}}(3)}
{1-\tau_{1}}+\tilde{\alpha}_{2}\right)\right)$,
$\alpha_{4}=\exp\left(-C\left(\frac{\alpha_{1}\ln^{1-\tau_{1}}(3)}{1-\tau_{1}}
-\tilde{\alpha}_{2}\right)\right)$.

By \Cref{examplealphack}, we have
\begin{equation*}
\displaystyle\lim\limits_{k\rightarrow\infty}\frac{c^2(k)}{\alpha(k)}
=\displaystyle\lim\limits_{k\rightarrow\infty}\frac{\frac{\alpha_{2}^2}{(k+3)^{2\tau_{2}}\ln^{2\tau_{3}}(k+3)}}
{\frac{\alpha_{1}}{(k+3)\ln^{\tau_{1}}(k+3)}}
=\displaystyle\frac{\alpha_{2}^2}{\alpha_{1}}\lim\limits_{k\rightarrow\infty}\frac{\ln^{\tau_{1}-2\tau_{3}}(k+3)}{(k+3)^{2\tau_{2}-1}}
=0.
\end{equation*}
Thus, Condition (C2) holds.

From the left side of \eqref{C1}, we have
$$\exp\left(-C\sum\limits\limits_{t=0}^{k}\alpha(t)\right)\leq\frac{1}{\alpha_{3}}\exp\Big(-\tilde{\alpha}_{0}\ln^{1-\tau_{1}}(k+3)\Big)
=\frac{1}{\alpha_{3}}(k+3)^{-\frac{\tilde{\alpha}_{0}}{\ln^{\tau_{1}}(k+3)}},\ k\geq0,$$
which gives
\begin{equation}\label{C12}
\begin{array}{rcl}
\displaystyle\alpha(k)\exp\left(-C\sum\limits\limits_{t=0}^{k}\alpha(t)\right)\leq\frac{\alpha_{1}}{\alpha_{3}(k+3)^{1+\frac{\tilde{\alpha}_{0}}
{\ln^{\tau_{1}}(k+3)}}\ln^{\tau_{1}}(k+3)}, \ k\geq0.
\end{array}
\end{equation}
Let
$$f(t)=\frac{\alpha_{1}}{\alpha_{3}(\exp(t))^{1+\frac{\tilde{\alpha}_{0}}
{t^{\tau_{1}}}}t^{\tau_{1}}}=\frac{\alpha_{1}}{\alpha_{3}
\exp(t+\tilde{\alpha}_{0}t^{1-\tau_{1}})t^{\tau_{1}}}.$$
By \eqref{C12}, we have
\begin{equation}\label{C13}
\begin{array}{rcl}
\alpha(k)\exp\left(-C\sum\limits\limits_{t=0}^{k}\alpha(t)\right)\leq f(\ln(k+3)),  \ k\geq0.
\end{array}
\end{equation}
Let $A_{k}=\int_{1}^{k}f(t)dt$, we have
\begin{align}
A_{k}=&\displaystyle\frac{\alpha_1}{\alpha_3(1-\tau_{1})}
\int_{1}^{k^{1-\tau_{1}}}
\frac{1}{\exp(s^{\frac{1}{1-\tau_{1}}}+\tilde{\alpha}_{0}s)}ds\notag\\
\leq&\displaystyle\frac{\alpha_1}{\alpha_3(1-\tau_{1})}\int_{1}^{k^{1-
\tau_{1}}}\frac{1}{\exp(\tilde{\alpha}_{0}s)}ds\notag\\
=&\displaystyle-\frac{\alpha_1}{\tilde{\alpha}_{0}\alpha_{3}(1
-\tau_{1})}\exp(-\tilde{\alpha}_{0}s)\Big|_{1}^{k^{1-\tau_{1}}}\notag\\
=&\displaystyle\frac{1}{\alpha_{3}C}(\exp(-\tilde{\alpha}_{0})-
\exp(-\tilde{\alpha}_{0}k^{1-\tau_{1}}))\notag\\
\leq&\displaystyle\frac{\exp(-\tilde{\alpha}_{0})}{\alpha_{3}C},\notag
\end{align}
i.e. $A_{k}$ is monotonically increasing and upper bounded, so it converges. Besides, noting that $f(t)$ is monotonically decreasing, by \Cref{lemma2.1}, we know that $\sum_{t=1}^{\infty}f(t)<\infty$, which together with \eqref{C13} leads to
$$\sum\limits\limits_{k=0}^{\infty}\alpha(k)\exp\left(-C\sum_{t=0}^{k}\alpha(t)\right)\leq\sum\limits\limits_{k=0}^{\infty}f(\ln(k+3))\leq\sum_{t=1}^{\infty}f(t)<\infty.$$
Therefore,  Condition (C3) holds.

For any given $0<\epsilon<1-\tau_{2}$, there exists a positive integer $k_{0}$, such that $\frac{\tilde{\alpha}_{0}}{\ln^{\tau_{1}}(k+3)}<\epsilon$, $k\geq k_{0}$. Thus, $$\exp\Big(\tilde{\alpha}_{0}\ln^{1-\tau_{1}}(k+3)\Big)=(k+3)^{\frac{\tilde{\alpha}_{0}}{\ln^{\tau_{1}}(k+3)}}
\leq(k+3)^{\epsilon},\ k\geq k_{0}.$$
By the right side of \eqref{C1} and the definition of $c(k)$, we have
\begin{equation*}
\begin{array}{rcl}
\displaystyle\frac{\alpha(k)\exp(C\sum_{t=0}^{k}\alpha(t))}{c(k)}\leq\frac{\frac{\alpha_{1}\alpha_{4}
\exp\big(\tilde{\alpha}_{0}\ln^{1-\tau_{1}}(k+3)\big)}{(k+3)\ln^{\tau_{1}}(k+3)}}
{\frac{\alpha_{2}}{(k+3)^{\tau_{2}}\ln^{\tau_{3}}(k+3)}}
\leq\frac{\alpha_{1}\alpha_{4}\ln^{\tau_{3}-\tau_{1}}(k+3)}{\alpha_{2}(k+3)^{1-\tau_{2}-\epsilon}},\ k\geq k_{0}.
\end{array}
\end{equation*}
Since $1-\tau_{2}-\epsilon>0$, we have
\begin{equation*}
\begin{array}{rcl}
\lim\limits_{k\rightarrow\infty}\dfrac{\alpha_{1}\alpha_{4}\ln^{\tau_{3}-\tau_{1}}(k+3)}
{\alpha_{2}(k+3)^{1-\tau_{2}-\epsilon}}=0.
\end{array}
\end{equation*}
Noting that $\frac{\alpha(k)\exp(C\sum_{t=0}^{k}\alpha(t))}{c(k)}>0$, we have
\begin{equation*}
\begin{array}{rcl}
\lim\limits_{k\rightarrow\infty}\dfrac{\alpha(k)\exp(C\sum_{t=0}^{k}\alpha(t))}{c(k)}=0,
\end{array}
\end{equation*}
which implies Condition (C4).

By the definition of $\alpha(k)$ in \Cref{examplealphack} and mean value theorem of integrals, we get
\begin{align}
\vspace{2mm}\displaystyle\int_{k}^{k+1}d\Big(\frac{1}{\alpha(t)}\Big)
\vspace{2mm}=&\displaystyle\int_{k}^{k+1}d\left(\frac{(t+3)\ln^{\tau_{1}}(t+3)}
{\alpha_{1}}\right)\cr
\vspace{2mm}=&\displaystyle\frac{1}{\alpha_{1}}\int_{k}^{k+1}(\ln^{\tau_{1}}(t+3)+\tau_{1}\ln^{\tau_{1}-1}(t+3))dt\cr
=&\displaystyle\frac{1}{\alpha_{1}}(\ln^{\tau_{1}}(s+3)+\tau_{1}\ln^{\tau_{1}-1}(s+3)),\  s\in[k,k+1].\label{C3}
\end{align}
Then, by  the monotone property of logarithmic functions, we have
\begin{equation}\label{C10}
\begin{array}{rcl}
\displaystyle\frac{1}{\alpha(k+1)}-\frac{1}{\alpha(k)}=\int_{k}^{k+1}d\Big(\frac{1}{\alpha(t)}\Big)\geq
\frac{1}{\alpha_{1}}\ln^{\tau_{1}}(k+3).
\end{array}
\end{equation}
For any given positive constant $C$, we have
\begin{align}\label{C9}
&\hspace{-1.2cm}\alpha(k)\exp\left(C\sum\limits\limits_{t=0}^{k}\alpha(t)\right)
-\alpha(k+1)\exp\left(C\sum\limits\limits_{t=0}^{k+1}\alpha(t)\right)\cr
&\hspace{-1.6cm}=\exp\left(C\sum\limits\limits_{t=0}^{k}\alpha(t)\right)(\alpha(k)-\alpha(k+1)\exp(C\alpha(k+1)))\cr
&\hspace{-1.6cm}=\exp\left(C\sum\limits\limits_{t=0}^{k}
\alpha(t)\right)\left(\alpha(k)-\alpha(k+1)\left(1+C\alpha(k+1)
+\frac{C^2\alpha^2(k+1)}{2}+o\left(C^2\alpha^2(k+1)\right)\right)\right)\cr
&\hspace{-1.6cm}=\alpha(k)\alpha(k+1)\exp\left(C\sum\limits\limits_{t=0}^{k}
\alpha(t)\right)\left(\frac{1}{\alpha(k+1)}-\frac{1}{\alpha(k)}
-\frac{C\alpha(k+1)}{\alpha(k)}
-\frac{C^2\alpha^2(k+1)}{2\alpha(k)}-o\left(\frac{C^2\alpha^2(k+1)}{\alpha(k)}\right)\right).
\end{align}
From \eqref{C10} and \eqref{C9}, it is known that there exists a positive integer $k_{2}$, such that
\begin{equation}\label{C11}
\begin{array}{rcl}
\alpha(k)\exp\left(C\sum\limits\limits_{t=0}^{k}\alpha(t)\right)-
\alpha(k+1)\exp\left(C\sum\limits\limits_{t=0}^{k+1}\alpha(t)\right)\geq0,\ k\geq k_{2},
\end{array}
\end{equation}
which means that the sequence  $\{\alpha(k)\exp(C\sum_{t=0}^k\alpha(t)),k\geq0\}$  decreases monotonically for sufficiently large $k$.

From the monotone property of logarithmic functions, $\tau_{1}-1<0$ and \eqref{C3}, we have
\begin{equation}\label{C4}
\begin{array}{rcl}
\displaystyle\int_{k}^{k+1}d\Big(\frac{1}{\alpha(t)}\Big)
\leq\frac{1}{\alpha_{1}}(\ln^{\tau_{1}}(k+4)+\tau_{1}\ln^{\tau_{1}-1}(k+3)).
\end{array}
\end{equation}
From the left side of \eqref{C1}, we get
\begin{equation}\label{C5}
\begin{array}{rcl}
\displaystyle\frac{1}{\exp(C\sum_{t=0}^{k}\alpha(t))}\leq\frac{1}{\alpha_{3}\exp(\tilde{\alpha}_{0}\ln^{1-\tau_{1}}(k+3))}.
\end{array}
\end{equation}
From $\alpha(k)>0$, we have $\exp(C\sum_{t=0}^{k}\alpha(t))\leq\exp(C\sum_{t=0}^{k+1}\alpha(t))$, therefore,
\begin{align}
\label{C2}
&\frac{\alpha(k)\exp(C\sum_{t=0}^{k}\alpha(t))-\alpha(k+1)\exp(C\sum_{t=0}^{k+1}\alpha(t))}
{\alpha^2(k)\exp(2C\sum_{t=0}^{k}\alpha(t))}\cr
\vspace{2mm}\leq&\displaystyle\frac{\alpha(k)\exp(C\sum_{t=0}^{k}\alpha(t))-\alpha(k+1)\exp(C\sum_{t=0}^{k}\alpha(t))}
{\alpha(k)\alpha(k+1)\exp(2C\sum_{t=0}^{k}\alpha(t))}\cr
\vspace{2mm}=&\displaystyle\frac{\alpha(k)-\alpha(k+1)}{\alpha(k)\alpha(k+1)\exp(C\sum_{t=0}^{k}\alpha(t))}\cr
\vspace{2mm}=&\displaystyle\Big(\frac{1}{\alpha(k+1)}-\dfrac{1}{\alpha(k)}\Big)\frac{1}{\exp(C\sum_{t=0}^{k}\alpha(t))}\cr
=&\frac{1}{\exp(C\sum_{t=0}^{k}\alpha(t))}\int_{k}^{k+1}d\Big(\dfrac{1}{\alpha(t)}\Big).
\end{align}
By \eqref{C4}-\eqref{C2}, we get
\begin{align}
\label{C6}
&\displaystyle\frac{\alpha(k)\exp(C\sum_{t=0}^{k}\alpha(t))-\alpha(k+1)
\exp\left(C\sum_{t=0}^{k+1}\alpha(t)\right)}
{\alpha^2(k)\exp\left(2C\sum_{t=0}^{k}\alpha(t)\right)}\cr
\leq&\displaystyle\frac{1}{\alpha_{1}}\left(\frac{\ln^{\tau_{1}}(k+4)}{\alpha_{3}\exp(\tilde{\alpha}_{0}\ln^{1-\tau_{1}}(k+3))}
+\frac{\tau_{1}\ln^{\tau_{1}-1}(k+3)}{\alpha_{3}\exp(\tilde{\alpha}_{0}\ln^{1-\tau_{1}}(k+3))}\right).
\end{align}
For the first term in the bracket on the right side of  \eqref{C6}, we get
\begin{align}
\label{C7}
\displaystyle\lim\limits_{k\rightarrow\infty}\frac{\ln^{\tau_{1}}
(k+4)}{\alpha_{3}\exp\left(\tilde{\alpha}_{0}\ln^{1-\tau_{1}}(k+3)\right)}
=&\displaystyle\lim\limits_{k\rightarrow\infty}\frac{\ln^{\tau_{1}}(k+3)}
{\alpha_{3}\exp\left(\tilde{\alpha}_{0}\ln^{1-\tau_{1}}(k+3)\right)}
\frac{\ln^{\tau_{1}}(k+4)}{\ln^{\tau_{1}}(k+3)}\cr
=&\displaystyle\lim\limits_{t\rightarrow\infty}\frac{t^{\frac{\tau_{1}}{1-\tau_{1}}}}{\alpha_{3}\exp(\tilde{\alpha}_{0}t)}
\lim\limits_{k\rightarrow\infty}\frac{\ln^{\tau_{1}}(k+4)}{\ln^{\tau_{1}}(k+3)}\cr
=&0.
\end{align}
For the second term in the bracket on the right side of  \eqref{C6}, we get
\begin{equation}\label{C8}
\begin{array}{rcl}
\displaystyle\lim\limits_{k\rightarrow\infty}\frac{\tau_{1}\ln^{\tau_{1}-1}(k+3)}
{\alpha_{3}\exp\left(\tilde{\alpha}_{0}\ln^{1-\tau_{1}}(k+3)\right)}
=\lim\limits_{k\rightarrow\infty}\frac{\tau_{1}}
{\alpha_{3}\ln^{1-\tau_{1}}(k+3)\exp\left(\tilde{\alpha}_{0}
\ln^{1-\tau_{1}}(k+3)\right)}
=0.
\end{array}
\end{equation}
Therefore, from \eqref{C11} and \eqref{C6}-\eqref{C8}, we obtain
\begin{equation*}
\lim\limits_{k\rightarrow\infty}\frac{\alpha(k)\exp(C\sum_{t=0}^{k}\alpha(t))-\alpha(k+1)\exp(C\sum_{t=0}^{k+1}\alpha(t))}
{\alpha^2(k)\exp(2C\sum_{t=0}^{k}\alpha(t))}=0.
\end{equation*}
Thus, Condition (C5) holds.
\end{proof}

\end{document}